\documentclass[reqno]{amsart}

\usepackage{amssymb}
\usepackage{amsthm,amsmath}
\usepackage{amsfonts}
\usepackage{graphicx}
\usepackage{bm}
\usepackage[utf8]{inputenc}
\numberwithin{equation}{section}
\newtheorem{theo}{Theorem}[section]
\newtheorem{lem}[theo]{Lemma}
\newtheorem{prop}[theo]{Proposition}

\theoremstyle{definition}

\theoremstyle{remark}
\newtheorem{rem}[theo]{Remark}

\usepackage{caption}
\usepackage{color}
\definecolor{applegreen}{rgb}{0.55, 0.71, 0.0}
\definecolor{yellowcut}{RGB}{0, 200, 0}
\definecolor{orangecross}{RGB}{248, 48, 8}
\definecolor{violetcross}{RGB}{178, 8, 248}
\definecolor{ao(english)}{rgb}{0.0, 0.5, 0.0}
\definecolor{green(html/cssgreen)}{rgb}{0.0, 0.5, 0.0}
\definecolor{lincolngreen}{rgb}{0.11, 0.35, 0.02}

\usepackage{tikz}

\newcommand\Corr[2]{\langle#2 \rangle_{#1}}

\def\arg{\mathop{\mathrm{arg}}\nolimits}

\def\det{\mathop{\mathrm{det}}\nolimits}

\def\Ker{\mathop{\mathrm{Ker}}\nolimits}

\def\res{\mathop{\mathrm{res}}\limits}

\def\supp{\mathop{\mathrm{supp}}\nolimits}

\def\Xsym{X^\mathrm{sym}_{[\mathbf v,\mathbf u]}}
\def\Xanti{X^\mathrm{anti}_{[\mathbf v,\mathbf u]}}

\begin{document}

\title[Zig-zag layered Ising model and orthogonal polynomials]{Magnetization in the zig-zag layered Ising model and orthogonal polynomials}
\author[Dmitry Chelkak]{Dmitry Chelkak$^\mathrm{a,b}$}
\author[Cl\'ement Hongler]{Cl\'ement Hongler$^\mathrm{c}$}
\author[R\'emy Mahfouf]{R\'emy Mahfouf$^\mathrm{a}$}

\thanks{\textsc{${}^\mathrm{A}$ ENS--MHI Chair, D\'epartement de math\'ematiques et applications, \'Ecole Normale Sup\'erieure, CNRS, PSL University, 45 rue d'Ulm, 75005 Paris, France.}}

\thanks{\textsc{${}^\mathrm{B}$ On leave from St.~Petersburg Dept. of Steklov Mathematical Institute RAS, Fontanka 27, 191023 St.~Petersburg, Russia.}}

\thanks{\textsc{${}^\mathrm{C}$ Chair of Statistical Field Theory, MATHAA Institute, \'Ecole Polytechnique F\'ed\'erale de Lausanne, Station 8, 1015 Lausanne, Switzerland.}}

\thanks{\emph{E-mail:} \texttt{dmitry.chelkak@ens.fr}, \texttt{clement.hongler@epfl.ch}, \texttt{remy.mahfouf@ens.fr}}

\begin{abstract} We discuss the magnetization~$M_m$ in the $m$-th column of the zig-zag layered 2D Ising model on a half-plane using Kadanoff--Ceva fermions and orthogonal polynomials techniques. Our main result gives an explicit representation of~$M_m$ via~$m\times m$ Hankel determinants constructed from the spectral measure of a certain Jacobi matrix which encodes the interaction parameters between the columns. We also illustrate our approach by giving short proofs of the classical Kaufman--Onsager--Yang and McCoy--Wu theorems in the homogeneous setup and expressing~$M_m$ as a Toeplitz+Hankel determinant for the homogeneous sub-critical model in presence of a boundary magnetic field.
\end{abstract}



\maketitle

\section{Introduction} The planar Ising (or Lenz--Ising) model, introduced by Lenz almost a century ago, has an extremely rich history which is impossible to overview in a short introduction, instead we refer the interested reader to the monographs~\cite{mccoy-wu-book,baxter-book,palmer-book,smirnov-duminil-lectures} as well as the papers~\cite{niss-I,niss-II,niss-III,mccoy-maillard-12,chelkak-cimasoni-kassel} and references therein for more information on various facets of this history. From the `classical analysis' viewpoint, one of the particularly  remarkable aspects is a fruitful interplay between the explicit computations for the planar Ising model and the theory of {Toeplitz determinants}. This interplay originated in the groundbreaking work of Kaufman and Onsager in late 1940s (see~\cite{baxter-onsager-I,baxter-onsager-II}) and, in particular, lead Szeg\"o to the strong form of his famous theorem on asymptotics of Toeplitz determinants; we refer the interested reader to the recent survey~\cite{deift-its-krasovsky} due to Deift, Its and Krasovsky for more information on the developments of this link since then.

{Besides having representations via (Toeplitz or more complicated) determinants, spin correlations in the planar Ising model are known to satisfy quadratic identities~\cite{Perk-81dubna,mccoy-perk-wu-quadratic} arising when one changes the position of a spin variable by one lattice step. Though a direct asymptotic analysis of these determinants can be easily performed only for `diagonal' or `horizontal' correlations (e.g., see~\cite[Chapters~VIII and~XII]{mccoy-wu-book}), one can then use the aforementioned quadratic identities to analyze asymptotics near these special directions; e.g., see a discussion in~\cite[Section~2]{perk-au-yang-09}. In their turn, the quadratic identities for spin correlations are deeply related to the theory of (discrete) isomonodromic deformations and $\tau$-functions obtained thereof~\cite{SMJ-77}. This also leads to the famous appearance of discrete Painlev\'e equations in the planar Ising model~\cite{SMJ-80b} and in its massive scaling limit~\cite{Wu-McCoy-Tracy-Barouch}. We refer the interested reader to the monograph~\cite{palmer-book} for an account of these developments and only mention that this deep interplay of several topics still remains an active research subject in analysis; e.g., see~\cite{basor-et-al-2015,witte-07} and references therein.}

It is nevertheless worth noting that the research direction {outlined above} mostly originated in questions related to the homogeneous model in the infinite-volume limit -- a well-understood case from the {statistical physics} perspective. At the same time, it seems that the much richer setup of the \emph{layered model} -- first considered by McCoy--Wu and Au-Yang--McCoy in~\cite{McCoy-Wu-random-I,McCoy-Wu-random-II,McCoy-random-III,AuYang-McCoy-layered-I,AuYang-McCoy-layered-II}, see also~\cite[Sections~3.1,3.2]{McCoy-1999} and~\cite{pelizolla1997} for historical comments -- did not attract much attention of mathematicians. Unfortunately, \emph{tour de force} computations summarized in the monograph~\cite{mccoy-wu-book}, are nowadays often considered (at least, in several mathematical sub-communities interested in 2D statistical {physics}) as being too technically involved to develop their analysis further. Certainly, this is an abnormal situation and by writing this paper we hope to bring the attention to this `layered' setup, targeting not only probabilists but also the {spectral theory/orthogonal polynomials} community. In the mathematical physics literature, the interest to the layered Ising model also reappeared recently; e.g. see~\cite{au-yang-13},~\cite{CGG-19-MWmodel} and references therein.

Our paper should not be considered as a `39999th solution of the Ising model'. On the contrary, the methods we use can be viewed as a simplification of the classical ones in presence of the translation and reflection symmetry in the direction orthogonal to the line connecting spins under consideration. Comparing to~\cite{mccoy-wu-book}, this simplification {(which was first presented in~\cite[Section~3]{chelkak-ecm2016} based upon an early version of this paper)} comes from the fact that we use the Kadanoff--Ceva lattice instead of the Onsager (or Fisher) one and, more importantly, work directly with \emph{orthogonal polynomials} instead of Toeplitz determinants. Though such details are not vital in the homogeneous case, this allows us to perform computations for a general `zig-zag layered' model in a transparent way (see Theorem~\ref{thm:layered}); in the latter case, the polynomials are orthogonal with respect to a certain measure on the segment~$[0,1]$ constructed out of a given sequence of interaction constants.

It is worth mentioning that {the simplification discussed above} manifests itself even in the homogeneous setup since we always deal with \emph{real weights}, the simplest possible framework of the OPUC/OPRL theory. From the perspective of the `free fermion algebra' solution~\cite{schultz-mattis-lieb} of the planar Ising model, our derivations can be viewed as its translation to the language of discrete fermionic observables, see~\cite{hongler-kytola-zahabi} for a discussion of such a correspondence. The latter viewpoint was advertised by Smirnov in his celebrated work on the critical Ising model (e.g., see~\cite{smirnov-duminil-lectures} and references therein). We refer the interested reader to~\cite[Section~3]{chelkak-cimasoni-kassel} for a discussion of equivalences between various combinatorial formalisms used to study the planar Ising model, see also~\cite{mercat-CMP} and~\cite[Section~3.2]{chelkak-smirnov-12}. In this paper we also want to make a link between discrete complex analysis techniques and classical computations more transparent; similar ideas are applied to the quantum 1d Ising model in~\cite{li-mahfouf}.

Before formulating our main result -- Theorem~\ref{thm:layered} -- for the layered Ising model, let us briefly mention the list of questions that we discuss along the way in the homogeneous setup:
\begin{itemize}
\item Kaufman--Onsager--Yang theorem on the spontaneous magnetization below criticality: Theorem~\ref{thm:KOY}, cf.~\cite[Section~X.4]{mccoy-wu-book};
\item McCoy--Wu theorem on the asymptotic behavior of the horizontal spin-spin correlations at criticality: Theorem~\ref{thm:crit-homogen}, cf.~\cite[Section~XI.5]{mccoy-wu-book};
\item the {wetting phase transition in the subcritical model caused by a boundary magnetic field~\cite{frohlich-pfister-87,pfister-velenik-96} (which was interpreted as a hysteresis effect in the earlier work~\cite{mccoy-wu-book}):} we discuss a setup similar to~\cite[Section~XIII]{mccoy-wu-book} in Section~\ref{subsect:wetting} and reduce the problem to the analysis of explicit Toeplitz+Hankel determinants, see Theorem~\ref{thm:wetting};
\item Wu's explicit formula for diagonal spin-spin correlations in the fully homogeneous critical Ising model (see~\cite[Section~XI.4]{mccoy-wu-book}). {This very short computation via Legendre polynomials already appeared in~\cite[Section~3]{chelkak-ecm2016}, we repeat it here to emphasize a direct link with similar formulas for the magnetization in the zig-zag half-plane.} Note that we were unable to find neither Theorem~\ref{thm:Mm-homo-even} nor the identity~\eqref{eq:MM=D-crit} in the literature.
\end{itemize}

We now move on to the layered Ising model in a half-plane. Instead of working in the original framework of Au-Yang, McCoy and Wu {who considered the layered Ising model in a discrete half-plane with \emph{straight} boundary and translation invariant interaction constants,} %~\cite{McCoy-Wu-random-I,McCoy-Wu-random-II,McCoy-random-III,AuYang-McCoy-layered-I,AuYang-McCoy-layered-II},
we slightly simplify the setup by working in the left half-plane of the $\frac{\pi}{4}$-rotated square grid, which we call the \emph{zig-zag half-plane}~$\mathbb H^\diamond$, {and require that the interaction constants assigned to all edges separating each pair of neighboring columns are the same;} see Fig.~\ref{fig:GlobalLayered}. We believe that such a simplification does not change key features of the problem, at the same it allows us to obtain more transparent results in full generality. We are mostly interested in making our main result -- Theorem~\ref{thm:layered} -- easily accessible to the \emph{mathematical} community interested in orthogonal polynomials rather than in discussing the \emph{physics} behind the problem.
%and refer the interested reader to the book~\cite[Sections~XIV, XV]{mccoy-wu-book} and references therein.
It is worth emphasizing that Theorem~\ref{thm:layered} does \emph{not} express~$M_m$ as a Toeplitz determinant. Nevertheless, we believe that the formula~\eqref{eq:Mm-layered-UJH} is amenable for the asymptotic analysis and is of interest from the mathematical perspective.

The (half-)infinite volume limit of the Ising model on~$\mathbb H^\diamond$ is defined as a limit of probability measures on an increasing sequence of finite domains exhausting~$\mathbb H^\diamond$, with~`$+$' boundary conditions at the right-most column~$\mathrm{C}_0$ and at infinity. All interaction parameters between the columns~$\mathrm{C}_{p-1}$ and~$\mathrm{C}_p$ are assumed to be the same and equal to~$x_p=\exp[-2\beta J_p]=\tan\tfrac{1}{2}\theta_p$, where~$\theta_p\in (0,\frac{1}{2}\pi)$ can be viewed as a convenient parametrization of~$\beta J_p$, see Section~\ref{subsect:contours} for more details. Let
\begin{equation}
\label{eq:Mm-def}
M_m=M_m(\theta_1,\theta_2,\dots)\ :=\ \mathbb E^+_{\mathbb H^\diamond}[\,\sigma_{(-2m-\frac{1}{2},0)}\,]
\end{equation}
be the magnetization in the~$(2m)$-th column (the analysis for odd columns can be done similarly). Denote
\begin{equation}
\label{eq:Deven-def}
D_{\mathrm{even}}:= i\left[\begin{array}{cccc}
\cos\theta_1\cos\theta_2 & 0 & 0 & \dots \\
-\sin\theta_2\sin\theta_3 & \cos\theta_3\cos\theta_4 & 0 & \dots \\
0 & -\sin\theta_4\sin\theta_5 & \cos\theta_5\cos\theta_6 & \dots \\
\dots & \dots & \dots &  \dots\end{array}\right]
\end{equation}
and let
\begin{equation}
\label{eq:polar}
D_\mathrm{even}^*\ =\ U_\mathrm{even}{S_\mathrm{even},\qquad S_\mathrm{even}=}(D^{\phantom{*}}_\mathrm{even}D_\mathrm{even}^*)^{1/2}
\end{equation}
be the polar decomposition of the operator~$D_\mathrm{even}^*$, see also Remark~\ref{rem:U-as-Hilbert} for another interpretation of the (partial) isometry~$U_\mathrm{even}$\,. Further, denote $J:=D^{\phantom{*}}_\mathrm{even}D^*_\mathrm{even}$. A~straightforward computation shows that
\begin{equation}\label{eq:J-def}
J=\left[\begin{array}{cccc}b_1 & -a_1 & 0 & \dots\\ -a_1 & b_2 & -a_2 & \dots \\ 0 & -a_2 & b_3 & \dots \\ \dots & \dots & \dots & \dots \end{array}\right]\quad \begin{array}{l} b_k\,=\,\cos^2\theta_{2k-1}\cos^2\theta_{2k}\\
\phantom{b_k\,=\,}+\sin^2\theta_{2k-2}\sin^2\theta_{2k-1}\,,\\[4pt] a_k\,=\, \cos\theta_{2k-1}\cos\theta_{2k}\sin\theta_{2k}\sin\theta_{2k+1}\,,\end{array}
\end{equation}
where~$\theta_0:=0$ and~$b_1=\cos^2\theta_1\cos^2\theta_2$. Let~$\nu_J$ be the spectral measure of~$J$ associated with the first basis vector. It is easy to see that~$0\le J\le 1$ and thus \mbox{$\supp \nu_J\in [0,1]$}. Given a measure~$\mu$ on~$[0,1]$, let
$\textstyle \mathrm{H}_m[\,\mu\,]\ :=\ \det [\,\int_0^1\lambda^{p+q}\mu(d\lambda)\,]_{p,q=0}^{m-1}$
be the~$m$-th Hankel determinant composed from the moments of this measure. Denote by $P_m$ the orthogonal projector on the space of first $m$ coordinates of $\ell^2$.

\begin{theo}\label{thm:layered} For all~$\theta_1,\theta_2,\ldots \in(0,\frac{\pi}{2})$ and~$m\ge 1$, we have
\begin{align}
\label{eq:Mm-layered-UJH}
M_m %=\mathbb E^+_{\mathbb H^\diamond}[\sigma_{(-2m-\frac{1}{2},0)}]
\ &=\ |\det P_mU_\mathrm{even}P_m|
\ =\ \frac{\det P_m J^{1/2}P_m}{\prod_{k=1}^{2m}\cos\theta_k}\ =
\ \frac{\mathrm{H}_m[\lambda^{1/2}\nu_J]}
{(\,\mathrm{H}_m[\nu_J]\cdot\mathrm{H}_m[\lambda\nu_J]\,)^{1/2}}\,,
\end{align}
where~$U_\mathrm{even}$ is the (partial) isometry factor in the polar decomposition~\eqref{eq:polar}, the Jacobi matrix~$J\!=\!D_\mathrm{even}^{\vphantom{*}}D_\mathrm{even}^*$ is given by~\eqref{eq:J-def}, and $\nu_J$ is the spectral measure of~$J$.
\end{theo}

\begin{rem}
\label{rem:J-homo}
Assume that~$\theta_k=\theta$ for all~$k\ge 1$, i.e., that we work with the fully homogeneous model. One can easily see that
\[
\mathrm{supp}\,\nu_J\;=\;[\cos^2(2\theta)\,,1]\ \ \text{if}\ \theta\le \tfrac{\pi}{4}\quad \text{while}\quad \mathrm{supp}\,\nu_J\;=\;\{0\}\cup [\cos^2(2\theta)\,,1]\ \ \text{if}\ \ \theta>\tfrac{\pi}{4}.
\]
In particular, this clearly marks the critical value~$\theta_\mathrm{crit}=\frac{\pi}{4}$ of the interaction parameter. Moreover, in the supercritical regime~$\theta>\theta_\mathrm{crit}$, the existence of an exponentially decaying eigenfunction~$\psi^\circ_k=(\cot\theta)^{2k}$, $\psi^\circ\in \mathrm{Ker} D_\mathrm{even}^*$, directly leads to the exponential decay of the truncated determinants~$|\det P_mU_{\mathrm{even}}P_m|$.
%as~$U_{\mathrm{even}}$ is a \emph{partial} isometry of~$(\mathrm{Ker} D_\mathrm{even}^*)^\perp$ onto~$\ell^2$.
\end{rem}

\begin{rem}
\label{rem:IDS-periodic}
Assume now that $\theta_{k+2n}=\theta_k$ for all~$k\ge 1$ and some~$n\ge 1$. In this case, the criticality condition reads as~$\prod_{k=1}^{2n}\tan\theta_k=1$, see Lemma~\ref{lem:periodic-crit} below. {This condition is equivalent to the fact} that the continuous spectrum of~$J$ begins at~$0$. {Moreover (see Section~\ref{sub:IDS=}), in this setup the integrated density of states of the periodic Jacobi matrix~$J$} behaves like~$C_J\cdot \pi^{-1}\sqrt{\lambda}$ as~$\lambda\to 0$, {where}
\begin{equation}
\label{eq:IDS=}
\textstyle C_J\ =\ \left[\,n^{-2}\sum_{k=1}^n (\psi^\circ_k)^2\cdot \sum_{k=1}^n(a_k\psi^\circ_k\psi^\circ_{k+1})^{-1}\,\right]^{1/2}
\end{equation}
{and}~$\psi^\circ_k$ denotes the periodic vector solving the equation~$J\psi^\circ=0$. In Section~\ref{subsect:s-embeddings} we show that the quantity~\eqref{eq:IDS=} also admits a clear {geometric} interpretation in the context of the so-called \emph{s-embeddings} of planar Ising models, see~\eqref{eq:BS=CJ} and a discussion following that identity.
\end{rem}

It is clear that the spectral properties of the matrix~$J$ (which can be viewed as an effective propagator in the direction orthogonal to the boundary of~$\mathbb H^\diamond$) are directly related to the behavior of the magnetization~$M_m$ as~$m\to \infty$. Nevertheless, we are not aware of asymptotical results for~\eqref{eq:Mm-layered-UJH} in the general case, especially when~$J$ has a \emph{singular continuous spectrum}.
%, which should be a typical situation, e.g., when dealing with \emph{random interaction constants} considered in~\cite[Sections~XIV, XV]{mccoy-wu-book}.
This leads to the following question:
\begin{itemize}
\item to find necessary and sufficient conditions on the measure~$\nu_J$ that imply the~asymptotics (a)~$\liminf_{m\to\infty} M_m =0$ (b) $\limsup_{m\to\infty}M_m=0$ in~\eqref{eq:Mm-layered-UJH}.
\end{itemize}
We believe that an answer to this question should shed more light, in particular, on the \emph{random} layered 2D Ising model. Moreover, it would be very interesting
\begin{itemize}
\item to understand the dynamics of the measure~$\nu_J$ when the inverse temperature~$\beta$ varies from~$\infty$ to~$0$ and hence all \mbox{$\theta_p=2\arctan\exp[-2\beta J_p]$} increase from~$0$ to~$1$ in a coherent way.
\end{itemize}
Classically, this dynamics should lead to the Griffiths--McCoy phase transition for i.i.d. interaction parameters between the columns and also could give rise to less known effects in the dependent case. As already mentioned above, one of the goals of this paper is to bring the attention of the probability and orthogonal polynomials communities to these questions.

The rest of the paper is organized as follows. In Section~\ref{sc:combinatorics} we review the Kadanoff--Ceva formalism of spin-disorder operators in the planar Ising model. In Section~\ref{sc:homogeneous} we illustrate our approach by giving streamlined proofs of two classical results due to Kaufman--Onsager--Yang and McCoy--Wu, respectively: Theorem~\ref{thm:KOY} and Theorem~\ref{thm:crit-homogen}; we believe that this material should help the reader to position this proof into the classical Ising model landscape. We prove our main result -- Theorem~\ref{thm:layered} -- in Section~\ref{sc:layered}. In Section~\ref{sc:geometry} we briefly discuss the geometric interpretation of our results via s-embeddings of planar Ising models, a generalization of isoradial embeddings of the critical Baxter's Z-invariant model introduced in~\cite{chelkak-icm2018,chelkak-semb}.
The appendix is devoted to the explicit analysis of diagonal correlations (Wu's formula) and of the zig-zag half-plane magnetization at criticality via Legendre polynomials.

\subsection*{Acknowledgements} {We are grateful to Yvan Velenik for bringing our attention to the papers~\cite{frohlich-pfister-87,pfister-velenik-96} on the wetting phase transition in the subcritical model, which was mentioned under the name hysteresis effect in the first version of our paper following the interpretation given in~\cite[Section~XIII]{mccoy-wu-book}.} {We also thank Jacques H.H.~Perk for useful comments on the immense literature on the Ising model correlations.} Several parts of this paper were known and reported since 2012/2013 but caused a very limited interest, we are grateful to colleagues who encouraged us to carry this project out.
Dmitry Chelkak would like to thank Alexander Its, Igor Krasovsky, Leonid Parnovski and Alexander Pushnitski for helpful discussions. The research of Dmitry Chelkak and R\'emy {Mahfouf} was partially supported by the ANR-18-CE40-0033 project DIMERS.
Cl\'ement Hongler would like to acknowledge the support of the ERC SG \mbox{CONSTAMIS}, the
NCCR SwissMAP, the Blavatnik Family Foundation and the Latsis Foundation. We also thank Jhih-Huang Li and the anonymous referees for carefully reading earlier versions of this manuscript.

\section{Combinatorics of the planar Ising model}\label{sc:combinatorics}

In order to keep the presentation self-contained, in this section we collect basic definitions and properties of the planar Ising model observables. Below we adopt the notation from~\cite{chelkak-icm2018,chelkak-semb,chelkak-hongler-izyurov-21}, the interested reader is also referred to~\cite{chelkak-cimasoni-kassel} or~\cite{hongler-kytola-viklund} for more details (note however that these papers use slightly different definitions). Even though we discuss the spin-disorder observables in the full generality ($m$ spins and $n$ disorders), below we are interested in the situations $m=n=2$ (Section~\ref{sc:homogeneous} and Appendix) and~$m=1$, $n=2$ (Section~\ref{sc:layered} and Appendix) only.

\subsection{Definition and domain wall representation}
\label{subsect:contours}
Let~$G$ be a finite connected \emph{planar} graph embedded into the plane such that all its edges are straight segments.
We denote by $G^{\bullet}$ the set of its vertices and by $G^{\circ} $ the set of its faces (identified with their centers).
The (ferromagnetic) \emph{nearest-neighbor} Lenz-Ising model on the graph \emph{dual} to~$G$ is a random assignment of spins~$\sigma_u\in\{\pm 1\}$ to the~\emph{faces} $u\in G^\circ$ such that the probability of a spin configuration~$\sigma\!=\!(\sigma_u)$ is proportional to
\[%\begin{equation}\label{eq:P[sigma]-def}
\textstyle \mathbb{P}_G[\,\sigma\,]\propto\exp\,[\,\beta \sum_{u\sim w} J_e\sigma_u\sigma_w\,]\,,\qquad e=(uw)^*,
\]%\end{equation}
where a positive parameter~$\beta=1/kT$ is called the \emph{inverse temperature}, the sum is taken over all pairs of adjacent faces~$u,w$ (equivalently, edges~$e$) of~$G$, and~$J=(J_e)$ is a collection of positive \emph{interaction constants}, indexed by the edges of~$G$. Below we use the following \emph{parametrization} of~$J_e$:
\begin{equation}\label{eq:parametrization-model}
\textstyle x_e=\tan\frac{1}{2}\theta_e:=\exp[-2\beta J_e].
\end{equation}
Note that the quantities~$x_e\in (0,1)$ and~$\theta_e:=2\arctan x_e\in (0,\frac{1}{2}\pi)$ have the same monotonicity as the temperature~$\beta^{-1}$.

We let the spin~$\sigma_{\mathrm{out}}$ of the outermost face of~$G$ be fixed to~$+1$, in other words we impose \emph{`$+$' boundary conditions}. In this case, the~\emph{domain wall representation} (also known as the \emph{low-temperature expansion}) of the Ising model is a~$1$-to-$1$ correspondence between spin configurations and even subgraphs~$P$ of~$G$: given a spin configuration,~$P$ consists of all edges that separate pairs of disaligned spins. One can consider a decomposition (not unique in general) of~$P$ into a collection of \emph{non-intersecting and non-self-intersecting loops}. The above correspondence implies that
\[%\begin{equation}\label{eq:spin-spin-low-temperature}
\textstyle \mathbb{E}_G[\sigma_{u_1}\dots\sigma_{u_m}]\;=\; \mathcal{Z}_G^{-1}\sum_{P\in\mathcal{E}_G}x(P)(-1)^{\mathrm{loops}_{[u_1,...,u_m]}(P)}
\]%\end{equation}
for~$u_1,\dots,u_m\in G^\circ$, where~$\mathcal{E}_G$ denotes the set of all even subgraphs of~$G$,
\[%\begin{equation}\label{eq:ZG-def}
\textstyle \mathcal{Z}_G:=\sum_{P\in\mathcal{E}_G} x(P),\qquad x(P):=\prod_{e\in P} x_e,%\exp[-2\beta\sum_{e\in P} J_e],
\]%\end{equation}
and $\mathrm{loops}_{[u_1,...,u_m]}(P)$ is the number (always well defined modulo~$2$) of loops in~$P$ surrounding an odd number of faces~$u_1,...,u_m$. Up to a factor~$\exp[\beta\sum_{e\in \mathcal{E}_G}J_e]$, the quantity~$\mathcal{Z}_G$ is the~\emph{partition function} of the Ising model on~$G^\circ$.

\subsection{Disorder insertions}
\label{subsect:spin-disorder} Following Kadanoff and Ceva~\cite{kadanoff-ceva-71}, given an even number of vertices
$v_1,\dots,v_{n}\in G^\bullet$ we define the correlation of \emph{disorders}~$\mu_{v_1},\dots,\mu_{v_{n}}$
\begin{equation}
\label{eq:disorders-def}
\textstyle \Corr{G}{\mu_{v_1}\dots\mu_{v_{n}}}:= \mathcal{Z}_G^{-1}\cdot \mathcal{Z}_G^{[v_1,...,v_{n}]}\,,\quad \mathcal{Z}_G^{[v_1,...,v_{n}]}:=\sum_{P\in\mathcal{E}_G(v_1,...,v_{n})}x(P)\, ,
\end{equation}
where~$\mathcal{E}_G(v_1,...,v_{n})$ denotes the set of subgraphs $P$ of~$G$ such that each of the vertices~$v_1,\dots,v_{n}$ has an odd degree in~$P$ while all other vertices have an even degree. Probabilistically, one can easily see that
\begin{equation}
\label{eq:disorders-corr-proba}
\textstyle \Corr{G}{\mu_{v_1}\dots\mu_{v_{n}}}=\mathbb E_G\big[\exp[-2\beta\sum_{(uw)^*\in P_0(v_1,\dots,v_{n})} J_e\sigma_u\sigma_w]\,\big],
\end{equation}
where~$P_0(v_1,\dots,v_{n})$ is a fixed collection of edge-disjoint paths matching in pairs the vertices~$v_1,\dots,v_{n}$; note that the right-and side does not depend on the choice of these paths. The \emph{Kramers--Wannier duality} implies {(e.g., see~\cite{kadanoff-ceva-71})} that
\begin{equation}
\label{eq:disorders=dual-spins}
\textstyle \Corr{G}{\mu_{v_1}\dots\mu_{v_{n}}}=\mathbb E^\star_{G^\bullet}\big[\sigma^\bullet_{v_1}\dots\sigma^\bullet_{v_n}\,\big],
\end{equation}
where the expectation in the right-hand side is taken with respect to the Ising model on \emph{vertices} of~$G$, with dual weights~$x_{e^*}:=\tan\tfrac{1}{2}(\frac{\pi}{2}-\theta_e)$ and free boundary conditions. Indeed,~\eqref{eq:disorders-def} is nothing but the \emph{high-temperature expansion} of~\eqref{eq:disorders=dual-spins}.

Similarly to~$\mathcal{Z}_G$, one can interpret~$\mathcal{Z}_G^{[v_1,...,v_{n}]}$ as the {low-temperature (domain walls) expansion of} the partition function of the Ising model defined on the faces of a \emph{double cover}~$G^{[v_1,...,v_{n}]}$ of the graph~$G$ that branches over~$v_1,\dots,v_{n}$, with the following \emph{spin-flip symmetry constraint}: we require $\sigma_{u}\sigma_{u^\star}=-1$ for any pair of faces of the double cover such that~$u$ and $u^\star$ lie over the same face in~$G$. Using this interpretation, we introduce mixed correlations
\begin{equation}
\label{eq:disorders-spins-mixed-def}
\Corr{G}{\mu_{v_1}\dots\mu_{v_{n}}\sigma_{u_1}\dots\sigma_{u_m}}\;:=\; \Corr{G}{\mu_{v_1}\dots\mu_{v_{n}}}\cdot \mathbb{E}_{G^{[v_1,...,v_{n}]}}[\sigma_{u_1}\dots\sigma_{u_m}]\, ,
\end{equation}
where~$u_1,\dots,u_m$ should be understood as faces of the double cover~$G^{[v_1,...,v_{n}]}$. Similarly to~\eqref{eq:disorders-corr-proba} one can easily give a probabilistic interpretation of these quantities in terms of the original Ising model on~$G$. Nevertheless, we prefer to speak about the Ising model on~$G^{[v_1,...,v_{n}]}$ as this approach is more invariant and does not require to fix an {arbitrary} choice of the disorder lines~$P_0(v_1,\dots,v_{n})$.

By definition of the Ising model on~$G^{[v_1,...,v_{n}]}$, the correlation~\eqref{eq:disorders-spins-mixed-def} fulfills the sign-flip symmetry constraint between the sheets of the double cover.
%They also admit the following combinatorial expansion: %(see \cite{ChICM}):
%\[
%\textstyle \Corr{G}{\mu_{v_1}...\mu_{v_{n}}\sigma_{u_1}...\sigma_{u_m}} \;=\;
%\pm\, \mathcal{Z}_G^{-1}\cdot \sum_{P\in\mathcal{E}_G(v_1,...,v_{n})}x(P)(-1)^{\mathrm{loops}_{[u_1,...,u_m]}(P\triangle P_0)}\,,
%\]
%where the~$\pm$ sign depends on the identification of~$u_1,...,u_m$ with faces of~$G$ and~$P_0$ is a fixed collection of edge-disjoint paths matching in pairs the vertices~$v_1,...,v_{n}$.
When considered as a function of both vertices~$v_p$ and faces~$u_q$, this correlation is defined on a double cover %~$G^{n,m}_{[\bullet,\circ]}$
of~$(G^\bullet)^{n}\times(G^\circ)^m$ and changes sign each time one of the vertices~$v_p\in G^\bullet$ turns around one of the vertices~$u_q\in G^\circ$ (or vice versa).
We call \emph{spinors} functions defined on double covers that obey such a sign-flip property.

\subsection{Fermions and the propagation equation} We need an additional notation. {Let~$\Lambda(G)$ be a planar bipartite graph (the so-called \emph{quad-graph}) whose vertices are $G^\bullet\cup G^\circ$ and the set of (degree four) faces~$\diamondsuit(G)$ is in a $1$-to-$1$ correspondence with the set of edges of~$G$; in other words, the edges of $\Lambda(G)$ connect a vertex $v\in G^\bullet$ with all adjacent vertices $u\in G^\circ$ of the dual graph and vice versa.} Let~$\Upsilon(G)$ denote the \emph{medial} graph of~$\Lambda(G)$, %where the vertices of~$\Upsilon(G)$
whose vertices are in a $1$-to-$1$ correspondence with edges~$(vu)$ of~$\Lambda(G)$ and are also called \emph{corners} of~$G$, while the faces of~$\Upsilon(G)$ correspond either to vertices of~$G^\bullet$ or to vertices of~$G^\circ$ or to quads from~$\diamondsuit(G)$. We denote by~$\Upsilon^\times(G)$ a double cover of {the graph}~$\Upsilon(G)$ that branches around each of its faces (e.g., see~\cite[Fig.~3A]{chelkak-semb} or \cite[Fig.~27]{mercat-CMP},~\cite[Fig.~6]{chelkak-smirnov-12}). For a corner $c=(v(c)u(c))\in \Upsilon^\times(G)$ (with $u(c)\in G^\circ $ and $v(c)\in G^\bullet$), let
\begin{equation}
\label{eq:Dirac_spinor}
\eta_c:=i\cdot\exp[-\tfrac{i}{2}\arg(v(c)-u(c))],
\end{equation}
where the global prefactor~$i$ is chosen for later convenience. Though a priori the sign in the expression~\eqref{eq:Dirac_spinor} is ambiguous, it can be fixed so that~$\eta_c$ is a spinor on~$\Upsilon^\times(G)$, called the \emph{Dirac spinor}, by requiring that the values of~$\eta_c$ at corners $c$ surrounding each face of $\Upsilon(G)$ are defined in a `continuous' way. (In particular, this local definition implies the spinor property of~$\eta_c$ on~$\Upsilon^\times(G)$.)

Given $c\in \Upsilon^\times(G)$, one defines the \emph{Kadanoff--Ceva fermion} as $\chi_c:=\mu_{v(c)}\sigma_{u(c)}$.
More accurately, we set
\begin{equation}\label{eq:chi:=mu_sigma_def}
X_{\varpi}(c)%= \langle\chi_c\mu_{v_1}\dots\mu_{v_{m-1}}\sigma_{u_1}\dots\sigma_{u_{n-1}}\rangle
:= \Corr{G}{\mu_{v(c)}\mu_{v_1}\dots\mu_{v_{n-1}}\sigma_{u(c)}\sigma_{u_1}\dots\sigma_{u_{m-1}}},
\end{equation}
for $\varpi:=(v_1,\dots,v_{n-1},u_1,\dots,u_{m-1})\in (G^\bullet)^{n-1}\times(G^\circ)^{m-1}$. Let~$\Upsilon^\times_\varpi(G)$ denote a double cover of~$\Upsilon(G)$ that branches over each of {the faces of~$\Upsilon(G)$} \emph{except} {those corresponding to} the points from~$\varpi$.
The preceding discussion of mixed spin-disorder correlations ensures that $X_{\varpi}$ is a spinor on~$\Upsilon^\times_\varpi(G)$. Finally, let
%Finally, we define formally the compex valued correlator $\psi_c:=\eta_c\chi_c$, with~$\eta_c$ defined in~\eqref{eq:Dirac_spinor} and the complex valued observable
\begin{equation}
\label{eq:psi:=eta_chi_def}
\Psi_{\varpi}(c):=\eta_c X_{\varpi}(c)
%= \Corr{G}{\mu_{v(c)}\mu_{v_1}\dots\mu_{v_{n-1}}\sigma_{u(c)}\sigma_{u(c)}\sigma_{u_1}\dots\sigma_{u_{m-1}}},
%\langle\psi_c\mu_{v_1}...\mu_{v_{m-1}}\sigma_{u_1}...\sigma_{u_{n-1}}\rangle
%:= \eta_c\langle\chi_c\mu_{v_1}...\mu_{v_{m-1}}\sigma_{u_1}...\sigma_{u_{n-1}}\rangle.
\end{equation}
where~$\eta_c$ is defined by~\eqref{eq:Dirac_spinor}. The function $\Psi_{\varpi}$ locally does \emph{not} branch (the signs changes of~$\chi_c$ and~$\eta_c$ cancel each other). More precisely, $\Psi_\varpi$ is a spinor on the double cover~$\Upsilon_\varpi(G)$ of $\Upsilon(G)$ that branches \emph{only} over points from~$\varpi$: it changes the sign only when~$c$ turns around one of the vertices $v_p$ or the faces~$u_q$.

We now move on to the crucial three-term equation for the correlations~\eqref{eq:chi:=mu_sigma_def}, called the \emph{propagation equation} for Kadanoff--Ceva fermions on~$\Upsilon^\times(G)$, see~\cite{Perk-81dubna,dotsenko-dotsenko,mercat-CMP} or~\cite[Section~3.5]{chelkak-cimasoni-kassel} for more details. For a quad~$z_e\in\diamondsuit(G)$ corresponding to an edge~$e$ of~$G$, we denote its vertices by~\mbox{$v_0(z_e)\in G^\bullet$}, \mbox{$u_0(z_e)\in G^\circ$}, \mbox{$v_1(z_e)\in G^\bullet$}, and \mbox{$u_1(z_e)\in G^\circ$}, listed in the counterclockwise order. Further, for $p,q \in \{ 0 , 1 \}$, let \mbox{$c_{p,q}(z_e):=(v_p(z_e)u_q(z_e))$}. %\in\Upsilon^{\times}(G)$.
The following identity holds for all triples of consecutive (on~$\Upsilon^\times_\varpi(G)$) corners $c_{p,1-q}(z_e)$, $c_{p,q}(z_e)$ and $c_{1-p,q}(z_e)$ surrounding the edge~$e$:
\begin{equation}
\label{eq:propagation_on_corners}
X_\varpi(c_{p,q})=X_\varpi(c_{p,1-q})\cos\theta_e+X_\varpi(c_{1-p,q})\sin\theta_e\,,
\end{equation}
where~$\theta_e$ stands for the parametrization~\eqref{eq:parametrization-model} of the Ising model weight~$x_e$ of~$e$.
%(Moreover, a straighforward computation shows that the equation \eqref{eq:propagation_on_corners} implies the spinor property of~$X_\varpi$: if $c_{pq},c^\star_{pq}\in\Upsilon^\times_\varpi(G)$ lie over the same corner of~$G$, we have ~$F(c_{pq})=-F(c^\star_{pq})$.)
In recent papers, the equation~\eqref{eq:propagation_on_corners} is often used in the context of rhombic lattices, in which case the parameter~$\theta_e$ admits a geometric interpretation (see Section~\ref{subsect:isoradial}), but in fact it does not rely upon a particular choice of an embedding (up to a homotopy) of~$\diamondsuit(G)$ into~$\mathbb C$ provided that~$\theta_e$ is \emph{defined} by~\eqref{eq:parametrization-model}.

\begin{figure}
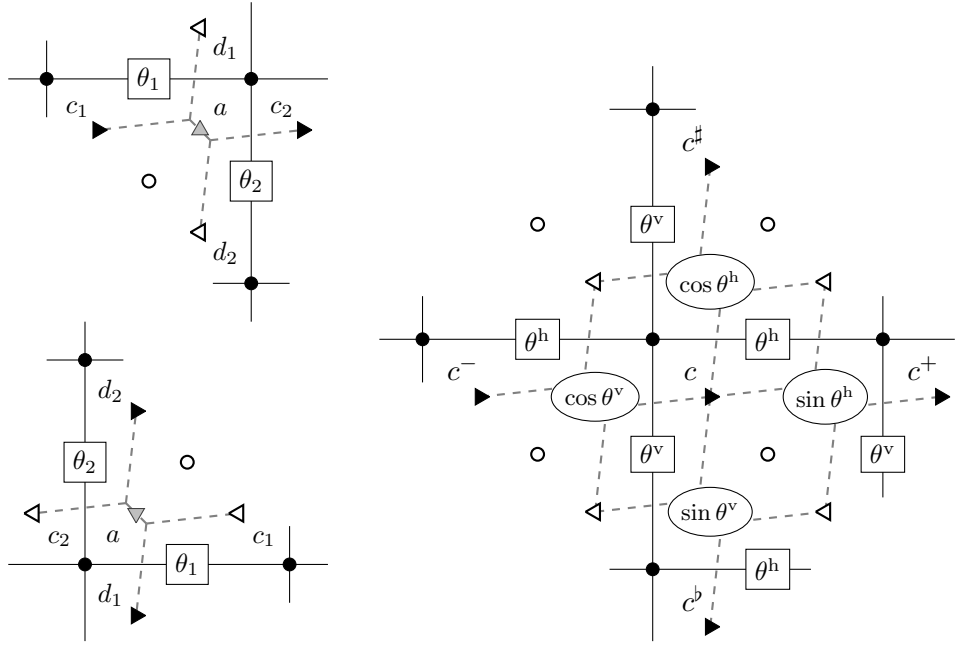
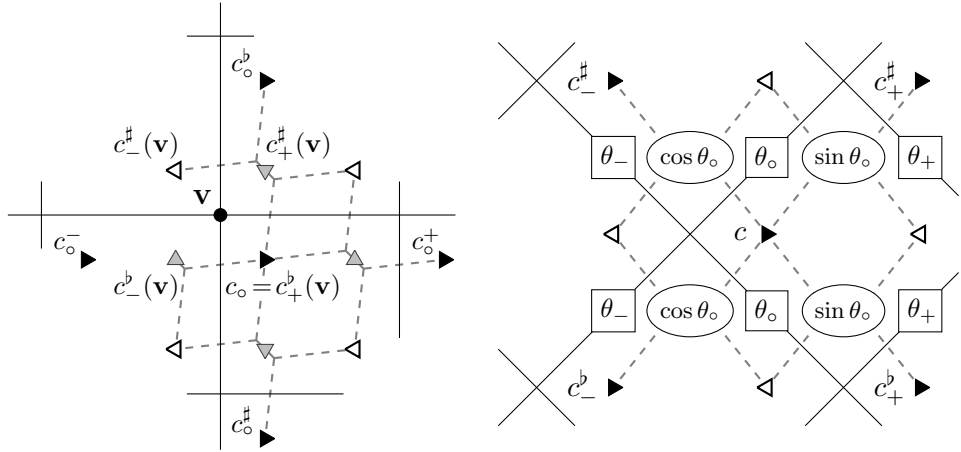

\centering{\begin{minipage}[b][][b]{0.36\textwidth}
\begin{tikzpicture}[scale=0.17]
\usetikzlibrary{shapes}
\input{Fig1A.txt}
\end{tikzpicture}

\bigskip

\textsc{(A)} The notation used in Proposition~\ref{prop:Cauchy-Riemann-general} (Cauchy--Riemann equations~\eqref{eq:Cauchy-Riemann-general}).
\end{minipage}
\hskip 0.05\textwidth
\begin{minipage}[b]{0.58\textwidth}
\hskip -0.05\textwidth
\begin{tikzpicture}[scale=0.17]
\usetikzlibrary{shapes}
\input{Fig1B.txt}
\end{tikzpicture}

\bigskip

\textsc{(B)} The notation used in Proposition~\ref{prop-massive-harmonicity} (massive harmonicity of fermionic observables in the homogeneous model away from the branchings).
\end{minipage}

\vskip 20pt

\begin{minipage}[b]{0.47\textwidth}
\begin{tikzpicture}[scale=0.17]
\usetikzlibrary{shapes}
\input{Fig1C.txt}
\end{tikzpicture}

\bigskip

\textsc{(C)} The notation used in the proof of Lemma~\ref{lem:laplacians} (the value $[\Delta^{(m)}X^\mathrm{sym}_{[\mathbf v,\mathbf u]}]$ near the branching point~$\mathbf v=(0,\frac{1}{4})$).

\end{minipage}
\hskip 0.05\textwidth
\begin{minipage}[b]{0.47\textwidth}
\hskip -0.03\textwidth\begin{tikzpicture}[scale=0.17]
\usetikzlibrary{shapes}
\input{Fig1D.txt}
\end{tikzpicture}
\bigskip

\textsc{(D)} The notation used in Proposition~\ref{prop-harmonicity-layered} (harmonicity-type identities in the zig-zag layered model).
\end{minipage}}
\caption{Local relations for Kadanoff--Ceva fermionic observables. {We indicate the four `types' of corners of (subgraphs of) the square grid by orienting and coloring the triangles depicting them.}}
\label{fig:LocalRelations}
\end{figure}

\subsection{Cauchy--Riemann and Laplacian-type identities on the square grid} From now on we assume that~$G$ is a subgraph of the regular square grid~$\mathbb Z^2\subset\mathbb C$. In this situation one can use~\eqref{eq:propagation_on_corners} to derive a version of discrete Cauchy--Riemann equations for the complex-valued observable~$\Psi_\varpi$ defined by~\eqref{eq:psi:=eta_chi_def}.

\begin{prop}\label{prop:Cauchy-Riemann-general}
Let~$c_1,d_1,c_2,d_2$ be corners of~$G$ located as in Fig.~\ref{fig:LocalRelations}A (and located on the same sheet of the double cover~$\Upsilon_\varpi(G)$). Let~$\theta_1,\theta_2$ be the interaction parameters assigned via~\eqref{eq:parametrization-model} to the edges~$e_1,e_2$. Then, the following identity holds:
\begin{equation}\label{eq:Cauchy-Riemann-general}
[\Psi_\varpi(c_2)\cos\theta_2-\Psi_\varpi(c_1)\sin\theta_1]\; =\; \pm i\cdot [\Psi_\varpi(d_2)\sin\theta_2-\Psi_\varpi(d_1)\cos\theta_1]\,,
\end{equation}
where the `$\pm$' sign is `$+$' if the square $(c_1d_2c_2d_1)$ is oriented counterclockwise (top picture in Fig.~\ref{fig:LocalRelations}A) and~`$-$' otherwise (bottom picture in Fig.~\ref{fig:LocalRelations}A).
\end{prop}
\begin{proof} Let~$a\in \Upsilon^\times_\varpi(G)$ be the center of the square~$(c_1d_2c_2d_2)$ and let~$c_1,d_2,c_2,d_1$ be the neighbors of~$a$ on~$\Upsilon^\times_\varpi(G)$. Writing two propagation equations~\eqref{eq:propagation_on_corners} at~$a$ one gets the identity
\[
X_\varpi(c_2)\cos\theta_2+X_\varpi(d_2)\sin\theta_2 \;=\; X_\varpi(a) \;=\; X_\varpi(c_1)\sin\theta_1+X_\varpi(d_1)\cos\theta_1.
\]
Since~$\eta_{d_1}=\eta_{d_2}=e^{\pm i\frac{\pi}{4}}\eta_a$ (with the same choice of the sign: `$+$' for the left picture, `$-$' for the right one) and~$\eta_{c_1}=\eta_{c_2}=e^{\mp i\frac{\pi}{4}}\eta_a$, the result immediately follows.
\end{proof}

Below we often focus on the values of observables~$\Psi_\varpi$ or~$X_\varpi$ at corners $c\in\Upsilon(G)$ of one of four `types'; by a type of~$c$ we mean its geometric position inside the face of~$G\subset \mathbb Z^2$ to which $c$ belongs, see~Fig.~\ref{fig:LocalRelations}. For each type of corners, the values~$\eta_c$ are all the same and, moreover, the branching structure of $\Upsilon^\times_\varpi(G)$ restricted to this type of corners coincides with the one of~$\Upsilon_\varpi(G)$. In other words, $\Psi_\varpi$ and~$X_\varpi$ differ only by a global multiplicative constant on each of the four types of corners.

In this paper, we are interested in the following two setups:

-- \emph{homogeneous} model, in which all the parameters~$\theta_e$ corresponding to horizontal edges of~$\mathbb Z^2$ have the common value $\theta^{\mathrm{h}}$ (resp., $\theta^{\mathrm{v}}$ for vertical edges);

-- \emph{zig-zag layered} model on the $\tfrac{\pi}{4}$-rotated grid, in which all interaction constants between each pair of adjacent columns have the same value (see Fig.~\ref{fig:GlobalLayered}).

In both situations, one can use~\eqref{eq:Cauchy-Riemann-general} to derive a harmonicity-type identity for the values of~$X_\varpi$ (note however that this is not possible in the general case).

\begin{prop}\label{prop-massive-harmonicity}
In the homogeneous setup, assume that a corner~$c\in\Upsilon_\varpi(G)$ is not located near the branching, i.e., that neither~$v(c)$ nor~$u(c)$ are in~$\varpi$. Then, the observable~$X_\varpi$ satisfies the following equation at~$c$:
\[
X_\varpi(c)\; =\; \tfrac{1}{2}{\sin\theta^{\mathrm{h}}\cos\theta^{\mathrm{v}}}\!\cdot[X_\varpi(c^{+})\!+\!X_\varpi(c^{-})] + \tfrac{1}{2}{\cos\theta^{\mathrm{h}}\sin\theta^{\mathrm{v}}}\!\cdot[X_\varpi(c^{\sharp})\!+\!X_\varpi(c^{\flat})],
\]
where $c^{+},c^{\sharp},c^{-},c^{\flat}$ are the four nearby corners {having} the same type as~$c$, located at the east, north, west and south direction from~$c$, respectively {(see Fig.~\ref{fig:LocalRelations}B).}
\end{prop}
\begin{proof} Recall that, at corners of a given type, the values~$X_\varpi$ and~$\Psi_\varpi$ differ only by a multiplicative constant. Due to the symmetry of the homogeneous model, we can assume that~$c,c^{+},c^{\sharp},c^{-},c^{\flat}$ are located as in Fig.~\ref{fig:LocalRelations}B. Let us write four Cauchy--Riemann equations~\eqref{eq:Cauchy-Riemann-general} between {$c$ and~$c^{+}$, $c$ and $c^{\sharp}$, $c$ and $c^{-}$, $c$ and $c^{\flat}$.} Multiplying the first equation by~$\sin\theta^{\mathrm{h}}$, the second by~$\cos\theta^{\mathrm{h}}$, the third by~$\cos\theta^{\mathrm{v}}$, the fourth by~$\sin\theta^{\mathrm{v}}$, and taking the sum with appropriate signs we get the result.
\end{proof}

\begin{rem} \label{rem:massive-Laplacian}
(i) Proposition~\ref{prop-massive-harmonicity} can be reformulated as the {massive harmonicity} condition~$[\Delta^{(m)}X_\varpi](c)=0$, where the \emph{massive Laplacian}~$\Delta^{(m)}$ is defined as
\[%\begin{equation}\label{eq:deformed-laplacian}
[\Delta^{(m)}F](c) := -F(c)+\tfrac{1}{2}{\sin\theta^{\mathrm{h}}\cos\theta^{\mathrm{v}}}\!\cdot[F(c^{+})+F(c^{-})] + \tfrac{1}{2}{\cos\theta^{\mathrm{h}}\sin\theta^{\mathrm{v}}}\!\cdot[F(c^{\sharp})+F(c^{\flat})].
\]%\end{equation}
It is worth noting that~$\Delta^{(m)}$ is a generator of a (continuous time) random walk on~$\mathbb Z^2$ with killing rate~$1-\sin(\theta^{\mathrm{h}}\!+\!\theta^{\mathrm{v}})$; in particular, one can easily guess from Proposition~\ref{prop-massive-harmonicity} the classical criticality condition
\[
\textstyle \theta^{\mathrm{h}}+\theta^{\mathrm{v}}=\frac{\pi}{2}\quad\Leftrightarrow\quad \sinh [2\beta J^\mathrm{h}]\cdot\sinh[2\beta J^\mathrm{v}]=1.
\]

\noindent (ii) The fact that the near-critical homogeneous Ising model on~$\mathbb{Z}^2$ admits a description via massive holomorphic fermions is a commonplace in the theoretical physics literature. In the probabilistic community, an explicit link between formulas for spin-spin correlations derived in~\cite{mccoy-wu-book} and the partition functions of killed random walks was pointed out in~\cite{messikh-arxiv06}. We refer the interested reader to the paper~\cite{beffara-duminil}, in which the massive holomorphicity property of fermionic observables was used for the analysis of the exponential decay rate of spin-spin correlations~$\mathbb{E}[\sigma_{0}\sigma_{n\alpha}]$, $n\to\infty$, and of its dependence on the direction~$\alpha$ in the \emph{super}-critical model on~$\mathbb{Z}^2$.
\end{rem}

A similar identity holds in the layered setup (see Fig.~\ref{fig:LocalRelations}D for the notation). Assume that~$c$ is a west corner of a face on the~$\frac{\pi}{4}$-rotated square grid. Denote by $c^\sharp_\pm$, $c^\flat_\pm$ the four nearby corners of the same type as~$c$ and let~$\theta_{-}$,~$\theta_\circ$ and~$\theta_{+}$ be the parameters assigned via~\eqref{eq:parametrization-model} to the edges to the left of~$c^{\sharp,\flat}_-$, to the left of $c$, and to the left of~$c^{\sharp,\flat}_+$, respectively.
\begin{prop} \label{prop-harmonicity-layered}
In the setup described above (see also {Fig.~\ref{fig:LocalRelations}D}), assume that neither~$v(c)$ nor~$u(c)$ are in~$\varpi$. Then, the following identity holds:
\[
X_\varpi(c)=\tfrac{1}{2}\sin\theta_-\cos\theta_\circ\!\cdot[X_\varpi(c^{\sharp}_{-})\!+\!X_\varpi(c^{\flat}_{-})]+ \tfrac{1}{2}\sin\theta_\circ\cos\theta_+\!\cdot[X_\varpi(c^{\sharp}_{+})\!+\!X_\varpi(c^{\flat}_{+})].
\]
\end{prop}
\begin{proof} The result follows by summing four Cauchy--Riemann equations~\eqref{eq:Cauchy-Riemann-general} with coefficients $\pm\cos\theta_\circ$, $\pm\sin\theta_\circ$ similarly to the proof of Proposition~\ref{prop-massive-harmonicity}.
\end{proof}
\begin{rem} It is worth emphasizing that the harmonicity-type identities discussed in Propositions~\ref{prop-massive-harmonicity} and~\ref{prop-harmonicity-layered} {fail} when~$c$ is located near the branching. The reason is that applying~\eqref{eq:Cauchy-Riemann-general} four times and summing the results as in the proofs of Proposition~\ref{prop-massive-harmonicity} and Proposition~\ref{prop-harmonicity-layered} one gets the difference~$X_\varpi(d^*)-X_\varpi(d)$ with~$d^*,d$ located over the same point on the \emph{different sheets} of the double cover~$\Upsilon_\varpi(G)$.
\end{rem}

%There also exist an analog identity on other types of corners, but contrarily to the horizontal homogeneous case  -where the weights for local values identities on one type of corners are all the same-, the weights now depend on the type of corners. We will see in the following sections that local identities (on the horizontal and the rotated grid) fail near the branchings.

\section{Homogeneous model}\label{sc:homogeneous}
In this section we discuss classical results on the horizontal spin-spin correlations in the infinite volume for the homogeneous model. Namely, we assume that all horizontal edges have a weight $\exp[-2\beta J^{\mathrm{h}}]=\tan \frac{1}{2}\theta^{\mathrm{h}}$ while all vertical edges have a weight $\exp[-2\beta J^{\mathrm{v}}]=\tan \frac{1}{2}\theta^{\mathrm{v}}$, see also Appendix in which the diagonal spin-spin correlations are treated in the fully homogeneous critical case~$\theta^{\mathrm{h}}=\theta^{\mathrm{v}}=\frac{\pi}{4}$. Though these results and even a roadmap of the proofs are well-known (e.g., see the classical treatment by McCoy and Wu~\cite{mccoy-wu-book}), we use this setup to illustrate a simplification that comes from working directly with {\emph{real-valued}} orthogonal polynomials instead of Toeplitz determinants, an approach that we apply to the layered model.

\begin{figure}
\centering{\begin{tikzpicture}[scale=0.15]
\input{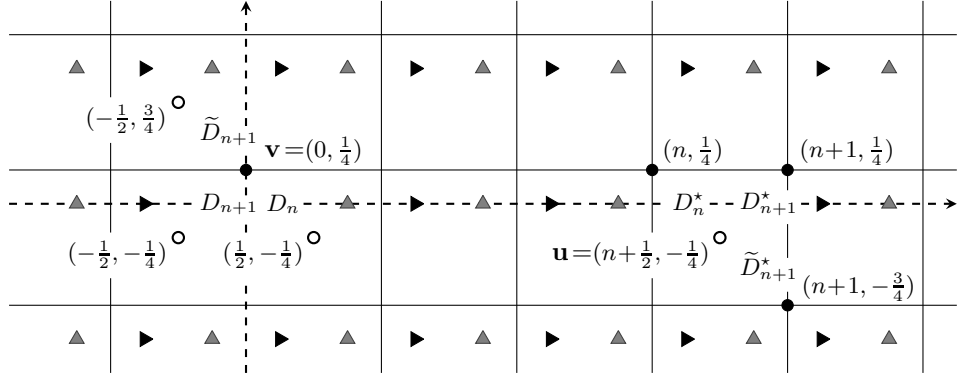}
\end{tikzpicture}

\bigskip

{\textsc {(A)} Particular values, considered \emph{up to the sign}, of the Kadanoff--Ceva $\phantom{(A}$~fermionic~observable $X_{[\mathbf v,\mathbf u]}$ near its branching points $\mathbf v,\mathbf u$; see~\eqref{eq:Dn-etc=}.}

\bigskip

\begin{tikzpicture}[scale=0.15]
\input{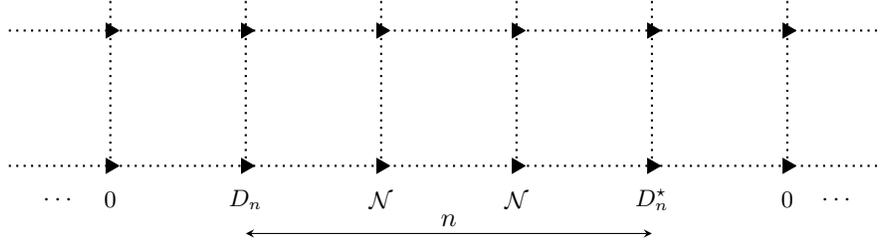}
\end{tikzpicture}

\bigskip

{\textsc (B)} Boundary value problem~$\mathrm{[P^{sym}_n]}$ for the symmetrized observable~$X^{\mathrm{sym}}_{[\mathbf v,\mathbf u]}$.

\bigskip

\begin{tikzpicture}[scale=0.15]
\input{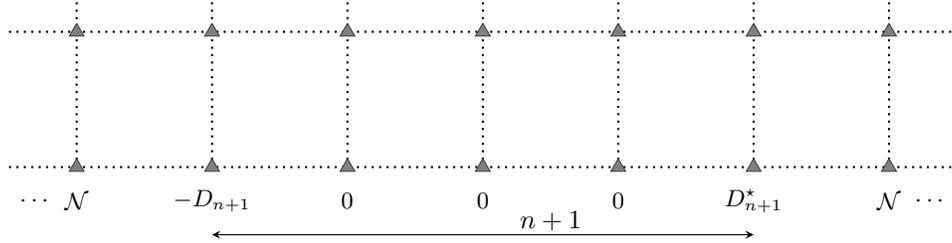}
\end{tikzpicture}

\bigskip

{\textsc (C)} Boundary value problem~$\mathrm{[P^{anti}_{n+1}]}$ for the anti-symmetrized observable~$X^{\mathrm{anti}}_{[\mathbf v,\mathbf u]}$.}

\caption{To derive recurrence relations on horizontal spin-spin correlations, we consider the Kadanoff--Ceva fermionic observable~$X_{[\mathbf v,\mathbf u]}$ with two branchings at~$\mathbf v=(0,\frac{1}{4})$ and~$\mathbf u=(n+\frac{1}{2},-\frac{1}{4})$. The symmetrized observable~$X^\mathrm{sym}_{[\mathbf v,\mathbf u]}$ is defined on north-west corners (marked as~$\triangleright$ in the figure) and the anti-symmetrized observable~$X^\mathrm{anti}_{[\mathbf v,\mathbf u]}$ is defined on north-east corners (marked as~$\scriptscriptstyle{\triangle}$). Both~$X^\mathrm{sym}_{[\mathbf v,\mathbf u]}$ and~$X^\mathrm{anti}_{[\mathbf v,\mathbf u]}$ are massive harmonic in the upper half-plane~$\mathbb Z\times\mathbb N_0$ and solve boundary value problems~$\mathrm{[P^{sym}_n]}$,~$\mathrm{[P^{anti}_{n+1}]}$, respectively; the sign $\mathcal{N}$ denotes Neumann boudary conditions.}
\label{fig:GlobalPicture}
\label{fig:sym-anti-bvp}
\end{figure}

\subsection{Full-plane observable with two branchings} Assume that the square grid on which the Ising model lives is shifted so that its vertices coincide with~\mbox{$\mathbb Z\times (\mathbb Z+\tfrac{1}{4})$} and the centers of faces are~\mbox{$(\mathbb Z+\frac{1}{2})\times (\mathbb Z-\tfrac{1}{4})$}, see Fig.~\ref{fig:GlobalPicture}A. It is well known (e.g., see~\cite{friedli-velenik-book}) that there are no more than two extremal Gibbs measures (coming from~`$+$' and `$-$' boundary conditions at infinity) and that the spin correlations in the infinite volume limit are translationally invariant. Given~$n\ge 0$, we define the horizontal and next-to-horizontal correlations
\begin{align}
D_n:=\mathbb{E}_{\mathbb Z^2}\big[\sigma_{(\frac{1}{2},-\frac{1}{4})} \sigma_{(n+\frac{1}{2},-\frac{1}{4})}]\,,\qquad & D_n^{\star}:= \mathbb{E}^\star_{(\mathbb Z^2)^\bullet}[\sigma^\bullet_{(0,\frac{1}{4})} \sigma^\bullet_{(n,\frac{1}{4})}]\,,\notag \\
D_{n+1}:=\mathbb{E}_{\mathbb Z^2}\big[\sigma_{(-\frac{1}{2},-\frac{1}{4})} \sigma_{(n+\frac{1}{2},-\frac{1}{4})}]\,,\qquad
  & D_{n+1}^{\star}:= \mathbb{E}^\star_{(\mathbb Z^2)^\bullet}[\sigma^\bullet_{(0,\frac{1}{4})} \sigma^\bullet_{(n+1,\frac{1}{4})}]\,,
  \label{eq:Dn-etc=}\\ \widetilde{D}_{n+1}:= \mathbb{E}_{\mathbb Z^2}[\sigma_{(-\frac{1}{2},\frac{3}{4})}\sigma_{(n+\frac{1}{2},-\frac{1}{4})}]\,,\qquad & \widetilde{D}_{n+1}^{\star}:= \mathbb{E}^\star_{(\mathbb Z^2)^\bullet}[\sigma^\bullet_{(0,\frac{1}{4})} \sigma^\bullet_{(n+1,-\frac{3}{4})}]\,,\notag
\end{align}
where the expectations in the second column are taken for the \emph{dual} Ising model with interaction parameters~$\tan \frac{1}{2}(\frac{\pi}{2}-\theta^\mathrm{v})$ and~$\tan \frac{1}{2}(\frac{\pi}{2}-\theta^\mathrm{h})$ assigned to horizontal and vertical edges of the dual {square} grid $(\mathbb{Z}^2)^\bullet$, respectively. Due to~\eqref{eq:disorders=dual-spins} one can view these quantities as disorder-disorder correlations in the original model.

Let~$\mathbf{v}=(0,\frac{1}{4})$ and $\mathbf{u}=(n+\frac{1}{2},-\frac{1}{4})$. Below we rely upon the \emph{full-plane} observable~$X_{[\mathbf{v},\mathbf{u}]}$ which can be thought of as a (subsequential) limit of the similar observables defined on finite graphs~$G$ exhausting the square grid. Indeed, since
\begin{equation}
\label{eq:X-bound-onG}
\big|\Corr{G}{\mu_{v(c)}\mu_{\mathbf{v}}\sigma_{u(c)}\sigma_{\mathbf{u}}}\big|\ \le\ \Corr{G}{\mu_{v(c)}\mu_{\mathbf{v}}}\ =\ \mathbb{E}_{G^\bullet}^{\star}[\sigma^\bullet_{v(c)}\sigma^\bullet_{\mathbf v}]\ \le\ 1,
\end{equation}
{a point-wise subsequential limit exists; its uniqueness (and hence the existence of the true limit) follows from Lemma~\ref{lem:uniqueness} given below. Moreover, in Section~\ref{subsect:construction-homogen}, we provide an explicit construction of functions satisfying the conditions listed in Lemma~\ref{lem:uniqueness}, which allows us to identify~$X_{[\mathbf{v},\mathbf{u}]}$ with these explicit functions.}

Let~$[(\mathbb Z\pm\tfrac{1}{4})\times\mathbb Z;\mathbf v,\mathbf u]$ denote the double cover of the lattice $(\mathbb Z\pm\tfrac{1}{4})\times\mathbb Z$ branching over~$\mathbf v$ and~$\mathbf u$.
We now introduce the following \emph{symmetrized} and \emph{anti-symmetrized} versions of the observable~$X_{[\mathbf{v},\mathbf{u}]}(\cdot)$ on north-west and north-east corners, respectively (see Fig.~\ref{fig:GlobalPicture}):
\begin{align}
\label{eq:Xsym-def}%X^{sym}(\bullet)=
\Xsym(c):=\tfrac{1}{2}[X_{[\mathbf{v},\mathbf{u}]}(c)+X_{[\mathbf{v},\mathbf{u}]}(\bar{c})],\quad &c\in [(\mathbb Z+\tfrac{1}{4})\times \mathbb Z;\mathbf v,\mathbf u],\\
\label{eq:Xanti-def}%X^{ant}(\bullet)=
\Xanti(c):=\tfrac{1}{2}[X_{[\mathbf{v},\mathbf{u}]}(c)-X_{[\mathbf{v},\mathbf{u}]}(\bar{c})],\quad &c\in [(\mathbb Z-\tfrac{1}{4})\times\mathbb Z;\mathbf v,\mathbf u],
\end{align}
where the continuous conjugation $z\mapsto \bar{z}$ on~$[(\mathbb Z\pm\tfrac{1}{4})\times\mathbb Z;\mathbf v,\mathbf u]$ is defined so that it maps the segment $[\tfrac{1}{4},n+\tfrac{1}{4}]\times \{ 0 \} $ %between $\mathbf v$ and $\mathbf u$
to itself (i.e., the conjugate of each point located over this segment is chosen to be on the \emph{same} sheet of the double cover). Once~$z\mapsto \bar{z}$ is specified in between the branching points, it {can be `continuously' extended} to the entire {double cover} $[(\mathbb Z\pm\tfrac{1}{4})\times\mathbb Z;\mathbf v,\mathbf u]$. In particular, the points~$c$ located over the real line outside of the segment~$[\frac{1}{4},n+\frac{1}{4}]$ are mapped by~$z\mapsto\overline{z}$ to their counterparts~$c^*$ on the \emph{other} sheet of the double cover.
%(Note that there is an alternative choice of the mapping $z\mapsto\bar{z}$ for which the points~$c\in [\tfrac{1}{4};n-\tfrac{1}{4}]\times \{ 0 \}$ are mapped to their counterparts~$c^\star$ located on the other sheet of the double cover.)

Let us list basic properties of the observables~$\Xsym$ and~$\Xanti$ and show that they characterize these observables uniquely. Due to~\eqref{eq:X-bound-onG} we have
\[%\begin{equation}\label{eq:X-bound}
\phantom{x}\;\qquad \big|\Xsym(k\!+\!\tfrac{1}{4},s)\big|\le 1\quad \text{and}\quad \big|\Xanti(k\!-\!\tfrac{1}{4},s)\big|\le 1\quad \text{for all}\ \ k,s\in\mathbb Z.
\]%\end{equation}
Proposition~\ref{prop-massive-harmonicity} (see also Remark~\ref{rem:massive-Laplacian}) ensures that the observables~$\Xsym$ and~$\Xanti$ are massive harmonic away from the branching points~$\mathbf v,\mathbf u$. In particular, one has
\begin{align}
\label{eq:X-laplacian}
[\Delta^{(m)}\!\Xsym]((k\!+\!\tfrac{1}{4},s))=0\quad \text{and}\quad [\Delta^{(m)}\!\Xanti]((k\!-\!\tfrac{1}{4},s))=0\quad \text{if}\ \ s\ne 0.
\end{align}
Further, the spinor property of the observable~$X_{[\mathbf v,\mathbf u]}$ together with the choice of the conjugation described above gives
\begin{align}
\label{eq:Xsym-vanishes}
\Xsym((k\!+\!\tfrac{1}{4},0))\!=\!0,\ k\not\in[0,n],\quad & [\Delta^{(m)}\!\Xsym]((k\!+\!\tfrac{1}{4},0))\!=\!0,\ k\in [1,n\!-\!1];\\
\label{eq:Xanti-vanishes}
\Xanti((k\!-\!\tfrac{1}{4},0))\!=\!0, \ k\in [1,n],\quad & [\Delta^{(m)}\!\Xanti]((k\!-\!\tfrac{1}{4},0))\!=\!0,\ k\not\in [0,n\!+\!1].
\end{align}
Finally, the definition of~$X_{[\mathbf v,\mathbf u]}$
%and the identities~$\mu_{\mathbf v}\mu_{\mathbf v}=1$, $\sigma_{\mathbf u}\sigma_{\mathbf u}=1$
implies
\begin{align}
\label{eq:Xsym-values}
\Xsym((\tfrac{1}{4},0))=D_n,\qquad & \Xsym((n\!+\!\tfrac{1}{4},0))=D_n^\star;\\
\label{eq:Xanti-values}
\Xanti((-\tfrac{1}{4},0))=-D_{n+1},\qquad & \Xanti((n\!+\!\tfrac{3}{4},0))=D_{n+1}^\star,
\end{align}
where we assume that these pairs of corners are located on the \emph{same sheet} of the double cover~$[(\mathbb Z\pm \tfrac{1}{4})\times\mathbb Z;\mathbf v,\mathbf u]$ \emph{as viewed from the upper half-plane}; this is why the value~$D_{n+1}$ at~$(-\tfrac{1}{4},0)$ appears with the different sign.

\begin{lem} \label{lem:uniqueness}
(i) The uniformly bounded observable~$\Xsym$ given by~\eqref{eq:Xsym-def} is uniquely characterized by the properties~\eqref{eq:X-laplacian}, \eqref{eq:Xsym-vanishes} and its values~\eqref{eq:Xsym-values} near~$\mathbf v$ and~$\mathbf u$.

\smallskip
\noindent (ii) Similarly, the uniformly bounded observable~$\Xanti$ given by~\eqref{eq:Xanti-def} is uniquely characterized by the properties~\eqref{eq:X-laplacian}, \eqref{eq:Xanti-vanishes} and its values~\eqref{eq:Xanti-values} near~$\mathbf v$ and~$\mathbf u$.
\end{lem}

{
\begin{proof} (i) Let $X_1$ and $X_2$ be two bounded spinors satisfying~\eqref{eq:X-laplacian},\eqref{eq:Xsym-vanishes} and~\eqref{eq:Xsym-values}. Denote $X_0:=X_1-X_2$. It follows from~\eqref{eq:X-laplacian} and \eqref{eq:Xsym-vanishes} that $[\Delta^{(m)}X_0](k+\frac{1}{4},s)=0$ if $s\ne 0$ or if $s=0$ and $k\in[1,n-1]$. Moreover, $X_0(k+\frac{1}{4},0)=0$ if $k\not\in [1,n-1]$ due to~\eqref{eq:Xsym-vanishes} and~\eqref{eq:Xsym-values}. We need to show that these conditions, together with the uniform boundedness, imply that $X_0$ vanishes everywhere on $[(\mathbb Z+\tfrac{1}{4})\times\mathbb Z;\mathbf u,\mathbf v]$.

Let $(Z_j)_{j\geq 0}$ be a random walk (with killing) on the double cover $[(\mathbb Z+\tfrac{1}{4})\times\mathbb Z;\mathbf u,\mathbf v]$ that corresponds to the massive Laplacian~$\Delta^{(m)}$. (In other words, at each step this random walk makes a $(\pm 1,0)$ jump with probability $\frac{1}{2}\sin\theta^\mathrm{h}\cos\theta^\mathrm{v}$, a $(0,\pm 1)$ jump with probability $\frac{1}{2}\cos\theta^\mathrm{h}\sin\theta^\mathrm{v}$ and, if $\theta^\mathrm{h}+\theta^\mathrm{v}\ne \frac12\pi$, dies with probability $1-\sin(\theta^\mathrm{h}\!+\!\theta^\mathrm{v})$.) Let~$\tau$ be the first time when this random walk hits the set $\{(k\!+\!\tfrac{1}{4},0),k\not\in [1,n\!-\!1]\}\}$ or dies; clearly $\tau<+\infty$ almost surely. Since the process $(X_0(Z_j))_{j\ge 0}$ is a uniformly bounded martingale, the optional stopping theorem implies that $X_0(c)=\mathbb{E}[X_0(Z_\tau)]=0$ for all starting points $c\in [(\mathbb Z+\tfrac{1}{4})\times\mathbb Z;\mathbf u,\mathbf v]$.

The proof of (ii) is similar. In this case $\tau$ defines the first moment when a similar random walk on $[(\mathbb Z-\tfrac{1}{4})\times\mathbb Z;\mathbf u,\mathbf v]$ hits the set $\{(k\!-\!\tfrac{1}{4},0),k\in [0,n\!+\!1]\}\}$ or dies. Note that $\tau<+\infty$ almost surely even if $\theta^\mathrm{h}+\theta^\mathrm{v}=\frac{\pi}{2}$ as $(Z_k)_{k\ge 0}$ is recurrent.
\end{proof} }

%In the next section we construct bounded spinor satisfying~\eqref{eq:X-laplacian}--\eqref{eq:Xanti-values} by means of the Fourier transform and orthogonal polynomials techniques.
The next lemma allows one to use {an explicit construction of functions $X_{[\mathbf{v},\mathbf{u}]}$ given in Section~\ref{subsect:construction-homogen}} in order to get a recurrence relation for the spin-spin correlations. For~$n\ge 1$, denote
\[%\begin{equation}\label{eq:L-def}
L_n:=\tfrac{1}{2}\cos\theta^{\mathrm{v}}\cdot[D_{n}+\cos\theta^\mathrm{h}\cdot \widetilde{D}_{n}],\qquad
L_n^\star:=\tfrac{1}{2}\sin\theta^{\mathrm{h}}\cdot[D_{n}^\star+\sin\theta^\mathrm{v}\cdot \widetilde{D}_{n}^\star].
\]%\end{equation}
\begin{lem} \label{lem:laplacians}
For each~$n\ge 1$, the following identities are fulfilled:
\begin{align}
\label{eq:Xsym-laplacians}
-[\Delta^{(m)}\Xsym]((\tfrac{1}{4},0))=L_{n+1},\qquad & -[\Delta^{(m)}\Xsym]((n\!+\!\tfrac{1}{4},0))=L_{n+1}^\star;\\
\label{eq:Xanti-laplacians}
-[\Delta^{(m)}\Xanti]((-\tfrac{1}{4},0))=-L_n,\qquad & -[\Delta^{(m)}\Xanti]((n\!+\!\tfrac{3}{4},0))=L_n^\star,
\end{align}
with the same choice of points on the double covers~$[(\mathbb Z\pm \tfrac{1}{4})\times\mathbb Z;\mathbf v,\mathbf u]$ as above. %in~\eqref{eq:Xsym-values},~\eqref{eq:Xanti-values}.
If~$n=0$, the identities~\eqref{eq:Xsym-laplacians} should be replaced by~$-[\Delta^{(m)}\Xsym]((\tfrac{1}{4},0))=L_1\!+\!L_1^\star$ while \eqref{eq:Xanti-laplacians} hold with $L_0:=\cos\theta^{\mathrm{v}}$ and $L_0^\star:=\sin\theta^{\mathrm{h}}$.
\end{lem}

\begin{proof} We focus on the first identity in~\eqref{eq:Xsym-laplacians}. Let~$c_\circ=c_+^\flat(\mathbf v):=(\tfrac{1}{4},0)$, see Fig.~\ref{fig:LocalRelations}C for the notation. First, note that~$\Xsym(c_\circ^-)=0$ and hence
\begin{align*}
\notag -[\Delta^{(m)}\Xsym](c_\circ)=X_{[\mathbf v,\mathbf u]}(c_\circ) & -
\tfrac{1}{2}{\sin\theta^{\mathrm{h}}\cos\theta^{\mathrm{v}}}\!\cdot X_{[\mathbf v,\mathbf u]}(c_\circ^+) \\ &- \tfrac{1}{2}{\cos\theta^{\mathrm{h}}\sin\theta^{\mathrm{v}}}\!\cdot [X_{[\mathbf v,\mathbf u]}(c_\circ^\sharp)+X_{[\mathbf v,\mathbf u]}(c_\circ^\flat)]\,.
%\label{eq:x-Xsym-laplacian}
\end{align*}
Recall that we deduced the massive harmonicity property of the observables~$X_{[\mathbf v,\mathbf u]}$ away from the branchings from four Cauchy--Riemann identities~\eqref{eq:Cauchy-Riemann-general}, each of them based upon two propagation equations~\eqref{eq:propagation_on_corners}; {see Fig.~\ref{fig:LocalRelations}B.} We now repeat the same proof but with \emph{seven} three-terms identities~\eqref{eq:propagation_on_corners} instead of eight ones required to prove Proposition~\ref{prop-massive-harmonicity}, the one involving the values of~$X_{[\mathbf v,\mathbf u]}$ at \mbox{$c_\circ^-=(-\tfrac{3}{4},0)$}, \mbox{$c_-^\flat(\mathbf v)=(-\tfrac{1}{4},0)$} and \mbox{$c_-^\sharp(\mathbf v)=(-\tfrac{1}{4},\tfrac{1}{2})$} missing; {see Fig.~\ref{fig:LocalRelations}C.} As a result, one sees that the value $[\Delta^{(m)}\Xsym](c_+^\flat(\mathbf v))$ is $\tfrac{1}{2}\cos\theta^\mathrm{v}$ times the missing linear combination of the values
\[
X_{[\mathbf v,\mathbf u]}(c_-^\flat(\mathbf v))=D_{n+1}
 \quad \text{and}\quad X_{[\mathbf v,\mathbf u]}(c_-^\sharp(\mathbf v))\cdot\cos\theta^{\mathrm h}=\widetilde{D}_{n+1}\cdot\cos\theta^\mathrm{h},
\]
which leads to the first identity in~\eqref{eq:Xsym-laplacians} (we let the reader to check the signs obtained along the computation). The proofs of the other three identities for~$n\ge 1$ are similar. If~\mbox{$n=0$}, one should sum \emph{six} three-term identities~\eqref{eq:propagation_on_corners} when dealing with~$\Xsym$ and \emph{eight} ones when dealing with~$\Xanti$. In the latter case, the values~$L_0$ and~$L_0^\star$ appear due to the presence of the branchings~$\mathbf v,\mathbf u$ near the points at which~$\Delta^{(m)}\Xanti$ is computed (and due to the fact that~$D_0=D_0^\star=1$).
\end{proof}

\subsection{Construction via the Fourier transform and orthogonal polynomials} \label{subsect:construction-homogen}
In this section we construct two bounded functions satisfying the properties~\eqref{eq:X-laplacian}--\eqref{eq:Xanti-values} using Fourier transform and orthogonal polynomials techniques, the explicit formulas are given in Lemma~\ref{lem:solution-sym} and Lemma~\ref{lem:solution-anti}. Recall that these explicit solutions must coincide with~$\Xsym$ and~$\Xanti$ due to Lemma~\ref{lem:uniqueness}. Instead of the double covers~$[(\mathbb Z\pm\tfrac{1}{4})\times\mathbb Z;\mathbf u,\mathbf v]$, we work in the upper half-plane~$\mathbb Z\times \mathbb N_0$ only (see Lemma~\ref{lem:H-to-dbl-cover} for the link between the two setups).

For a function~$V:\mathbb Z\times\mathbb N_0\to\mathbb R$ we use the same definition of the massive Laplacian~$[\Delta^{(m)}V](k,s)$ as above for~$s\ge 1$ and introduce the values
\begin{align}
\notag [\mathcal{N}V](k,0):=  V(k,0) & - \cos\theta^{\mathrm{h}}\sin\theta^{\mathrm{v}}\cdot V(k,1)\\
 & - \tfrac{1}{2} \sin\theta^{\mathrm{h}}\cos\theta^{\mathrm{v}}\cdot [V(k\!-\!1,0)+V(k\!+\!1,0)]
 \label{eq:N-def}
\end{align}
which might be viewed as a version of the normal derivative of~$V$ at the point~$(k,0)$. We now formulate two problems~$\mathbf{[P^{sym}_n]}$ and~$\mathbf{[P^{anti}_n]}$ to solve. {Due to Lemma~\ref{lem:uniqueness}, these problems are equivalent to constructing explicitly the functions~$X^\mathrm{sym}_{[\mathbf{v},\mathbf{u}]}$ and $X^\mathrm{anti}_{[\mathbf{v},\mathbf{u}]}$, respectively; see also Fig.~\ref{fig:sym-anti-bvp}B and Fig.~\ref{fig:sym-anti-bvp}C.}

\smallskip

\begin{itemize}
\item $\mathbf{[P^{sym}_n]}:$ given~$n\ge 1$, to construct a bounded function~$V:\mathbb Z\times\mathbb N_0\to\mathbb R$ such~that the following conditions are fulfilled:
    \begin{align*}
    [\Delta^{(m)}V](k,s)=0~\text{if}~s\ge 1;\qquad & [\mathcal{N}V](k,0)=0~\text{for}~k\in [1,n\!-\!1];\\
    V(k,0)=0~\text{for}~k\not\in [0,n];\qquad & V(0,0)=D_n\ \ \text{and}\ \ V(n,0)=D_n^\star\,.
    \end{align*}
\item $\mathbf{[P^{anti}_{n+1}]}:$ given~$n\ge 0$, to construct a bounded function~$V:\mathbb Z\times\mathbb N_0\to\mathbb R$ such~that the following conditions are fulfilled:
    \begin{align*}
    \quad\qquad [\Delta^{(m)}V](k,s)=0~\text{if}~s\ge 1;\qquad & [\mathcal{N}V](k,0)=0~\text{for}~k\not\in [0,n\!+\!1];\\
    V(k,0)=0~\text{for}~k\in [1,n];\qquad & V(0,0)=-D_{n+1};\quad V(n\!+\!1,0)=D_{n+1}^\star\,.
    \end{align*}
\end{itemize}

\begin{lem}\label{lem:H-to-dbl-cover} Assume that a function~$V^\mathrm{sym}_n$ (resp., $V^\mathrm{anti}_{n+1}$) solves the problem $\mathrm{[P^{sym}_n]}$ (resp., $\mathrm{[P^{anti}_{n+1}]}$). Then, the following identities hold:
\begin{align}
\label{eq:Vsym-laplacians}
[\mathcal{N}V^\mathrm{sym}_n](0,0)=L_{n+1},\qquad & [\mathcal{N}V^\mathrm{sym}_n](n,0)=L_{n+1}^\star;\\
\label{eq:Vanti-laplacians}
[\mathcal{N}V^\mathrm{anti}_{n+1}](0,0)=-L_n,\,\qquad & [\mathcal{N}V^\mathrm{anti}_{n+1}](n\!+\!1,0)=L_n^\star.
\end{align}
\end{lem}

\begin{proof} Consider a section of the double cover~$[(\mathbb Z\pm \frac{1}{4})\times\mathbb Z;\mathbf v,\mathbf u]$ with a cut going along the horizontal axis outside the segment $[0,n\!+\!\tfrac{1}{2}]$ for the problem~$\mathrm{[P^{sym}_n]}$ and along $[0,n\!+\!\tfrac{1}{2}]$ for the problem~$\mathrm{[P^{anti}_{n+1}]}$. Define two functions on north-west and north-east, respectively, corners of the grid by
\[
V^\mathrm{sym}_{[\mathbf v,\mathbf u]}((\pm k\!+\!\tfrac{1}{4},s)):=V^\mathrm{sym}_n(k,s)\qquad V^\mathrm{anti}_{[\mathbf v,\mathbf u]}((\pm k\!-\!\tfrac{1}{4},s)):=\pm V^\mathrm{anti}_{n+1}(k,s).
\]
These functions vanish on the cuts and thus can be viewed as bounded spinors on the double covers~$[(\mathbb Z\pm \frac{1}{4})\times\mathbb Z;\mathbf v,\mathbf u]$, which satisfy all the conditions~\eqref{eq:X-laplacian}--\eqref{eq:Xanti-values}. Due to the uniqueness result provided by Lemma~\ref{lem:uniqueness}, this implies~$\Xsym=V^\mathrm{sym}_{[\mathbf v,\mathbf u]}$ and~$\Xanti=V^\mathrm{anti}_{[\mathbf v,\mathbf u]}$. The identities~\eqref{eq:Vsym-laplacians}, \eqref{eq:Vanti-laplacians} now easily follow from~\eqref{eq:Xsym-laplacians}, \eqref{eq:Xanti-laplacians} and the definition~\eqref{eq:N-def}.
\end{proof}

Let~$V$ be a solution to the problem~$\mathrm{[P^{sym}_n]}$, recall that this solution is unique due to Lemma~\ref{lem:uniqueness}. To construct it explicitly, we start with a heuristic argument. Assume for a moment that the Fourier series
\begin{equation}
\textstyle \widehat{V}_s(e^{it}):=\sum_{k \in \mathbb{Z}} V(k,s)e^{ikt},\quad s\ge 0, \quad t\in [0,2\pi],\nonumber
\end{equation}
are well-defined. The massive harmonicity property~$[\Delta^{(m)}V](k,s)=0$ for~$s\ge 1$ can be rewritten as the recurrence relation
\[%\begin{equation}\label{eq:Fourier-series-reccurence}
[1-\sin\theta^{\mathrm{h}}\cos\theta^{\mathrm{v}}\cos t]\cdot \widehat{V}_s(e^{it})\ =\ \tfrac{1}{2}\cos\theta^{\mathrm{h}}\sin\theta^{\mathrm{v}}\cdot  [\widehat{V}_{s-1}(e^{it}) +\widehat{V}_{s+1}(e^{it})].
\]%\end{equation}
A general solution to this recurrence relation is a linear combination of the functions $(y_-(t;\theta^{\mathrm{h}}\!,\theta^{\mathrm{v}}))^s$ and $(y_+(t;\theta^{\mathrm{h}}\!,\theta^{\mathrm{v}}))^s$, where $0\le y_-\le 1\le y_+$ solve the quadratic equation
\[
[1-\sin\theta^{\mathrm{h}}\cos\theta^{\mathrm{v}}\cos t]\cdot y(t)\ =\ \tfrac{1}{2}\cos\theta^{\mathrm{h}}\sin\theta^{\mathrm{v}}\cdot[(y(t))^2+1].
\]
At level~$s=0$ we have $\widehat{V}_0(e^{it})=Q_n(e^{it})$, an unknown trigonometric polynomial of degree $n$. Since we are looking for bounded Fourier coefficients of~$\widehat{V}_s$, we are tempted to say that $\widehat{V}_s(e^{it})=Q_n(e^{it})\cdot(y_-(t;\theta^{\mathrm{h}}\!,\theta^{\mathrm{v}}))^s$ for~$s\ge 1$. A straightforward computation shows that
\begin{align}
\notag %\label{eq:NV-fourier}
\textstyle \sum_{k\in\mathbb Z}[\mathcal{N}V](k,0)e^{ikt}\ &=\ w(t;\theta^{\mathrm{h}}\!,\theta^{\mathrm{v}})Q_n(e^{it}),\\
\label{eq:w-def}
w(t;\theta^{\mathrm{h}}\!,\theta^{\mathrm{v}})\
%&=\ 1-\sin\theta^{\mathrm{h}}\cos\theta^{\mathrm{v}}\cos t- \cos\theta^{\mathrm{h}}\sin\theta^{\mathrm{v}}\cdot y_-(t)\\
& :=\ \big[(1\!-\!\sin\theta^{\mathrm{h}}\cos\theta^{\mathrm{v}}\cos t)^2-(\cos\theta^{\mathrm{h}}\sin\theta^{\mathrm{v}})^2\big]^{1/2}.
\end{align}
The key observation of this section is that {the property~$[\mathcal{N}V](k,0)=0$, $k\in[1,n-1]$, reads as} a simple orthogonality condition for the polynomial~$Q_n(e^{it})$.

We now use the heuristics developed in the previous paragraph to rigorously identify the unique solution to~$\mathrm{[P^{sym}_n]}$.
\begin{lem}%[Solution to the problem~$\mathrm{[P^{sym}_n]}$]
\label{lem:solution-sym}
Let~$n\ge 1$. If a trigonometric polynomial \mbox{$Q_n(e^{it})=D_n\!+\ldots+\!D_n^{\star}e^{int}$} of degree~$n$ with prescribed free and leading coefficients is orthogonal to the family $\{e^{it},\ldots,e^{i(n-1)t}\}$ with respect to the measure $w(t;\theta^{\mathrm{h}}\!,\theta^{\mathrm{v}})\frac{dt}{2\pi}$, then the function
\[
\textstyle V(k,s):=\frac{1}{2\pi}\int_{-\pi}^{\pi} e^{-ikt} Q_n(e^{it})(y_-(t;\theta^{\mathrm{h}}\!,\theta^{\mathrm{v}}))^sdt
\]
is uniformly bounded and solves the problem $\mathrm{[P^{sym}_n]}$. Moreover,
\begin{equation}
\label{eq:<Qn>=L-sym}
\langle Q_n,1\rangle_{\frac{w}{2\pi}dt} =L_{n+1} \quad\text{and}\quad \langle Q_n,e^{int}\rangle_{\frac{w}{2\pi}dt}=L_{n+1}^\star,
\end{equation}
where the scalar product is taken with respect to the same measure on the unit circle.
\end{lem}
\begin{proof} The values~$V(k,s)$ are uniformly bounded as~$0\le y_-\le 1$, the massive harmonicity property~$[\Delta^{(m)}V](k,s)=0$ for~$s\ge 1$ is straightforward and the required properties of the values~$V(k,0)$ and~$[\mathcal{N}F](k,0)$ follow from the assumptions made on the polynomial~$Q_n$. The identities~\eqref{eq:Vsym-laplacians} give~\eqref{eq:<Qn>=L-sym}.
\end{proof}

A similar construction can be done for the problem~$\mathrm{[P^{anti}_{n+1}]}$, {see Fig.~\ref{fig:sym-anti-bvp}C.} The only difference is that at level~$s=0$ we now require that
$\widehat V_0(e^{it})$ does not contain monomials~$e^{it},\dots,e^{i(n+1)t}$ while
\[%\begin{equation}\label{eq:wV=poly-anti}
\textstyle \sum_{k\in\mathbb Z}[\mathcal{N}V](k,0)e^{ikt}\ =\ w(t;\theta^{\mathrm{h}}\!,\theta^{\mathrm{v}})\widehat V_0(e^{it})\ =\ -L_n+\ldots+L_n^\star e^{i(n+1)t}
\]%\end{equation}
is a trigonometric polynomial of degree~$n\!+\!1$. In other words, this polynomial is orthogonal to~$\{e^{it},\dots,e^{int}\}$ with respect to the weight
\begin{equation}
\label{eq:whash-def}
w^\# (t;\theta^{\mathrm{h}}\!,\theta^{\mathrm{v}})\; :=\; (w(t;\theta^{\mathrm{h}}\!,\theta^{\mathrm{v}}))^{-1},\qquad t\in [0,2\pi].
\end{equation}
provided that~$w^\#$ is integrable on the unit circle. One can easily see from~\eqref{eq:w-def} that this is true if and only if~$\theta^{\mathrm{h}}+\theta^{\mathrm{v}}\ne \frac{\pi}{2}$. We discuss a modification of the next claim required for the analysis of the critical case~$\theta^{\mathrm{h}}+\theta^{\mathrm{v}}=\frac{\pi}{2}$ in Section~\ref{subsect:crit-homogen}.
\begin{lem}%[Solution to the problem~$\mathrm{[P^{anti}_{n+1}]}$]
\label{lem:solution-anti}
Let~$n\ge 0$ and assume that~$\theta^{\mathrm{h}}+\theta^{\mathrm{v}}\ne \frac{\pi}{2}$. If a trigonometric polynomial $Q^\#_{n+1}(e^{it})=-L_n+\ldots+L_n^{\star}e^{i(n+1)t}$ is orthogonal to the family $\{e^{it},\ldots,e^{int}\}$ with respect to the measure~$w^\#\!(t;\theta^{\mathrm{h}}\!,\theta^{\mathrm{v}})\frac{dt}{2\pi}$, then the function
\begin{equation}
\label{eq:solution-anti}
\textstyle V(k,s):=\frac{1}{2\pi}\int_{-\pi}^{\pi} e^{-ikt} Q^\#_{n+1}(e^{it})(y_-(t;\theta^{\mathrm{h}}\!,\theta^{\mathrm{v}}))^s w^\#\!(t;\theta^{\mathrm{h}}\!,\theta^{\mathrm{v}})dt
\end{equation}
is uniformly bounded and solves the problem $\mathrm{[P^{anti}_{n+1}]}$. Moreover,
\begin{equation}
\label{eq:<Qn>=L-anti}
\langle Q^\#_{n+1},1\rangle_{\frac{w^\#}{2\pi}dt} = -D_{n+1} \quad\text{and}\quad \langle Q^\#_{n+1},e^{i(n+1)t}\rangle_{\frac{w^\#}{2\pi}dt}= D_{n+1}^\star,
\end{equation}
where the scalar product is taken with respect to the same measure on the unit circle.
\end{lem}
\begin{proof}
The proof repeats the arguments used in the proof of Lemma~\ref{lem:solution-sym}.
\end{proof}

\subsection{Horizontal spin-spin correlations below criticality} \label{subsect:below-crit-homogen}
In this section we combine the results of Lemmas~\ref{lem:solution-sym} and~\ref{lem:solution-anti} into a single result on asymptotics of the horizontal spin-spin correlations~$D_n$ as~$n\to \infty$. We assume that~$\theta^{\mathrm{h}}+\theta^{\mathrm v}<\frac{\pi}{2}$ and rely upon the fact that~$D_n^{\star}\to 0$ as~$n\to \infty$. This can be easily derived from the monotonicity of~$D_n$ with respect to the temperature and the fact that~$D_n=D_n^\star\to 0$ as~$n\to\infty$ in the critical regime \mbox{$\theta^{\mathrm{h}}+\theta^{\mathrm{v}}=\frac{\pi}{2}$} which is discussed in the next section.

\begin{theo}[{\bf Kauffman--Onsager--Yang}] \label{thm:KOY}
Let~$\theta^{\mathrm{h}}+\theta^{\mathrm{v}}<\frac{\pi}{2}$. Then, the spontaneous magnetization~$\mathcal{M}(\theta^\mathrm{h},\theta^\mathrm{v})$ of the homogeneous Ising model is given by
\begin{equation}
\label{eq:M=(KOY-thm)}
\mathcal{M}(\theta^\mathrm{h}\!,\theta^\mathrm{v})\ := \lim\limits_{n\rightarrow \infty} D_n^{1/2}\ =\ \big[1-(\tan\theta^{\mathrm{h}}\tan\theta^{\mathrm{v}})^2\big]^{1/8}.
%\ =\ \big[1-(\sinh(2\beta J^\mathrm{h})\sinh(2\beta J^\mathrm{v}))^{-2}\big]^{1/8}.
\end{equation}
(Note that under the parametrization~\eqref{eq:parametrization-model} one has $\tan\theta_e=(\sinh(2\beta J_e))^{-1}$.)
\end{theo}
{\begin{rem}
It is worth mentioning that the value~$\tan\theta^\mathrm{h}\tan\theta^\mathrm{v}$ also admits a fully \emph{geometric} interpretation as Baxter's elliptic parameter of the Z-invariant Ising model on isoradial graphs~\cite[Eq.~(7.10.50)]{baxter-book}, see Section~\ref{subsect:isoradial} for details.
\end{rem}}
\begin{proof} Classically, the computation given below is based upon the strong Szeg\"o theorem on the asymptotics of the norms of orthogonal polynomials on the unit circle. Note however that we use this result in its simplest form, for \emph{real} weights~$w$ and~$w^\#$ given by~\eqref{eq:w-def} and~\eqref{eq:whash-def}.

Let $\Phi_n(z)=z^n+\ldots -\alpha_{n-1}$ be the $n$-th monic orthogonal polynomial on the unit circle with respect to the measure~$w(t;\theta^{\mathrm{h}}\!,\theta^{\mathrm{v}})\frac{dt}{2\pi}$, the real number~$\alpha_{n-1}$ is called the \emph{Verblunsky coefficient}, recall that~$|\alpha_{n-1}|<1$ for all~$n\ge 1$. Denote by~$\Phi_n^* :=z^{n}\Phi_n(z^{-1})=-\alpha_{n-1}z^n+\ldots+1$ the reciprocal polynomial. Matching the free and the leading coefficients, it is easy to see that the polynomial~$Q_n$ from Lemma~\ref{lem:solution-sym} can be written as
\[
Q_n(e^{it})=c_n\Phi_n(e^{it})+c_n^*\Phi_n^*(e^{it}),\ \ \text{where}\ \ \left[ \begin{array}{c} c_n^* \\ c_n \end{array} \right]= \left[ \begin{array}{cc} 1 & -\alpha_{n-1} \\ -\alpha_{n-1} & 1 \end{array}\right]^{-1}\left[\begin{array}{c} D_n^\star \\ D_n \end{array} \right]\!.
\]
Moreover, one has~$\langle\Phi_n,e^{int}\rangle=\langle\Phi^*_n,1\rangle=\|\Phi_n\|^2=:\beta_n=\beta_0\prod_{k=1}^n(1-\alpha_{k-1}^2)$ (e.g.,~see~\cite[Theorem~2.1]{simon-OPUC-one-foot}) and $\langle \Phi_n,1\rangle= \langle \Phi^*_n,e^{int}\rangle=0$, here and below we drop the measure~$w\frac{dt}{2\pi}$ from the notation for shortness. Therefore, the identities~\eqref{eq:Vsym-laplacians} imply that
\begin{equation}
\label{eq:Dn-to-Ln+1}
\left[ \begin{array}{c} L_{n+1}^\star \\ L_{n+1} \end{array} \right]= \beta_n\left[ \begin{array}{c} c_n^* \\ c_n \end{array} \right] =\beta_{n-1} \left[ \begin{array}{cc} 1 & \alpha_{n-1} \\ \alpha_{n-1} & 1 \end{array}\right]\;\left[\begin{array}{c} D_n^\star \\ D_n \end{array} \right],
\end{equation}
and hence
\begin{equation}
\label{eq:Dn-to-Ln+1-squares}
L_{n+1}^2-(L_{n+1}^\star)^2\ =\ \beta_n\beta_{n-1}\cdot (D_n^2-(D_n^\star)^2)\quad \text{for}\ \ n\ge 1.
\end{equation}

Similarly, it follows from Lemma~\ref{lem:solution-anti} that
\begin{equation}
\label{eq:Ln-to-Dn+1}
\left[ \begin{array}{c} D_{n+1}^\star \\ -D_{n+1} \end{array} \right]\ =\ \beta^\#_n\left[ \begin{array}{cc} 1 & \alpha^\#_n \\ \alpha^\#_n & 1 \end{array}\right]\;\left[\begin{array}{c} L_n^\star \\ -L_n \end{array} \right],
\end{equation}
where~$\alpha^\#_n$ and~$\beta^\#_n$ stand for the Verblunsky coefficients and squared norms of monic orthogonal polynomials corresponding to the weight~\eqref{eq:whash-def}. In particular, we have
\begin{equation}
\label{eq:Ln-to-Dn+1-squares}
D_{n+1}^2-(D_{n+1}^\star)^2\ =\ \beta^\#_{n+1}\beta^\#_n\cdot (L_n^2-(L_n^\star)^2)\quad \text{for}\ \ n\ge 0.
\end{equation}
The recurrence relations~\eqref{eq:Ln-to-Dn+1-squares}, \eqref{eq:Dn-to-Ln+1-squares} applied for even and odd indices~$n$, respectively, lead to the formula
\begin{align*}
\textstyle D_{2m+1}^2-(D_{2m+1}^\star)^2\ &=\ \beta^\#_{2m+1}\beta^\#_{2m}\cdot\beta_{2m-1}\beta_{2m-2}\cdot (D_{2m-1}^2-(D_{2m-1}^\star)^2)\\
&\textstyle =\ \dots\ =\ \prod_{k=0}^{2m+1}\beta^\#_k\cdot\prod_{k=0}^{2m-1}\beta_k\cdot (L_0^2-(L_0^*)^2)\,,
\end{align*}
note that~$L_0^2-(L_0^*)^2=(\cos\theta^\mathrm{v})^2-(\sin\theta^\mathrm{h})^2= \cos(\theta^\mathrm{h}\!+\theta^\mathrm{v})\cos(\theta^\mathrm{h}\!-\theta^\mathrm{v})$.

Recall that~$D^\star_{2m+1}\to 0$ as~$m\to \infty$. It remains to apply the Szeg\"o theory (e.g., see \cite[Section~5.5]{szego-grenander-book} or~\cite[Theorems~8.1~and~8.5]{simon-OPUC-one-foot}) to the weights~\eqref{eq:w-def} and~\eqref{eq:whash-def}. %Let~$a:=\sin\theta^\mathrm{h}\cos\theta^\mathrm{v}$ and~$b=\cos\theta^\mathrm{h}\sin\theta^\mathrm{v}$.
A straightforward computation shows that
\begin{align*}
w(t;\theta^\mathrm{h}\!,\theta^\mathrm{v})\ =\ Cw_{q_-}(t)w_{q_+}(t),\qquad & \text{where}\ \ C=(\cos \tfrac{1}{2}\theta^\mathrm{h})^2\cos\theta^\mathrm{v},\\ w_q(t):=\big[(1\!+\!q^2)^2-(2q\cos\tfrac{t}{2})^2\big]^{1/2}\!=|1\!-\!q^2e^{it}|,\qquad  &q_{\pm}^2=\tan(\tfrac{1}{2}\theta^\mathrm{h})\tan(\tfrac{\pi}{4}\mp\tfrac{1}{2}\theta^\mathrm{v})\,.
\end{align*}
Since~$w^\#(t;\theta^\mathrm{h}\!,\theta^\mathrm{v})=(w(t;\theta^\mathrm{h}\!,\theta^\mathrm{v}))^{-1}$, we have
\[
\textstyle \lim_{m\to\infty} \prod_{k=0}^{2m+1}\beta^\#_k\cdot\prod_{k=0}^{2m-1}\beta_k\ =\ C^{-2}\cdot G^2,
\]
where
\begin{align*}
G\ &=\ \exp\Big[\frac{1}{4\pi} \iint_{\mathbb D} \Big|\frac{d}{dz} \big(\log(1\!-\!q_-^2 z)+ \log(1-q_+^2 z) \Big|^2 \mathrm{d}^2z \Big] \\ %\exp\Big[\frac{1}{4\pi} |\sum\limits_{k\geq 1}  (q_1^{2k}+ q_2^{2k})z^{k-1}|^2 dA(z) \Big]
&\textstyle =\ \exp \big[-\sum_{k\geq 1} \tfrac{1}{4k} (q_-^{2k} +q_+^{2k})^2 \big]\ =\ \big[(1-q_-^4)(1-q_+^4)(1-q_-^2 q_+^2)^2\big]^{-1/4}\\
&=\ {(\cos\tfrac{1}{2}\theta^\mathrm{h})^2(\cos\theta^\mathrm{v})^{1/2}(\cos\theta^\mathrm{h})^{-1/2}} {(\cos(\theta^\mathrm{h}\!+\theta^\mathrm{v})\cos(\theta^\mathrm{h}\!-\theta^\mathrm{v}))^{-1/4}}
\end{align*}
{(in the first line, $\mathrm{d}^2z$ denotes the area measure in the unit disc $\mathbb{D}$).} Putting all the factors together, one gets~\eqref{eq:M=(KOY-thm)}.
\end{proof}

\begin{rem}
The identity~\eqref{eq:Ln-to-Dn+1} with~$n=0$ also provides a formula
\[
D_1=\beta_0^\#\cdot [\cos\theta^\mathrm{v}-\alpha_0^\#\sin\theta^\mathrm{h}]
%,\qquad D^\star_1=\beta_0^\#\cdot [\sin\theta^\mathrm{v}-\alpha_0^\#\cos\theta^\mathrm{h}]
\]
for the \emph{energy density} (on a vertical edge) of the homogeneous Ising model.
\end{rem}

\subsection{Asymptotics of horizontal correlations~$\bm{D_n}$ as~$\bm{n\to\infty}$ at criticality} \label{subsect:crit-homogen}
Assume now that~$\theta^\mathrm{h}+\theta^\mathrm{v}=\frac{\pi}{2}$. Another classical result that we discuss in this section is that spin-spin correlations~$D_m$ decay like~$m^{-1/4}$ at large distances.
\begin{theo}[{\bf McCoy--Wu}] \label{thm:crit-homogen}
Let~$\mathcal{C}_\sigma:=2^{\frac{1}{6}}e^{\frac{3}{2}\zeta'(-1)}$, $\theta^\mathrm{h}=\theta$ and~$\theta^\mathrm{v}=\frac{\pi}{2}-\theta$. Then,
\begin{equation}
\label{eq:crit-homogen}
D_m\ \sim\ \mathcal{C}_\sigma^2\cdot (2m\cos\theta)^{-1/4}\quad \text{as}\ \ m\to\infty.
\end{equation}
\end{theo}

\begin{proof}
A straightforward computation shows that
\[%\begin{equation}\label{eq:w-real}
w(t;\theta,\tfrac{\pi}{2}-\theta)=2\sin\theta\cdot [1-(\sin\theta\cos\tfrac{1}{2}t)^2]^{1/2}\cdot |\sin\tfrac{1}{2}t|\,.
\]%\end{equation}
In particular, the weight~$w^\#:=w^{-1}$ is not integrable and the arguments used in the proof of Theorem~\ref{thm:KOY} require a modification. Also, the Kramers--Wannier duality ensures that~$D_n=D_n^\star$, $L_n=L_n^\star$ and hence the identities~\eqref{eq:Dn-to-Ln+1-squares},~\eqref{eq:Ln-to-Dn+1-squares} become useless (though one still could use~\eqref{eq:Dn-to-Ln+1}).
In this situation we prefer to switch to the framework of orthogonal polynomials on the \emph{real line} (more precisely, on the segment~$[-1,1]$) for computations. Let
\begin{equation}
\label{eq:w-real}
\overline{w}(x;\theta)\ :=\ [\,1-(\sin\theta\cdot x)^2\,]^{1/2}\,,\quad x\in [-1;1],
\end{equation}
and let~$P_n(x)=x^n+\dots$ be the monic orthogonal polynomial of degree~$n$ on~$[-1,1]$ with respect to the weight $\overline{w}(x,\theta)$. It is easy to check that the trigonometric polynomial
\[
Q_n(e^{it})\ :=\ D_n \cdot e^{\frac{1}{2}int}\cdot 2^nP_n(\cos\tfrac{1}{2}t)
\]
fits the construction given in Lemma~\ref{lem:solution-sym} to solve the problem~$[\mathrm{P^{sym}_n}]$. The formula~\eqref{eq:<Qn>=L-sym} gives
\begin{align}
L_{n+1}\ &=\ \frac{1}{2\pi}\int_{-\pi}^{\pi} Q_{n}(e^{it})w(t;\theta,\tfrac{\pi}{2}\!-\!\theta)dt \notag\\
&=\ \frac{D_n2^{n-1}}{\pi}\int_{-\pi}^{\pi} \cos(\tfrac{1}{2}nt)P_n(\cos\tfrac{1}{2}t)w(t;\theta,\tfrac{\pi}{2}\!-\!\theta)dt \notag\\
&=\ \frac{D_n2^{n+1}\sin\theta}{\pi}\int_{-1}^1 (2^{n-1}x^n+\dots)P_n(x)\overline{w}(x;\theta)dx \notag \\ &=\ \pi^{-1}2^{2n}\sin\theta\cdot\|P_n\|^2_{\overline{w}dx}\cdot D_n,\qquad n\ge 1.
\label{eq:Dn-to-Ln+1-real}
\end{align}
Moreover, a similar computation for~$n=0$ implies that
\begin{equation}
\label{eq:D0-to-L1-real}
\textstyle 2L_1\ =\ 2\pi^{-1}\sin\theta\int_{-1}^1P_0(x)\overline{w}(x;\theta)dx = 2\pi^{-1}\sin\theta\cdot\|P_0\|^2_{\overline{w}dx}
\end{equation}
since~$D_0=1$ and due to the modification required in Lemma~\ref{lem:laplacians} in the case~$n=0$.

We can use the same line of reasoning to construct a solution of the problem~$\mathrm{[P^{anti}_{n+1}]}$ treated in Lemma~\ref{lem:solution-anti} in the non-critical regime. Namely, let~$P_n^\#(x)$ be the monic orthogonal polynomial of degree~$n$ on~$[-1,1]$ with respect to the weight
\begin{equation}
\label{eq:whash-real}
\overline{w}^\#(x)\ :=\ [\,1-(\sin\theta\cdot x)^2\,]^{-1/2}\,,\quad x\in [-1,1],
\end{equation}
and
\[
Q^\#_{n+1}(e^{it})\ :=\ L_n\cdot (e^{it}\!-\!1)e^{\frac{1}{2}int}\cdot 2^nP^\#_n(\cos\tfrac{1}{2}t).
\]
It is straightforward to check that the formula~\eqref{eq:solution-anti} gives a solution to the boundary value problem~$\mathrm{[P^{anti}_{n+1}]}$, note that the product~$(e^{it}-1)w^\#(t;\theta,\tfrac{\pi}{2}-\theta)$ {is} integrable on the unit circle as the first factor kills the singularity of~$w^\#$ at~$t=0$. Moreover, the computation~\eqref{eq:<Qn>=L-anti} remains valid and reads as
\begin{align}
D_{n+1}\ &=\ -\frac{1}{2\pi}\int_{-\pi}^{\pi} Q^\#_{n+1}(e^{it})w^\#(t;\theta,\tfrac{\pi}{2}\!-\!\theta)^{-1}dt \notag\\
&=\ \frac{L_n2^{n}}{\pi}\int_{-\pi}^{\pi} \frac{\sin(\tfrac{1}{2}(n\!+\!1)t)}{\sin\tfrac{1}{2}t} P^\#_n(\cos\tfrac{1}{2}t)\frac{(\sin\frac{1}{2}t)^2dt}{ w(t;\theta,\tfrac{\pi}{2}\!-\!\theta)} \notag\\
&=\ \frac{L_n2^{n}}{\pi\sin\theta}\int_{-1}^1 (2^{n}x^n+\dots)P^\#_n(x)\overline{w}^\#(x)dx \notag \\ &=\ \pi^{-1}2^{2n}(\sin\theta)^{-1}\cdot \|P^\#_n\|^2_{\overline{w}^\#dx}\cdot L_n,\quad n\ge 0.
\label{eq:Ln-to-Dn+1-real}
\end{align}

Recall that~$L_0=\sin\theta$ (see Lemma~\ref{lem:laplacians}). Taking a product of the recurrence relations \eqref{eq:D0-to-L1-real}, \eqref{eq:Dn-to-Ln+1-real} for~$n=1,\dots,m\!-\!1$, and~\eqref{eq:Ln-to-Dn+1-real} for~$n=0,\dots, m$, one obtains the identity
\begin{align*}%\label{eq:DmDm+1-crit}
D_{m+1}D_{m}\ &=\ \pi^{-2m-1} 2^{2m^2}\prod\nolimits_{k=0}^{m-1}\|P_k\|^2_{\overline{w}dx} \cdot \prod\nolimits_{k=0}^{m}\|P^\#_k\|^2_{\overline{w}^\#dx}\,,
%\\ &=\textstyle\ (2\pi)^{-2m-1} 2^{2m(m+1)}\cdot \det\big[\,\int_{-1}^1x^{p+q}w(x;\theta)dx\,\big]_{p,q=0}^{m-1}\cdot \det\big[\,\int_{-1}^1x^{p+q}w^\#(x;\theta)dx\,\big]_{p,q=0}^{m}\,.
\end{align*}
where the weights~$w(x;\theta)$ and~$w^\#(x;\theta)$ on~$[-1,1]$ are given by~\eqref{eq:w-real} and \eqref{eq:whash-real}.

This is again a classical setup of the orthogonal polynomials theory, note that, if one now passes back to the unit circle, then the $|t|$-type singularity of the weights appear at the point~$e^{it}=1$. One might now use the general results (summarized, e.g., in~\cite{deift-its-krasovsky-annals}) but we prefer to refer to a specific treatment~\cite{basor-et-al-2015}. Applying~\cite[Theorem~1.7]{basor-et-al-2015} with parameters~$\alpha=0,\beta=\pm \tfrac{1}{2}$ and~$k=\sin\theta$ one obtains the asymptotics
\[
D_{m+1}D_{m}\ \sim\ \pi[G(\tfrac{1}{2})]^4(1-k^2)^{-1/4}m^{-1/2}\ \sim\ 2^{2/3}e^{6\zeta'(-1)}(2m\cos\theta)^{-1/2},\quad m\to\infty,
\]
where~$G$ denotes the Barnes G-function. (Note that~\cite{basor-et-al-2015} also provides sub-leading terms of this asymptotics.) The proof of~\eqref{eq:crit-homogen} is complete modulo the fact that~$D_{m+1}\sim D_m$ as~$m\to\infty$. This statement can be proved by the arguments given in the next remark (or, alternatively, using probabilistic estimates).
\end{proof}

\begin{rem}\label{rem:quadratic} Due to the famous quadratic identities~\cite{Perk-81dubna,mccoy-perk-wu-quadratic} for the spin-spin correlations, one can write~\eqref{eq:Dn-to-Ln+1-real} and~\eqref{eq:Ln-to-Dn+1-real} as
\begin{align*}
A_n\ :=\ \pi^{-1}2^{2n}\|P_n\|^2_{\overline{w}dx}=\frac{D_{n+1}\!+\!\cos\theta\cdot\widetilde{D}_{n+1}}{D_n}= \frac{D_{n+2}}{D_{n+1}\!-\!\cos\theta\cdot\widetilde{D}_{n+1}}\,,\\
B_{n+1}\ :=\ \pi^{-1}2^{2n+2}\|P^\#_{n+1}\|^2_{\overline{w}^\#dx}=\frac{D_{n+2}}{D_{n+1}\!+\!\cos\theta\cdot\widetilde{D}_{n+1}}= \frac{D_{n+1}\!-\!\cos\theta\cdot\widetilde{D}_{n+1}}{D_n}.
\end{align*}
In fact, one can also prove these identities by considering the anti-symmetrization (resp., symmetrization) of the observable~$X_{[\mathbf{u},\mathbf{v}]}$ on the north-west (resp., north-east) corners of the lattice and noticing that, up to a multiplicative constant, it solves the problem~$\mathrm{[P^{anti}_{n+2}]}$ (resp., $\mathrm{[P^{sym}_{n-1}]}$). In particular, we have
\[
D_{m+1}/D_m\ =\ \tfrac{1}{2}(A_m+B_{m+1})\ =\ 2(A_{m-1}^{-1}+B_m^{-1})^{-1}
\]
so one can see that~$D_{m+1}\sim D_m$ and find sub-leading corrections to the asymptotics of~$D_m$ (and~$\widetilde{D}_m$) using the analysis of orthogonal polynomials performed in~\cite{basor-et-al-2015}.
\end{rem}

\section{Layered model in the zig-zag half-plane}\label{sc:layered}

In this section we work with the (half-)infinite volume limit of the Ising model on the zig-zag half-plane~$\mathbb H^\diamond$ (see Fig.~\ref{fig:GlobalLayered} for the notation), which is defined as a limit of probability measures on an increasing sequence of finite domains exhausting~$\mathbb H^\diamond$, with~`$+$' boundary conditions at the right-most column~$\mathrm{C}_0$ and at infinity. All interaction parameters between the columns~$\mathrm{C}_{p-1}$ and~$\mathrm{C}_p$ are assumed to be the same and equal to~$x_p=\exp[-2\beta J_p]=\tan\tfrac{1}{2}\theta_p$. The goal is to find a representation for the magnetization~$M_m$ at the column $\mathrm{C}_{2m}$, see~\eqref{eq:Mm-def}. The uniqueness of the relevant half-plane fermionic observable is discussed in Section~\ref{subsect:observable-layered} and our main result -- Theorem~\ref{thm:layered} -- is proved in Section~\ref{subsect:Mm-layered}. In Section~\ref{subsect:wetting} we use Theorem~\ref{thm:layered} to discuss the {\emph{wetting phase transition}~\cite{frohlich-pfister-87,pfister-velenik-96}} caused by a boundary magnetic field. In this case the Jacobi matrix~$J$ can be explicitly diagonalized and the final answer can be written in terms of the so-called Toeplitz+Hankel determinants.

\subsection{Half-plane fermionic observable} \label{subsect:observable-layered}
Let~$\mathbf v=(-2m\!-\!\tfrac{3}{2},0)$. Below we work with the fermionic observable~$X_{[\mathbf v]}$ defined by~\eqref{eq:chi:=mu_sigma_def}; comparing with Section~\ref{sc:homogeneous} one can think about the spin~$\sigma_\mathbf u:=\sigma_{\mathrm{out}}$ as being attached to the vertical boundary. We are mostly interested in the values of~$X_{[\mathbf v]}$ at west corners (see Fig.~\ref{fig:GlobalLayered})
\[
\ H(-k,s)\ :=\ \Psi_{[\mathbf v]}((-k,s))= X_{[\mathbf v]}((-k,s)),\quad k\in\mathbb N_0,\ s\in\mathbb Z,\ k+s\not\in 2\mathbb Z,
\]
note the convention on~$\eta_c$ chosen in~\eqref{eq:Dirac_spinor}. By definition, one has
\begin{equation}
\label{eq:V=Mm-layered}
H(-2m\!-\!1,0)\ =\ \mathbb E^+_{\mathbb H^\diamond}[\sigma_{(-2m-\frac{1}{2},0)}]\ =\ M_m.
\end{equation}
We also need the values of~$X_{[\mathbf v]}$ at east corners:
\[
H^\circ(-k,s)\ :=\ \Psi_{[\mathbf v]}((-k,s))=iX_{[\mathbf v]}((-k,s)),\quad k\in\mathbb N,\ s\in\mathbb Z,\ k+s\in 2\mathbb Z.
\]
It is convenient to set~$\theta_0:=0$ and~$H^\circ(0,s):=0$ for all~$s\in 2\mathbb Z$.

The infinite-volume observable~$X_{[\mathbf v]}$ is defined as a (subsequential) limit of the same observables constructed in finite regions. Subsequential limits exist due to the uniform bound~\eqref{eq:X-bound-onG} while the uniqueness of~$X_{[\mathbf v]}$ is given by Lemma~\ref{lem:uniqueness-layered}. The discrete Cauchy--Riemann identities~\eqref{eq:Cauchy-Riemann-general} can be written as
\begin{align}
\notag
& H(-k-1,s\pm 1)\sin\theta_{k+1}-H(-k,s)\cos\theta_k\\
&=\,\pm i\cdot[\,H^\circ(-k,s \pm 1)\sin\theta_k-H^\circ(-k\!-\!1,s)\cos\theta_{k+1}\,],\quad k\ge 1,\ k\!+\!s\not\in 2\mathbb Z.
\label{eq:CR-layered}
\end{align}
Near the vertical boundary, these equations should be modified as follows:
\begin{equation}
\label{eq:CR-boundary}
H(-1,s\pm 1)\sin\theta_1-H(0,s)\ =\ \mp i \cdot H^\circ(-1,s)\cos\theta_1, \quad s\not\in 2\mathbb Z.
\end{equation}
Indeed, $X_{[\mathbf v]}((-\frac{1}{2},s\pm\frac{1}{2}))=X_{[\mathbf v]}(0,s))=H(0,s)$ and hence~\eqref{eq:CR-boundary} are nothing but the three-term identities~\eqref{eq:propagation_on_corners}.

\begin{figure}
\centering{
\begin{tikzpicture}[scale=0.17]
\input{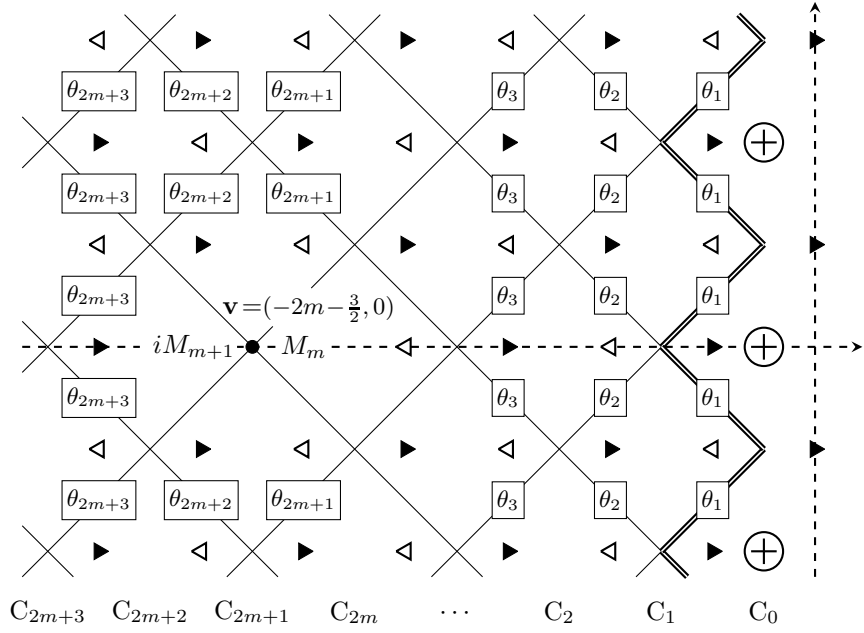}
\end{tikzpicture}}
\caption{The zig-zag layered model in the left half-plane~$\mathbb H^\diamond$. All the interaction parameters between two adjacent columns are assumed to be the same. The~`$+$' boundary conditions are imposed at the column~$\mathrm{C}_0$. To analyze the ratio~$M_{m+1}/M_m$ we consider the Kadanoff--Ceva fermionic observable branching at~$\mathbf v\!=\!(-2m-\frac{3}{2},0)$.}
\label{fig:GlobalLayered}
\end{figure}

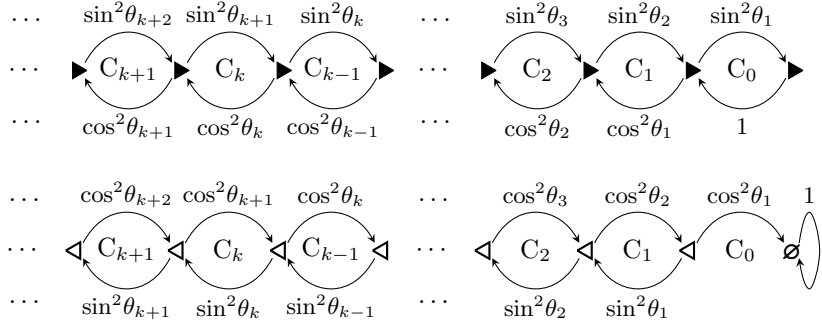
\begin{figure}
\centering{
\begin{tikzpicture}[scale=0.17]
\filldraw[fill=black,draw=black] ( 0.75000, 0.00000) -- (-0.37500, 0.64952) -- (-0.37500,-0.64952) -- ( 0.75000, 0.00000);
\node at (-4.00000, 0.00000) {$\mathrm{C}_0$};
\filldraw[fill=black,draw=black] (-7.25000, 0.00000) -- (-8.37500, 0.64952) -- (-8.37500,-0.64952) -- (-7.25000, 0.00000);
\node at (-12.00000, 0.00000) {$\mathrm{C}_1$};
\filldraw[fill=black,draw=black] (-15.25000, 0.00000) -- (-16.37500, 0.64952) -- (-16.37500,-0.64952) -- (-15.25000, 0.00000);
\node at (-20.00000, 0.00000) {$\mathrm{C}_2$};
\filldraw[fill=black,draw=black] (-23.25000, 0.00000) -- (-24.37500, 0.64952) -- (-24.37500,-0.64952) -- (-23.25000, 0.00000);
\node at (-28.00000, 0.00000) {$\cdots$};
\filldraw[fill=black,draw=black] (-31.25000, 0.00000) -- (-32.37500, 0.64952) -- (-32.37500,-0.64952) -- (-31.25000, 0.00000);
\node at (-36.00000, 0.00000) {$\mathrm{C}_{k-1}$};
\filldraw[fill=black,draw=black] (-39.25000, 0.00000) -- (-40.37500, 0.64952) -- (-40.37500,-0.64952) -- (-39.25000, 0.00000);
\node at (-44.00000, 0.00000) {$\mathrm{C}_k$};
\filldraw[fill=black,draw=black] (-47.25000, 0.00000) -- (-48.37500, 0.64952) -- (-48.37500,-0.64952) -- (-47.25000, 0.00000);
\node at (-52.00000, 0.00000) {$\mathrm{C}_{k+1}$};
\filldraw[fill=black,draw=black] (-55.25000, 0.00000) -- (-56.37500, 0.64952) -- (-56.37500,-0.64952) -- (-55.25000, 0.00000);
\node at (-60.00000, 0.00000) {$\cdots$};
\draw[->,>=stealth] (-7.40000, 0.60000) .. controls (-6.60000, 2.1000000000000000E+0000) and (-5.40000, 3.00000000000000E+000) .. (-4.00000, 3.00000) .. controls (-2.60000, 3.00000000000000E+000) and (-1.40000, 2.1000000000000000E+0000) .. (-0.60000, 6.0000000000000000E-0001);
\draw[->,>=stealth] (-0.60000,-0.60000) .. controls (-1.40000,-2.1000000000000000E+0000) and (-2.60000,-3.00000000000000E+000) .. (-4.00000,-3.00000) .. controls (-5.40000,-3.00000000000000E+000) and (-6.60000,-2.1000000000000000E+0000) .. (-7.40000,-6.0000000000000000E-0001);
\node at (-4.00000, 4.40000) {{\small $\sin^2\!\theta_1$}};
\node at (-4.00000,-4.40000) {{\small $1$}};
\draw[->,>=stealth] (-15.40000, 0.60000) .. controls (-14.60000, 2.1000000000000000E+0000) and (-13.40000, 3.00000000000000E+000) .. (-12.00000, 3.00000) .. controls (-10.60000, 3.00000000000000E+000) and (-9.40000, 2.1000000000000000E+0000) .. (-8.60000, 6.0000000000000000E-0001);
\draw[->,>=stealth] (-8.60000,-0.60000) .. controls (-9.40000,-2.1000000000000000E+0000) and (-10.60000,-3.00000000000000E+000) .. (-12.00000,-3.00000) .. controls (-13.40000,-3.00000000000000E+000) and (-14.60000,-2.1000000000000000E+0000) .. (-15.40000,-6.0000000000000000E-0001);
\node at (-12.00000, 4.40000) {{\small $\sin^2\!\theta_2$}};
\node at (-12.00000,-4.40000) {{\small $\cos^2\!\theta_1$}};
\draw[->,>=stealth] (-23.40000, 0.60000) .. controls (-22.60000, 2.1000000000000000E+0000) and (-21.40000, 3.00000000000000E+000) .. (-20.00000, 3.00000) .. controls (-18.60000, 3.00000000000000E+000) and (-17.40000, 2.1000000000000000E+0000) .. (-16.60000, 6.0000000000000000E-0001);
\draw[->,>=stealth] (-16.60000,-0.60000) .. controls (-17.40000,-2.1000000000000000E+0000) and (-18.60000,-3.00000000000000E+000) .. (-20.00000,-3.00000) .. controls (-21.40000,-3.00000000000000E+000) and (-22.60000,-2.1000000000000000E+0000) .. (-23.40000,-6.0000000000000000E-0001);
\node at (-20.00000, 4.40000) {{\small $\sin^2\!\theta_3$}};
\node at (-20.00000,-4.40000) {{\small $\cos^2\!\theta_2$}};
\node at (-28.00000, 4.00000) {$\ldots$};
\node at (-28.00000,-4.00000) {$\ldots$};
\draw[->,>=stealth] (-39.40000, 0.60000) .. controls (-38.60000, 2.1000000000000000E+0000) and (-37.40000, 3.00000000000000E+000) .. (-36.00000, 3.00000) .. controls (-34.60000, 3.00000000000000E+000) and (-33.40000, 2.1000000000000000E+0000) .. (-32.60000, 6.0000000000000000E-0001);
\draw[->,>=stealth] (-32.60000,-0.60000) .. controls (-33.40000,-2.1000000000000000E+0000) and (-34.60000,-3.00000000000000E+000) .. (-36.00000,-3.00000) .. controls (-37.40000,-3.00000000000000E+000) and (-38.60000,-2.1000000000000000E+0000) .. (-39.40000,-6.0000000000000000E-0001);
\node at (-36.00000, 4.40000) {{\small $\sin^2\!\theta_{k}$}};
\node at (-36.00000,-4.40000) {{\small $\cos^2\!\theta_{k-1}$}};
\draw[->,>=stealth] (-47.40000, 0.60000) .. controls (-46.60000, 2.1000000000000000E+0000) and (-45.40000, 3.00000000000000E+000) .. (-44.00000, 3.00000) .. controls (-42.60000, 3.00000000000000E+000) and (-41.40000, 2.1000000000000000E+0000) .. (-40.60000, 6.0000000000000000E-0001);
\draw[->,>=stealth] (-40.60000,-0.60000) .. controls (-41.40000,-2.1000000000000000E+0000) and (-42.60000,-3.00000000000000E+000) .. (-44.00000,-3.00000) .. controls (-45.40000,-3.00000000000000E+000) and (-46.60000,-2.1000000000000000E+0000) .. (-47.40000,-6.0000000000000000E-0001);
\node at (-43.50000, 4.40000) {{\small $\sin^2\!\theta_{k+1}$}};
\node at (-44.00000,-4.40000) {{\small $\cos^2\!\theta_{k}$}};
\draw[->,>=stealth] (-55.40000, 0.60000) .. controls (-54.60000, 2.1000000000000000E+0000) and (-53.40000, 3.00000000000000E+000) .. (-52.00000, 3.00000) .. controls (-50.60000, 3.00000000000000E+000) and (-49.40000, 2.1000000000000000E+0000) .. (-48.60000, 6.0000000000000000E-0001);
\draw[->,>=stealth] (-48.60000,-0.60000) .. controls (-49.40000,-2.1000000000000000E+0000) and (-50.60000,-3.00000000000000E+000) .. (-52.00000,-3.00000) .. controls (-53.40000,-3.00000000000000E+000) and (-54.60000,-2.1000000000000000E+0000) .. (-55.40000,-6.0000000000000000E-0001);
\node at (-52.25000, 4.40000) {{\small $\sin^2\!\theta_{k+2}$}};
\node at (-52.00000,-4.40000) {{\small $\cos^2\!\theta_{k+1}$}};
\node at (-60.00000, 4.00000) {$\ldots$};
\node at (-60.00000,-4.00000) {$\ldots$};
\draw[->,>=stealth,white] ( 0.60000, 0.60000).. controls ( 0.80000, 2.00000) and ( 1.10000, 3.00000) .. ( 1.40000, 3.00000).. controls ( 1.70000, 3.00000) and ( 2.20000, 1.20000) .. ( 2.20000,0.0).. controls ( 2.20000,-1.20000) and ( 1.70000,-3.00000) .. ( 1.40000,-3.00000).. controls ( 1.10000,-3.00000) and ( 0.80000,-2.00000) .. ( 0.60000,-0.60000);
\end{tikzpicture}

\bigskip

\begin{tikzpicture}[scale=0.17]
\node at ( 0.00000,-0.20000) {{\small $\bm{\varnothing}$}};
\node at (-4.00000, 0.00000) {$\mathrm{C}_0$};
\filldraw[fill=white,draw=black, thick] (-8.75000, 0.00000) -- (-7.62500,-0.64952) -- (-7.62500, 0.64952) -- (-8.75000, 0.00000);
\node at (-12.00000, 0.00000) {$\mathrm{C}_1$};
\filldraw[fill=white,draw=black, thick] (-16.75000, 0.00000) -- (-15.62500,-0.64952) -- (-15.62500, 0.64952) -- (-16.75000, 0.00000);
\node at (-20.00000, 0.00000) {$\mathrm{C}_2$};
\filldraw[fill=white,draw=black, thick] (-24.75000, 0.00000) -- (-23.62500,-0.64952) -- (-23.62500, 0.64952) -- (-24.75000, 0.00000);
\node at (-28.00000, 0.00000) {$\cdots$};
\filldraw[fill=white,draw=black, thick] (-32.75000, 0.00000) -- (-31.62500,-0.64952) -- (-31.62500, 0.64952) -- (-32.75000, 0.00000);
\node at (-36.00000, 0.00000) {$\mathrm{C}_{k-1}$};
\filldraw[fill=white,draw=black, thick] (-40.75000, 0.00000) -- (-39.62500,-0.64952) -- (-39.62500, 0.64952) -- (-40.75000, 0.00000);
\node at (-44.00000, 0.00000) {$\mathrm{C}_k$};
\filldraw[fill=white,draw=black, thick] (-48.75000, 0.00000) -- (-47.62500,-0.64952) -- (-47.62500, 0.64952) -- (-48.75000, 0.00000);
\node at (-52.00000, 0.00000) {$\mathrm{C}_{k+1}$};
\filldraw[fill=white,draw=black, thick] (-56.75000, 0.00000) -- (-55.62500,-0.64952) -- (-55.62500, 0.64952) -- (-56.75000, 0.00000);
\node at (-60.00000, 0.00000) {$\cdots$};
\draw[->,>=stealth] (-7.40000, 0.60000) .. controls (-6.60000, 2.1000000000000000E+0000) and (-5.40000, 3.00000000000000E+000) .. (-4.00000, 3.00000) .. controls (-2.60000, 3.00000000000000E+000) and (-1.40000, 2.1000000000000000E+0000) .. (-0.60000, 6.0000000000000000E-0001);
\node at (-4.00000, 4.40000) {{\small $\cos^2\!\theta_1$}};
\node at (-4.00000,-4.40000) {{\small $$}};
\draw[->,>=stealth] (-15.40000, 0.60000) .. controls (-14.60000, 2.1000000000000000E+0000) and (-13.40000, 3.00000000000000E+000) .. (-12.00000, 3.00000) .. controls (-10.60000, 3.00000000000000E+000) and (-9.40000, 2.1000000000000000E+0000) .. (-8.60000, 6.0000000000000000E-0001);
\draw[->,>=stealth] (-8.60000,-0.60000) .. controls (-9.40000,-2.1000000000000000E+0000) and (-10.60000,-3.00000000000000E+000) .. (-12.00000,-3.00000) .. controls (-13.40000,-3.00000000000000E+000) and (-14.60000,-2.1000000000000000E+0000) .. (-15.40000,-6.0000000000000000E-0001);
\node at (-12.00000, 4.40000) {{\small $\cos^2\!\theta_2$}};
\node at (-12.00000,-4.40000) {{\small $\sin^2\!\theta_1$}};
\draw[->,>=stealth] (-23.40000, 0.60000) .. controls (-22.60000, 2.1000000000000000E+0000) and (-21.40000, 3.00000000000000E+000) .. (-20.00000, 3.00000) .. controls (-18.60000, 3.00000000000000E+000) and (-17.40000, 2.1000000000000000E+0000) .. (-16.60000, 6.0000000000000000E-0001);
\draw[->,>=stealth] (-16.60000,-0.60000) .. controls (-17.40000,-2.1000000000000000E+0000) and (-18.60000,-3.00000000000000E+000) .. (-20.00000,-3.00000) .. controls (-21.40000,-3.00000000000000E+000) and (-22.60000,-2.1000000000000000E+0000) .. (-23.40000,-6.0000000000000000E-0001);
\node at (-20.00000, 4.40000) {{\small $\cos^2\!\theta_3$}};
\node at (-20.00000,-4.40000) {{\small $\sin^2\!\theta_2$}};
\node at (-28.00000, 4.00000) {$\ldots$};
\node at (-28.00000,-4.00000) {$\ldots$};
\draw[->,>=stealth] (-39.40000, 0.60000) .. controls (-38.60000, 2.1000000000000000E+0000) and (-37.40000, 3.00000000000000E+000) .. (-36.00000, 3.00000) .. controls (-34.60000, 3.00000000000000E+000) and (-33.40000, 2.1000000000000000E+0000) .. (-32.60000, 6.0000000000000000E-0001);
\draw[->,>=stealth] (-32.60000,-0.60000) .. controls (-33.40000,-2.1000000000000000E+0000) and (-34.60000,-3.00000000000000E+000) .. (-36.00000,-3.00000) .. controls (-37.40000,-3.00000000000000E+000) and (-38.60000,-2.1000000000000000E+0000) .. (-39.40000,-6.0000000000000000E-0001);
\node at (-36.00000, 4.40000) {{\small $\cos^2\!\theta_{k}$}};
\node at (-36.00000,-4.40000) {{\small $\sin^2\!\theta_{k-1}$}};
\draw[->,>=stealth] (-47.40000, 0.60000) .. controls (-46.60000, 2.1000000000000000E+0000) and (-45.40000, 3.00000000000000E+000) .. (-44.00000, 3.00000) .. controls (-42.60000, 3.00000000000000E+000) and (-41.40000, 2.1000000000000000E+0000) .. (-40.60000, 6.0000000000000000E-0001);
\draw[->,>=stealth] (-40.60000,-0.60000) .. controls (-41.40000,-2.1000000000000000E+0000) and (-42.60000,-3.00000000000000E+000) .. (-44.00000,-3.00000) .. controls (-45.40000,-3.00000000000000E+000) and (-46.60000,-2.1000000000000000E+0000) .. (-47.40000,-6.0000000000000000E-0001);
\node at (-43.50000, 4.40000) {{\small $\cos^2\!\theta_{k+1}$}};
\node at (-44.00000,-4.40000) {{\small $\sin^2\!\theta_{k}$}};
\draw[->,>=stealth] (-55.40000, 0.60000) .. controls (-54.60000, 2.1000000000000000E+0000) and (-53.40000, 3.00000000000000E+000) .. (-52.00000, 3.00000) .. controls (-50.60000, 3.00000000000000E+000) and (-49.40000, 2.1000000000000000E+0000) .. (-48.60000, 6.0000000000000000E-0001);
\draw[->,>=stealth] (-48.60000,-0.60000) .. controls (-49.40000,-2.1000000000000000E+0000) and (-50.60000,-3.00000000000000E+000) .. (-52.00000,-3.00000) .. controls (-53.40000,-3.00000000000000E+000) and (-54.60000,-2.1000000000000000E+0000) .. (-55.40000,-6.0000000000000000E-0001);
\node at (-52.25000, 4.40000) {{\small $\cos^2\!\theta_{k+2}$}};
\node at (-52.00000,-4.40000) {{\small $\sin^2\!\theta_{k+1}$}};
\node at (-60.00000, 4.00000) {$\ldots$};
\node at (-60.00000,-4.00000) {$\ldots$};
\node at ( 1.40000, 4.30000) {{\small $1$}};
\draw[->,>=stealth,black] ( 0.60000, 0.60000).. controls ( 0.80000, 2.00000) and ( 1.10000, 3.00000) .. ( 1.40000, 3.00000).. controls ( 1.70000, 3.00000) and ( 2.20000, 1.20000) .. ( 2.20000,0.0).. controls ( 2.20000,-1.20000) and ( 1.70000,-3.00000) .. ( 1.40000,-3.00000).. controls ( 1.10000,-3.00000) and ( 0.80000,-2.00000) .. ( 0.60000,-0.60000);
\end{tikzpicture}}

\caption{For appropriately chosen prefactors~$\varrho_k$ and~$\varrho^\circ_k$, the identities~\eqref{eq:H-Crelations},~\eqref{eq:Hstar-Crelations} (coming from Proposition~\ref{prop-harmonicity-layered}) can be written as the discrete harmonicity property of functions~$\varrho_k H(-k,s)$ and~$\varrho^\circ_k H^\circ(-k,s)$ with respect to random walks having the indicated transition probabilities in the horizontal direction (and~$\frac{1}{2}$ in the vertical one).
The first random walk (on~$\triangleright$) is reflected from the imaginary axis while the second (on~$\triangleleft$) is absorbed there.
}
\label{fig:layered-transition}
\end{figure}

\begin{lem}\label{lem:uniqueness-layered} The spinors~$H,H^\circ$ defined in~$\mathbb H^\diamond$ and branching over~$\mathbf v$ are uniquely determined by the following conditions: uniform boundedness, Cauchy--Riemann identities~\eqref{eq:CR-layered}, boundary relations~\eqref{eq:CR-boundary}, and the value~\eqref{eq:V=Mm-layered} of~$H$ near~$\mathbf v$.
\end{lem}
\begin{proof} Taking the difference of two solutions, assume that~$H,H^\circ$ are uniformly bounded, satisfy~\eqref{eq:CR-layered}, \eqref{eq:CR-boundary} and that~$H(-2m-1,0)=0$. Recall that Proposition~\ref{prop-harmonicity-layered} gives the harmonicity-type identity
\begin{align}
H(-k,s)\ &=\ \tfrac{1}{2}\sin\theta_{k+1}\cos\theta_k\cdot[H(-k\!-\!1,s\!+\!1)+H(-k\!-\!1,s\!-\!1)] \notag \\
&+\ \tfrac{1}{2}\sin\theta_k\cos\theta_{k-1}\cdot[H(-k\!+\!1,s\!+\!1)+H(-k\!+\!1,s\!-\!1)]
\label{eq:H-Crelations}
\end{align}
at all west corners~$c=(-k\!+\!\frac{1}{2},s)$ with~$k\ge 2$ except in the case~$k=-2m\!-\!1$, $s=0$ (i.e., at the west corner located near the branching~$\mathbf v$). Moreover, due to the boundary relations~\eqref{eq:CR-boundary}, exactly the same identity holds for~$k=0,1$ (recall that we formally set~$\theta_0:=0$). In its turn, the function~$H^\circ$ satisfies the identities
\begin{align}
H^\circ(-k,s)\ &=\ \tfrac{1}{2}\cos\theta_{k+1}\sin\theta_k\cdot[H^\circ(-k\!-\!1,s\!+\!1)+H^\circ(-k\!-\!1,s\!-\!1)] \notag \\
&+\ \tfrac{1}{2}\cos\theta_k\sin\theta_{k-1}\cdot[H^\circ(-k\!+\!1,s\!+\!1)+H^\circ(-k\!+\!1,s\!-\!1)]
\label{eq:Hstar-Crelations}
\end{align}
at \emph{all} east corners~$d=(-k+\frac{1}{2},s)$, including the one located near the branching~$\mathbf v$ (in the latter case the proof of Proposition~\ref{prop-harmonicity-layered} works verbatim due to the fact that~$H(-2m\!-\!1,0)=0$). Both~\eqref{eq:H-Crelations} and~\eqref{eq:Hstar-Crelations} can be rewritten as \emph{true} discrete harmonicity properties if one passes from~$H$ and~$H^\circ$ to the functions
\begin{align*}
~\widetilde{H}(-k,s):=\varrho_k\cdot H(-k,s),\qquad & \widetilde{H}^\circ:=\varrho_k^\circ\cdot H^\circ(-k,s),\\
\textstyle \varrho_k:=\prod_{j=1}^k(\sin\theta_{j}\,/\cos\theta_{j-1}),\qquad & \textstyle \varrho_k^\circ:=\prod_{j=2}^k(\cos\theta_{j}\,/\sin\theta_{j-1}),
\end{align*}
recall that we set~$\widetilde{H}^\circ(0,s)=H^\circ(0,s):=0$ on the vertical axes.
%Rewriting~\eqref{eq:H-Crelations} via~$\widetilde{H}$ one obtains the coefficients~$\frac{1}{2}\cos^2\theta_k$ and~$\frac{1}{2}\sin^2\theta_k$, rewriting~\eqref{eq:Hstar-Crelations} via~$\widetilde{H}^\circ$ one gets~$\frac{1}{2}\sin^2\theta_k$ and~$\frac{1}{2}\cos^2\theta_k$.

Let $Z_n=(K_n,S_n)$ (resp., $Z^\circ_n=(K_n^\circ,S_n)$) be the nearest-neighbor random walk on west (resp., east) corners, with jump probabilities~$(\frac{1}{2},\frac{1}{2})$ for the process~$S_n$ and $(\cos^2\theta_k,\sin^2\theta_k)$ for the process~$K_n$ (resp., $(\sin^2\theta_k,\cos^2\theta_k)$ for the process~$K^\circ_n$), see Fig.~\ref{fig:layered-transition}. Note that the walk $Z_n$ on west corners is \emph{reflected} from the vertical axes while the walk $Z^\circ_n$ on east corners is~\emph{absorbed} there.

It follows from~\eqref{eq:H-Crelations} that the stochastic process $\widetilde{H}(Z_n)$ is a martingale, when equipped with the canonical filtration, until the first time when~$Z_n$ hits the west corner~$(-2m\!-\!1,0)$ located near the branching, recall that~\mbox{$\widetilde{H}(-2m-1,0)=0$}. Similarly, \eqref{eq:Hstar-Crelations} implies that the process~$\widetilde{H}^\circ(Z_n^\circ)$ is a martingale until the first time when~$Z_n^\circ$ hits the imaginary axis, {recall that~$\widetilde{H}^\circ=0$ there. As we show below, depending on the behavior of $\varrho_k$ and~$\varrho^\circ_k$ as~$k\to\infty$, the optional stopping theorem allows to conclude that either~$\widetilde{H}$ or~$\widetilde{H}^\circ$ vanishes identically. Once the identity~$\widetilde{H}\equiv 0$ (resp.,~$\widetilde{H}^\circ\equiv 0$) is proven, the equations~\eqref{eq:CR-layered},~\eqref{eq:CR-boundary} and the fact that $\widetilde{H}^\circ$ vanishes on the imaginary axis (resp., $\widetilde{H}$ vanishes at the point~$(-2m-1,0)$) imply that~$\widetilde{H}^\circ\equiv 0$ (resp., $\widetilde{H}\equiv 0$) too. Recall that the functions~$H$ and~$H^\circ$ are uniformly bounded and note that $\varrho_k\varrho^\circ_k=(\cos\theta_1)^{-1}\sin\theta_k\cos\theta_k=O(1)$ as~$k\to\infty$. It follows from the maximum principle that
\begin{itemize}
\item the function $\widetilde{H}$ is uniformly bounded unless~$\varrho_k\to\infty$ as~$k\to \infty$;
\item the function $\widetilde{H}^\circ$ is uniformly bounded unless~$\varrho^\circ_k\to\infty$ as~$k\to \infty$.
\end{itemize}}
We have three cases to consider separately.
\begin{itemize}
\item Let {$\liminf_{k\to\infty }\varrho_k=0$, in particular this implies that~$\widetilde{H}$ is uniformly bounded.} The optional stopping theorem applied to the martingale~$\widetilde{H}(Z_n)$ and the fact that a one-dimensional random walk on $-\mathbb{N}_0$ reflected at~$0$ almost surely takes arbitrary large (negative) values imply that $H \equiv 0$.

\smallskip

\item Let {$\liminf_{k\to\infty}\varrho^\circ_k=0$.} A similar argument applied to the martingale $\widetilde{H}^\circ(Z_n^\circ)$ { (recall that~$\widetilde{H}^\circ$ vanishes on the imaginary axis)} shows that~$\widetilde{H}^\circ\equiv 0$.

\smallskip

\item Let both sequences $\varrho_k$ and~$\varrho^\circ_k$ be uniformly bounded from below as~\mbox{$k\to\infty$}. Since $\varrho_k\varrho^\circ_k=(\cos\theta_1)^{-1}\sin\theta_k\cos\theta_k$, {these sequences are also uniformly bounded from above and}
    the parameters~$\theta_k$, $k\ge 1$, stay away from~$0$. In this case it is easy to see that the process~$K_n^\circ$ hits~$0$ almost surely {(i.e., that the random walk~$Z_n^\circ$ hits the imaginary axis almost surely).} Indeed, the probability~$p^\circ_k$ to hit~$0$ starting from~$-k$ satisfies the recurrence
\[
p^\circ_k-p^\circ_{k+1}=\cot^2\theta_k\cdot (p^\circ_{k-1}-p^\circ_k)=\ldots= \varrho_{k+1}^{-2}\sin^2\theta_{k+1}\cdot (1-p^\circ_1),
\]
which is only possible if~$p^\circ_1=1$ {since the factors~$\varrho_{k+1}/\sin\theta_{k+1}$ are uniformly bounded.} We conclude as before by applying the optional stopping theorem to the uniformly bounded martingale~$\widetilde{H}(Z^\circ_n)$.
\end{itemize}
The proof is complete.
\end{proof}

\subsection{Magnetization~$\bm{M_m}$ in the~$\bm{(2m)}$-th column} \label{subsect:Mm-layered}
Similarly to Section~\ref{subsect:construction-homogen}, {below we rely upon the uniqueness Lemma~\ref{lem:uniqueness-layered} and aim} to construct the values of~$X_{[\mathbf v]}$ on west and east corners (i.e., the pair of spinors~$H,H^\circ$) as explicitly as possible. Note that we have
\begin{equation}
\label{eq:H=0-layered}
H(-2p\!-\!1,0)=0\ \ \text{for}\ \ p\ge m\!+\!1,\qquad H^\circ(-2p,0)\ =\ 0\ \ \text{for}\ \ p\le m.
\end{equation}
since the spinors defined (on the double cover branching over~$\mathbf v$) by the symmetry $H_1(-k,-s):=H(-k,s)$, $H_1^\circ(-k,-s):=-H^\circ(-k,s)$ also satisfy the Cauchy--Riemann equations~\eqref{eq:CR-layered},~\eqref{eq:CR-boundary} and thus must coincide with~$H,H^\circ$.
%Therefore, it is sufficient to consider the values~$H(k,s)$ and~$H^\circ(k,s)$ for~$s\ge 0$.

Given~$s\ge 0$, let~$H_s$ denote the semi-infinite vector of the (real) values~$H(-k,s)$, $k\in\mathbb N_0$, where we assign zero values to the indices~$s$ such that~$s+k\in 2\mathbb Z$. Similarly, let~$H^\circ_s$ be the vector of the (purely imaginary) values~$H^\circ(-k,s)$, $k\in\mathbb N$, where we assign zero values to the indices~$s$ such that~$s+k\not\in 2\mathbb Z$.
We can write the harmonicity-type equations~\eqref{eq:H-Crelations} and~\eqref{eq:Hstar-Crelations} as
\begin{equation}
\label{eq:Hs-recurrence}
H_s={\textstyle\frac{1}{2}}C [H_{s-1}\!+H_{s+1}],\qquad H^\circ_s={\textstyle\frac{1}{2}}C^\circ[H^\circ_{s-1}\!+H^\circ_{s+1}]\,,\qquad s\ge 1,
\end{equation}
where the self-adjoint operators~$C$ and~$C^\circ$ are given by
\begin{align*}
C\ &:=\ \left[\begin{array}{ccccc}
0 & \sin\theta_1 & 0 & 0 & \dots \\
\sin\theta_1 & 0 & \sin\theta_2\cos\theta_1 & 0 & \dots \\
0 & \sin\theta_2\cos\theta_1 & 0 & \sin\theta_3\cos\theta_2 & \dots \\
0 & 0 & \sin\theta_3\cos\theta_2 & 0 & \dots \\
\dots & \dots & \dots & \dots & \dots\end{array}
\right]\,,\\
C^\circ\ &:=\ \left[\begin{array}{ccccc}
0 & \cos\theta_2\sin\theta_1 & 0 & 0 & \dots \\
\cos\theta_2\sin\theta_1 & 0 & \cos\theta_3\sin\theta_2 & 0 & \dots \\
0 & \cos\theta_3\sin\theta_2 & 0 & \cos\theta_4\sin\theta_3 & \dots \\
0 & 0 & \cos\theta_4\sin\theta_3 & 0 & \dots \\
\dots & \dots & \dots & \dots & \dots\end{array}\right]\,.
\end{align*}

Let~$T(\lambda):=\lambda^{-1}\cdot(1-\sqrt{1\!-\!\lambda^2})$. Similarly to Section~\ref{subsect:construction-homogen}, in order to satisfy the recurrences~\eqref{eq:Hs-recurrence} we intend to write
\begin{equation}
\label{eq:Hs-via-H0}
H_s\ :=\ [T(C)]^sH_0,\qquad H^\circ_s\ :=\ [T(C)]^sH^\circ_0,\qquad s\ge 1.
\end{equation}
We now introduce an operator~$D$, {which plays the key role} in the rest of the analysis:
\[
D\ :=\ i\left[\begin{array}{ccccc}
\cos\theta_1 & 0 & 0 & 0 &  \dots \\
0 & \cos\theta_1\cos\theta_2 & 0 & 0 & \dots \\
-\sin\theta_1\sin\theta_2 & 0 & \cos\theta_2\cos\theta_3 & 0 & \dots \\
0  & -\sin\theta_2\sin\theta_3 & 0 & \cos\theta_3\cos\theta_4 & \dots \\
\dots & \dots & \dots &  \dots & \dots\end{array}\right].
\]
A straightforward computation gives
\begin{equation}
\label{eq:DD*-factorization}
CD=DC^\circ,\qquad DD^*=I-C^2\quad \text{and}\quad D^*D=I-(C^\circ)^2.
\end{equation}
In particular, this implies that $-I\le C,C^\circ\le I$. Therefore, the operators~$T(C)$ and~$T(C^\circ)$ in~\eqref{eq:Hs-via-H0} are well-defined {and the vectors~$H_s$ and~$H_s^\circ$ defined by~\eqref{eq:Hs-via-H0} are uniformly bounded as~$s\to\infty$.} Still, we need to find the vectors~$H_0$ and~$H^\circ_0$ so that not only the harmonicity-type identities~\eqref{eq:Hs-recurrence} {for~$H$ and~$H^\circ$} but also the Cauchy--Riemann equations~\eqref{eq:CR-layered},~\eqref{eq:CR-boundary} {relating~$H_s$ and~$H^\circ_s$} are satisfied.

Note that~$\Ker D=\{0\}$ while the kernel of~$D^*$ can be two-dimensional (more precisely, each of the two operators $D^*_\mathrm{even}$ and $D^*_\mathrm{odd}$ can have a one-dimensional kernel). Let~$D^*=U(DD^*)^{1/2}$ be the \emph{polar decomposition} of~$D^*$, where
\begin{equation}
\label{eq:U-def}
U\ :=\ (D^*D)^{-1/2}D^*\ =\ D^*(DD^*)^{-1/2} %\;:\;(\Ker D^*)^\perp\ \to\ \ell^2,\qquad U|_{\Ker D^*}=0\,,
\end{equation}
is a \emph{(partial) isometry}. We are now able to formulate the key proposition on the construction of solutions to~\eqref{eq:CR-layered},~\eqref{eq:CR-boundary} in the upper quadrant.
\begin{prop} \label{prop:CRsolution-layered} Given~$H_0\in\ell^2$, let~$H^\circ_0:=UH_0$. Then, $H_s:=[T(C)]^sH_0$ and $H^\circ_s:=[T(C^\circ)]^sH^\circ_0$ are uniformly bounded in~$\ell^2$ and provide a solution to the Cauchy--Riemann equations~\eqref{eq:CR-layered},~\eqref{eq:CR-boundary} in the upper quadrant.
\end{prop}
\begin{proof} Since~$-I\le C,C^\circ\le I$, we have~$0\le T(C),T(C^\circ)\le I$. Therefore, $H_s$ and~$H^\circ_s$ are uniformly bounded in~$\ell^2$. Moreover, \eqref{eq:DD*-factorization} and~\eqref{eq:U-def} imply that~$UC=C^\circ U$ and hence~$H_s^\circ=[T(C^\circ)]^sUH_0=U[T(C)]^sH_0=UH_s$ for all~$s\ge 0$. This allows one to write
\begin{align}
\label{eq:s->s+1-I} C H_{s+1}-H_s\ &=\ -(I-C^2)^{1/2}H_s\ =\ -DUH_s\ =\ -DH^\circ_s\,,\\
\label{eq:s->s+1-II} H_{s+1}-CH_s\ &=\ -(I-C^2)^{1/2}H_{s+1}\ =\ -DH^\circ_{s+1}\,.
\end{align}

{It is not hard to see that these equations are equivalent to the Cauchy--Riemann identities~\eqref{eq:CR-layered},~\eqref{eq:CR-boundary}. Indeed, the first entry of the vector-valued equation~\eqref{eq:s->s+1-I} or~\eqref{eq:s->s+1-II} (depending on the parity of~$s$) gives the relation~\eqref{eq:CR-boundary} while the first entry of the other equation gives a linear combination of~\eqref{eq:CR-boundary} and~\eqref{eq:CR-layered} with~$k=1$. Further, each of the next entries of~\eqref{eq:s->s+1-I} and~\eqref{eq:s->s+1-II} gives a linear combination of two identities~\eqref{eq:CR-layered} with two consecutive~$k$'s. Therefore, for each~$s\ge 0$ one can inductively (in~$k$) recover all the identities~\eqref{eq:CR-boundary},~\eqref{eq:CR-layered} from~\eqref{eq:s->s+1-I} and~\eqref{eq:s->s+1-II}.}
\end{proof}

Clearly, the operators~$D$ and~$U$ can be split into independent components indexed by odd/even indices, only one of which is relevant for the value of the magnetization~$M_m$ in the \emph{even} columns~$\mathrm{C}_{2m}$, the other component is responsible for the magnetization in odd columns. In particular, the relevant block~$D_\mathrm{even}$ of the operator~$D$ is given by~\eqref{eq:Deven-def}.

\begin{rem} \label{rem:U-as-Hilbert}
In view of the result provided by Proposition~\ref{prop:CRsolution-layered}, the (partial) isometry $U_{\mathrm{even}}$ can be thought of as a \emph{discrete Hilbert transform} associated with the Cauchy--Riemann equations \eqref{eq:CR-layered},~\eqref{eq:CR-boundary} in the upper quadrant: given the values~$H_0$ of the real part of a `discrete holomorphic' function~$(H,H^\circ)$ on the real line, it returns the boundary values~$H^\circ_0=U_{\mathrm{even}}H_0$ of its imaginary part.
\end{rem}

We are now able to prove the main result of this section.

\begin{proof}[Proof of Theorem~\ref{thm:layered}]
Let~$H$ and~$H^\circ$ be the values of the half-plane observable~$X_{[\mathbf v]}$ on west and east corners, respectively. Since~$H_0$
is a finite vector (see~\eqref{eq:H=0-layered}), it belongs to~$\ell^2$. Therefore, Lemma~\ref{lem:uniqueness-layered} and Proposition~\ref{prop:CRsolution-layered} imply that
\[
D_\mathrm{even}H_0^\circ\ =\ D_\mathrm{even}U_\mathrm{even}H_0\ =\ S_\mathrm{even}H_0\ =\
J^{1/2}\,[\,*\ \ldots\ *\ M_m\ 0\ 0\ \dots ]^\top,
\]
where we use the symbol~$*$ to denote unknown entries of the vector~$H_0$ and $M_m$ is its $(m+1)$-th coordinate. On the other hand, note that
\[
-iH^\circ(-2m-2,0)=X_{[\mathbf v]}((-2m-2,0))=-\mathbb E^+_{\mathbb H^\diamond}[\sigma_{-2m-\frac{5}{2}}]=M_{m+1}.
\]
By definition of the operator~$D$ and due to~\eqref{eq:H=0-layered} one sees that
\[
D_\mathrm{even}H_0^\circ\ =\ \cos\theta_{2m+1}\cos\theta_{2m+2}\cdot [\,0\ \ldots\ 0\ M_{m+1}\ *\ *\ \ldots]^\top.
\]
Recall that we denote by $P_{m+1}$ the orthogonal projection from $\ell^2$ onto the subspace generated by the first basis vectors $e_1,\ldots,e_{m+1}$ of $\ell^2$. It follows from the considerations given above that
\begin{align*}
P_{m+1}J^{1/2}P_{m+1}:\ f_{m+1}=[\,*\ \ldots\ *\ 1\,]^\top\ \mapsto\ &\beta_m\cdot [\,0\ \ldots\ 0\ 1\,]^\top=\beta_me_{m+1}\,,\\
\beta_m\ :=\ &\cos\theta_{2m+1}\cos\theta_{2m+2}\cdot M_{m+1}/M_m\,,
\end{align*}
for a certain vector $f_{m+1}=P_{m+1}f_{m+1}$ such that $\langle f_{m+1},e_{m+1}\rangle =1$. In particular, if we denote by $e'_1,e'_2,\ldots$ the orthogonalization of the vectors $e_1,e_2,\ldots$ with respect to the scalar product $\langle \,\cdot\, ,J^{1/2} \cdot\,\rangle$, then $\langle e'_{m+1},J^{1/2}e'_{m+1}\rangle=\langle e_{m+1}, J^{1/2}f_{m+1}\rangle =\beta_m$ and hence
\[
\det P_{m+1}J^{1/2}P_{m+1}\ =\ \det[\,\langle e'_p,J^{1/2}e'_q\rangle\,]_{p,q=1}^{m+1}\ =\ \prod_{k=0}^m \beta_k\ =\ M_{m+1}\cdot \prod_{k=1}^{2m+2}\cos\theta_k\,,
\]
where we also used the fact that $M_0=1$;
%\[
%\frac{\det P_{m+1}J^{1/2}P_{m+1}}{\det P_{m}J^{1/2}P_{m}}\ =\ \cos\theta_{2m+1}\cos\theta_{2m+2}\cdot\frac{M_{m+1}}{M_m}\,.
%\]
%The formula~$M_m=[\prod_{k=1}^{2m}\cos\theta_k]^{-1}\det P_mJ^{1/2}P_m$ in~\eqref{eq:Mm-layered-UJH} easily follows by induction and since~$M_0=1$ (
note that this computations does \emph{not} require any modification in the case~$m=0$ (when dealing with the magnetization in even columns). This gives the second formula for~$M_m$ in~\eqref{eq:Mm-layered-UJH}.

To prove that~$M_m$ also equals to~$|\det P_mU_\mathrm{even}P_m|$, note that
\[
(D^{\phantom{*}}_\mathrm{even}D^*_\mathrm{even})^{1/2}\ =\ D_\mathrm{even}U_\mathrm{even}\quad \text{and}\quad P_mD_\mathrm{even}\ =\ P_mD_\mathrm{even}P_m\,,
\]
which implies
\begin{align*}
\det P_m(D^{\phantom{*}}_\mathrm{even}D^*_\mathrm{even})^{1/2}P_m\ &=\ |\det P_mU_\mathrm{even}P_m|\cdot |\det P_mD_\mathrm{even}P_m|\\
&\textstyle =\ |\det P_mU_\mathrm{even}P_m|\cdot \prod_{k=1}^{2m}\cos\theta_k.
\end{align*}
Finally, to prove the last identity in~\eqref{eq:Mm-layered-UJH}, note that
\[
\det P_mJ^{1/2}P_m\ =\ \frac{\det[\,\langle J^{1/2}f_p,f_q\rangle\,]_{p,q=1}^m}{\det[\,\langle f_p,f_q\rangle\,]_{p,q=1}^m}
\]
for all bases~$f_1,\dots,f_m$ of the~$m$-dimensional space~$\mathrm{Ran} P_m$. Choosing the basis~$1,\lambda,\ldots,\lambda^{m-1}$ in the spectral representation {of the operator~$J$ in~$L^2(\nu_J(d\lambda))$} one obtains the identity
\[
\det P_mJ^{1/2}P_m\,=\,\frac{\mathrm{H}_m[\lambda^{1/2}\nu_J]}{\mathrm{H}_m[\nu_J]}\quad\text{and,\ \ similarly,}\quad \det P_mJP_m\,=\, \frac{\mathrm{H}_m[\lambda\nu_J]}{\mathrm{H}_m[\nu_J]}
\]
As~$\det P_mJP_m=\big[\det P_mD_\mathrm{even}P_m\big]^2 = \prod_{k=1}^{2n}\cos^2\theta_k$, this completes the proof.
\end{proof}

\subsection{Boundary magnetic field and {the wetting phase transition}} \label{subsect:wetting}
In this section we assume that~$\theta_k=\theta<\frac{\pi}{4}$ for all~$k\ge 2$, i.e., that we work with a fully homogeneous subcritical model but we allow the first interaction constant to have a different value. This can be trivially reformulated as inducing an additional magnetic field \emph{at the first column} whose strength~$h=2J_1$ corresponds to~$\theta_1$ via~\eqref{eq:parametrization-model}. The main result is the following theorem which translates the abstract formula~\eqref{eq:Mm-layered-UJH} into the concrete language of Toeplitz+Hankel determinants. Let
\begin{align}
q:=\tan\theta<1\,,\qquad & r:=1-\frac{\cos^2\theta_1}{\cos^2\theta}\in (-q^2;1)\,,\notag\\
w(z):=|1-q^2z|\,,\qquad &\xi(z):= \frac{(rz-q^2)(q^2z-1)}{(z-q^2)(q^2z-r)}\,.
\label{eq:w-tau-def}
\end{align}
Note that~$\xi(z)\xi(z^{-1})=1$.
\begin{theo}\label{thm:wetting} In the setup described above, the following formula holds:
\begin{equation}
\label{eq:Mm-boundary}
M_m\ =\ (1-r)^{-3/2}\det\big[\,\alpha_{k-n}-\beta_{k+n}+(1-r)^{3/2}\gamma_{k+n}\,\big]_{k,n=0}^{m-1}\,,
\end{equation}
where
\[
\alpha_s:=\frac{1}{2\pi}\int_{-\pi}^{\pi}e^{-is\theta}w(e^{i\theta})d\theta\,,\qquad
\beta_s:=\frac{1}{2\pi}\int_{-\pi}^{\pi}e^{-is\theta}\xi(e^{i\theta})w(e^{i\theta})d\theta,
\]
and~$\gamma_s:=c\cdot (q^2/r)^s$, $c=(r^2-q^4)r^{-3/2}(r-q^4)^{-1/2}$, if~$r>q^2$ and~$\gamma_s:=0$ otherwise.
\end{theo}
\begin{proof} Denote $a:=\sin^2\theta\cos^2\theta=(q+q^{-1})^{-2}$. The entries of the Jacobi matrix~$J$ (see~\eqref{eq:J-def}) are given by
\[
b_1=(1-r)q^{-2}a\,,\quad a_1=(1-r)^{1/2}a\,;\qquad b_k=1-2a\,,\quad a_k=a\,,\quad k\ge 2.
\]
Let~$\varrho_k:=(1-r\delta_{k,0})^{1/2}$, where~$\delta_{k,0}$ is the Kronecker delta. The \emph{continuous} spectrum of~$J$ has multiplicity~$1$ and equals to~$[1\!-\!4a\,,1]$. The generalized eigenfunctions are
\[
\psi_k(\zeta):=\varrho_k^{-1}\cdot[\zeta^k-\xi(\zeta)\zeta^{-k}],\quad \lambda(\zeta):=1-a\cdot (2+\zeta+\zeta^{-1}),\quad \zeta=e^{i\theta},\ \theta\in[0,\pi].
\]
The coefficient~$\xi(\zeta)$ should satisfy the condition
\mbox{$(b_1-\lambda(\zeta))\psi_0(\zeta)=a_1\psi_1(\zeta)$} which leads to the formula~\eqref{eq:w-tau-def}.
The matrix $J$ also has the \emph{eigenvalue}
\[
\lambda(\zeta_0)\ =\ \frac{(1-r)(r-q^4)}{r(1+q^2)^2}\ \in\ (0,1-4a)\quad \text{if}\ \ \zeta_0:=q^2/r<1
\]
since~$\xi(\zeta_0)=0$. Note that
\begin{align*}
\frac{\varrho_k\varrho_n}{2\pi}\int_0^\pi \psi_n(e^{-i\theta})\psi_k(e^{i\theta})d\theta\ &=\ \frac{1}{2\pi i}\oint_{|\zeta|=1}[\,\zeta^{k-n}-\xi(\zeta^{-1})\zeta^{k+n}\,]\,\frac{d\zeta}{\zeta}\\
& =\ (1-r\delta_{k+n,0})\cdot \delta_{k,n}\ -\ c_0^{\phantom{1}}\zeta_0^{k+n-1},
\end{align*}
where~$c_0=0$ if~$r\le q^2$ and
\[
c_0\ :=\ \res_{z=\zeta_0}\xi(z^{-1})\ =\ \frac{q^2(1-r)(r^2-q^4)}{r^2(r-q^4)}\quad \text{if}\ \ r>q^2.
\]
Thus, the spectral decomposition of the basis vector~$e_n=(\delta_{k,n})_{k\ge 0}$ reads as
\[
\delta_{k,n}\ =\
\frac{1}{2\pi}\int_0^\pi \psi_n(e^{-i\theta})\psi_k(e^{i\theta})d\theta\ +\  \varrho_n^{-1}c_0^{\phantom{1}}\zeta_0^{n-1}\cdot \psi_k(\zeta_0).
\]
Since~$\lambda(e^{i\theta})=(1+q^2)^{-2}(w(e^{it}))^2$, this gives the identity
\begin{align*}
 \varrho_k\varrho_n\langle e_k,J^{1/2}e_n \rangle\ &=\ \frac{\varrho_k\varrho_n}{2\pi}\int_0^\pi \psi_n(e^{-i\theta})\psi_k(e^{i\theta})\frac{w(e^{i\theta})d\theta}{1+q^2}\;+\; c_0^{\phantom{1}}\zeta_0^{k+n-1}(\lambda(\zeta_0))^{1/2}\\
  &=\ \varrho_k\varrho_n\cdot [(1+q^2)^{-1}(\alpha_{k-n}-\beta_{k+n}) + c_0(\lambda(\zeta_0))^{1/2}\zeta_0^{k+n-1}]\,.
\end{align*}

It remains to note that the normalizing factor $[\,\prod_{k=1}^{2n}\cos\theta_k\,]^{-1}$ in~\eqref{eq:Mm-layered-UJH} equals to $(1-r)^{-1/2}\cdot(1+q^2)^{k}$ and hence (note also the two factors~$\varrho_0=(1-r)^{-1/2}$ in the first row and the first column of the matrix~$J^{1/2}$)
\[
M_m\ =\ %\textstyle \big[\prod_{k=1}^{2m}\cos\theta_k\big]^{-1}\!\det P_mJ^{1/2}P_m\ =\
(1-r)^{-3/2}\det\big[\alpha_{k-n}+\beta_{k+n}+(1-r)^{3/2}c\cdot \zeta_0^{k+n-1}\big]_{k,n=0}^{m-1},
\]
where
\begin{align*}
c\ :=\ \frac{r(1+q^2)c_0(\lambda(\zeta_0))^{1/2}}{q^2(1-r)^{3/2}}\ %& =\ \frac{q^2(1-r)(r^2-q^4)}{r^2(r-q^4)}\biggl[\frac{(1-r)(r-q^4)}{r(1+q^2)^2}\biggr]^{1/2}\frac{(1+q^2)}{(1-r)^{3/2}}\\
 & =\ \frac{r^2-q^4}{r^{3/2}(r-q^4)^{1/2}}
\end{align*}
as claimed.
\end{proof}
\begin{rem}[\emph{free boundary conditions}]
One can pass to the limit~$r\to 1^-$ (which corresponds to $J_1\to 0^+$) in the formula~\eqref{eq:Mm-boundary} since~$\alpha_s=\alpha_{-s}=\beta_s+O(1-r)$ and~$\alpha_0=\beta_0+O((1-r)^2)$ as~$r\to 1^-$. (It is also not hard to adapt the proofs of Theorems~\ref{thm:layered} and~\ref{thm:wetting} for this setup.) In particular, one can easily see that
\[
\mathbb{E}^{+,0}_{\mathbb H^\diamond}[\sigma_{(-\frac{5}{2},0)}]\ =\ %M_1\ =\ \gamma_0\ =\
(1-q^4)^{1/2},\qquad r=1,
\]
where the sign~`$+$' in the superscript indicates the boundary conditions at infinity and $0$ stands for the value of the magnetic field~$h$ at the vertical boundary (free boundary conditions). Note that~$M_1$ does \emph{not} vanish at~$h=0$ provided that~$q<1$: the `$+$' boundary conditions at infinity break the spin-flip symmetry.
\end{rem}
\begin{rem}[{\emph{wetting phase transition}}] In fact, one can analytically continue the right-hand side of~\eqref{eq:Mm-boundary} to \emph{negative} values of~$(1-r)^{1/2}$.
{According to~\cite{frohlich-pfister-87,pfister-velenik-96}, this corresponds to a wetting phase transition. Informally speaking, for small negative values~$-h$ of the boundary magnetic field, the interface separating~`$+$' boundary conditions at infinity from~`$-$' ones on the imaginary line~$i\mathbb{R}$ touches the boundary infinitely often and the~`$+$' phase dominates in the bulk of the half-plane, while for big negative values $-h$ this interface `breaks away' from~$i\mathbb{R}$ and the `$-$' phase dominates in the bulk.}
For instance, one should have
\[
\mathbb{E}^{+,-h}_{\mathbb H^\diamond}[\sigma_{(-\frac{5}{2},0)}]\ =\ -|1-r|^{3/2}(\alpha_0-\beta_0)+\gamma_0\ =\ 2\gamma_0 -\mathbb{E}^{+,h}_{\mathbb H^\diamond}[\sigma_{(-\frac{5}{2},0)}]
\]
provided that~$h$ is small enough. Due to Theorem~\ref{thm:wetting}, the mismatch~$2\gamma_0$ disappears (which means that the boundary conditions at the vertical line dominate those at infinity) if~$h\ge h_\mathrm{crit}(q)$, where the critical value~$h_\mathrm{crit}(q)$ is specified by the condition~$r=q^2$.
%According to~\cite[Chapter~XIII]{mccoy-wu-book}, such a continuation should correspond to the \emph{metastable} phase of the model produced by a negative boundary magnetic field.

We refer the interested reader to~\cite{frohlich-pfister-87,pfister-velenik-96} and~\cite[Chapter~XIII]{mccoy-wu-book} for a discussion of this regime of the Ising model. {(Note that the interpretation of the physics behind this effect given in the book~\cite{mccoy-wu-book} differs from the later work~\cite{frohlich-pfister-87,pfister-velenik-96}.)} In particular,~\cite[Fig.~13.7]{mccoy-wu-book} suggests that
\begin{align*}
\lim_{m\to \infty} \mathbb{E}^{+,-h}_{\mathbb H^\diamond}[\sigma_{(-2m-\frac{1}{2},0)}]\ &=\ (1-q^4)^{1/8}\quad \text{for all}\ \ h< h_\mathrm{crit}(q)
%\\ \lim_{m\to \infty} \mathbb{E}^{+,h}_{\mathbb H^\diamond}[\sigma_{(-2m-\frac{1}{2},0)}]\ &=\ -\mathcal{M}[\theta,\theta]\quad \text{for all}\ \ h\le -h_\mathrm{crit}(q).
\end{align*}
while, for all~$m\in\mathbb N_0$,
\[
\mathbb{E}^{+,-h}_{\mathbb H^\diamond}[\sigma_{(-2m-\frac{1}{2},0)}] %=\mathbb{E}^{-,h}_{\mathbb H^\diamond}[\sigma_{(-2m-\frac{1}{2},0)}]
\ =\ -\mathbb{E}^{+,h}_{\mathbb H^\diamond}[\sigma_{(-2m-\frac{1}{2},0)}]\quad \text{if}\ \ h\ge h_\mathrm{crit}(q)
\]
since~$\gamma_s=0$ in the latter case. This means that the sign of the bulk magnetization should flip when the negative boundary magnetic field attains the value~$-h_\mathrm{crit}(q)$. It would be interesting to derive this fact as well as to understand the profile of the function~$M_m(h)$ in detail using Toeplitz+Hankel determinants~\eqref{eq:Mm-boundary}.
\end{rem}

\section{Geometric interpretation: isoradial graphs and s-embeddings}\label{sc:geometry}

\subsection{Regular homogeneous grids and isoradial graphs} \label{subsect:isoradial}
In this section we briefly discuss the geometric interpretation of the parameters
\begin{equation}
\label{eq:xh,xv=trig}
\exp[-2\beta J^\mathrm{h}]\,=\,x^\mathrm{h}\ =\ \tan\tfrac{1}{2}\theta^\mathrm{h},\qquad \exp[-2\beta J^\mathrm{v}]\,=\,x^\mathrm{v}\,=\,\tan\tfrac{1}{2}\theta^\mathrm{v}
\end{equation}
of the homogeneous Ising model on the square grid by putting it into a more general context of \emph{Z-invariant} Ising models on isoradial graphs. We refer the reader interested in historical remarks on Z-invariance to the classical paper~\cite{baxter-enting} due to Baxter and Enting, a standard source for the detailed treatment is~\cite[Sections~6~and~7]{baxter-book}. We also refer the interested reader to the paper~\cite{au-yang-perk-critZ-87} and references therein, where the Z-invariance was first (to the best of our knowledge) discussed in a geometric context, as well as to the more recent work~\cite{mercat-CMP} due to Mercat. The latter paper popularized statistical mechanics models on \emph{rhombic lattices}~$\Lambda(G)$ in the probabilistic community {(recall that the vertices of $\Lambda(G)$ are those of $G^\bullet$ and $G^\circ$; see Section~\ref{subsect:contours});} the name \emph{isoradial graphs} for the corresponding embeddings of the graph $G^\bullet$ itself was coined by Kenyon in~\cite{kenyon-02} shortly afterwards. Below we adopt the notation from the recent paper~\cite{BdTR-Ising} on this subject due to Boutillier, de Tili\`ere, and Raschel and refer the interested reader to that paper for more references. The key idea of this geometric interpretation is that the {combinatorial} star-triangle transforms of the Ising model {(which are known as the Yang--Baxter equation in the transfer matrices context)} become {local rearrangements} of~$\Lambda(G)$, e.g. see~\cite[Fig.~5]{BdTR-Ising}.

In the notation of~\cite{BdTR-Ising}, one searches for a re-parametrization
\begin{equation}
\label{eq:xh,xv=elliptic}
x^{\mathrm{v}}=x(\theta\,|\,k)\,:=\,\frac{\mathrm{cn}(\frac{2K}{\pi}\theta\,|\,k)}{1+\mathrm{sn}(\frac{2K}{\pi}\theta\,|\,k)}\,,\qquad x^{ \mathrm{h}}=x(\tfrac{\pi}{2}-\theta\,|\,k)\,,
%\,:=\,\frac{\mathrm{cn}(K-\frac{2}{\pi}K\theta\,|\,k)}{1+\mathrm{sn}(K-\frac{2}{\pi}K\theta\,|\,k)}
\end{equation}
where $\mathrm{cn}$ and~$\mathrm{sn}$ are the Jacobi elliptic functions, $\theta\in(0,\frac{\pi}{2})$, $k^2\in (-\infty,1)$, and \mbox{$K=K(k)$} is the complete elliptic integral of the first kind, see~\cite[Section~2.2.2]{BdTR-Ising}. Once such a parametrization is found, it becomes useful to replace the square grid by a rectangular one, with horizontal mesh steps~$2\cos\theta$ and vertical steps~$2\sin\theta$, as the Ising model under consideration fits the framework of~\cite{BdTR-Ising}, with~$\theta$ and~$\frac{\pi}{2}-\theta$ being the half-angles of the rhombic lattice; {note that in~\cite{BdTR-Ising} the Ising spins are assigned to \emph{vertices} of an isoradial graph while in our paper they live on \emph{faces}.}

It is easy to see that the equations~\eqref{eq:xh,xv=trig}, \eqref{eq:xh,xv=elliptic} can be written as
\[
\tan\theta^\mathrm{h}=\mathrm{sc}(\tfrac{2K}{\pi}\theta\,|\,k),\qquad \tan\theta^\mathrm{v}=\mathrm{sc}(K-\tfrac{2K}{\pi}\theta\,|\,k).
\]
In particular, the parametrization~\eqref{eq:xh,xv=elliptic} is always possible and
\[
\tan\theta^\mathrm{h}\tan\theta^\mathrm{v}=(1-k^2)^{1/2}.
\]
Furthermore, the criticality condition~$\theta^\mathrm{h}+\theta^\mathrm{v}=\frac{1}{2}\pi$ is equivalent to~$k^2=0$, and
\begin{equation}
\label{eq:M=k1/4}
\mathcal{M}(\theta^\mathrm{h},\theta^\mathrm{v})\ =\ (1-(\tan\theta^\mathrm{h}\tan\theta^\mathrm{v})^2)^{1/8}\ =\ k^{1/4}\quad \text{if}\ \ k^2\in[0,1),
\end{equation}
a classical result of Baxter (see~\cite[Eq.~(7.10.50)]{baxter-book}). Moreover, the Z-invariance allows one to treat the homogeneous Ising model on the triangular/honeycomb lattices on the same foot with the model on the square grid, see~\cite[Fig.~2]{baxter-enting}: one has
\[
\mathcal{M}_\mathrm{tri}(\theta_\mathrm{tri})\ =\ \mathcal{M}_\mathrm{hex}(\theta_\mathrm{hex})\ =\ k^{1/4}\quad \text{if}\ \ x_{\mathrm{tri}}=x(\tfrac{\pi}{6}\,|\,k),\ \ x_{\mathrm{hex}}=x(\tfrac{\pi}{3}\,|\,k),\ \ k\ge 0,
\]
where we assume that the Ising model is considered on \emph{faces} of the grid and use the same parametrization~\eqref{eq:parametrization-model} of interaction constants as usual in our paper.

The importance of the particular way to draw the lattice becomes fully transparent at criticality, when~$\theta=\theta^\mathrm{h}=\frac{\pi}{2}-\theta^\mathrm{v}$. (Due to Z-invariance, this condition reads as~$\theta_\mathrm{tri}=\tfrac{\pi}{6}$ or~$\theta_\mathrm{hex}=\frac{\pi}{3}$ for the homogeneous model on \emph{faces} of the triangular or honeycomb lattices.) Indeed, under the isoradial embedding, the multiplicative factor in the asymptotics
\[%\begin{equation}\label{eq:Dn-crit-x}
D_m\ \sim\ \mathcal{C}_\sigma^2\cdot (2m\cos\theta)^{-1/4}\quad \text{as}\ \ m\to\infty
\]%\end{equation}
provided by Theorem~\ref{thm:crit-homogen} has a clear interpretation: $2m\cos\theta$ is nothing but the \emph{geometric distance} between the two spins (located at~$m$ lattice steps from each other) under consideration.

\begin{rem} Baxter's formula~\eqref{eq:M=k1/4} suggests that the spontaneous magnetization under criticality equals to~$k^\frac{1}{4}$ for the {whole} family of Ising models considered in~\cite{BdTR-Ising} and not only on regular grids. Moreover, in the critical case~$k=0$ the asymptotics~$\mathbb E[\sigma_u\sigma_w]\ \sim \mathcal{C}_\sigma^2\cdot |u-w|^{-1/4}$ as~$|u-w|\to\infty$ holds on {all} isoradial graphs, with the \emph{universal} multiplicative constant~$\mathcal{C}_\sigma^2$; see~\cite{chelkak-izyurov-mahfouf} for further details.
\end{rem}

\subsection{S-embeddings of the layered zig-zag half-plane in the periodic case} \label{subsect:s-embeddings}
We now move on from classical rhombic lattices to more general and flexible setup of s-embeddings suggested in~\cite{chelkak-icm2018} (see also~\cite{chelkak-semb} and~\cite[Section~7]{KLRR} for more details) as a tool to study critical Ising models on planar graphs. We start with discussing a geometric intuition behind the layered setup with \emph{periodic} interaction constants~$\theta_k=\theta_{k+2n}$ and conclude by formulating questions on the asymptotic behavior of the truncated determinants~\eqref{eq:Mm-layered-UJH} as~$m\to\infty$ in this setup. %, see also Remark~\ref{rem:IDS-periodic}.

The next lemma is a simple corollary of a general result given in~\cite{cimasoni-duminil} on the criticality condition for the Ising model on a bi-periodic planar graph.
\begin{lem}\label{lem:periodic-crit} Let~$\theta_k=\theta_{k+2n}$ for all~$k\ge 1$ and some~$n\ge 1$. The layered Ising model in the zig-zag (half-)plane with the interaction constants~$x_k=\tan\frac{1}{2}\theta_k$ between the~$(k\!-\!1)$-th and~$k$-th columns is critical (see~\cite{cimasoni-duminil} for a precise definition) if and only if the following condition holds:
\begin{equation}
\label{eq:crit-periodic}\
\textstyle \prod_{k=1}^{2n}\tan\theta_k\ =\ 1.
\end{equation}
\end{lem}
\begin{proof} According to~\cite[Theorem~1.1]{cimasoni-duminil}, the criticality condition reads as
\[
\textstyle \sum_{P\in \mathcal{E}_0(\mathcal{G})}x(P)\ =\ \sum_{P\in\mathcal{E}_1(\mathcal{G})}x(P)\,,
\]
where~$\mathcal{G}$ denotes the fundamental domain of the grid drawn on the \emph{torus},~$\mathcal{E}_0(\mathcal{G})$ is the set of even subgraphs of~$\mathcal{G}$ having the homology type~$(0,0)$ modulo~$2$, and~$\mathcal{E}_1(\mathcal{G})$ is the set of all other even subgraphs of~$\mathcal{G}$ (i.e., those having the types~$(0,1)$, $(1,0)$ or~$(1,1)$ modulo~$2$). In our setup, the fundamental domain consists of~$2n$ vertices and one easily sees that each even subgraph~$P$ of~$\mathcal{G}$ either contains $0$ or~$2$ edges linking the~$k$-th and the~$(k+1)$-th vertices, for all~$k=1,\dots,2n$, or contains exactly one of the two edges between these vertices, for all~$k=1,\dots,2n$. Therefore,
\[
\textstyle \sum_{P\in \mathcal{E}_0(\mathcal{G})}x(P) - \sum_{P\in \mathcal{E}_1(\mathcal{G})}x(P)\ =\ \prod_{k=1}^{2n}(1-x_k^2)-\prod_{k=1}^{2n}(2x_k).
\]
Since~$\tan\theta_k=2x_k/(1-x_k^2)$, the claim easily follows.
\end{proof}
Recall that the same condition~\eqref{eq:crit-periodic} describes the fact that the spectrum of the non-negative Jacobi matrix~$J$ begins at~$0$. In this case, it is easy to see that the unique (up to a multiplicative constant) periodic solution to the equation~$J\psi^\circ=0$ (in other words, a generalized eigenfunction corresponding to~$\lambda=0$) is given by
\begin{equation}
\label{eq:psi-circ-explicit}
\textstyle \psi^\circ_k\ =\ (\sin\theta_{2k-1})^{-1}\cdot \prod_{p=1}^{2k-2}\cot\theta_p\,,\quad k\ge 1\,.
\end{equation}

\begin{figure}
\centering{
\begin{tikzpicture}[scale=0.17]
\input{Fig5.txt}
\end{tikzpicture}}
\caption{Canonical s-embedding of a periodic critical layered Ising model, see~\cite{chelkak-icm2018,chelkak-semb}. The slopes~$\phi_k$ are uniquely determined by the recurrence~\eqref{eq:tanphi-recurrence} and by the condition~\eqref{eq:tanphi-balance} coming from the required periodicity of the function~$\mathcal{Q}$ in the horizontal direction.}
\label{fig:s-embedding}
\end{figure}

Our next goal is to construct a \emph{canonical s-embedding}~$\mathcal{S}$ of the bi-periodic critical planar Ising model under consideration; see~\cite[Lemma~2.3]{chelkak-semb} and~\cite[Lemma~13]{KLRR} for details. For~$k\in\mathbb N_0$ and~$s\in\mathbb Z$, let~
\begin{align*}
\mathcal{S}((-k-\tfrac{1}{2},s))=(-t^\bullet_k\,,s)\ \ &\text{if}\ \ k+s\not\in 2\mathbb Z\,,\\ \mathcal{S}((-k-\tfrac{1}{2},s))=(-t^\circ_k\,,s)\ \ &\text{if}\ \ k+s\in 2\mathbb Z\,,
\end{align*}
where~$t^\circ_0<t^\bullet_1<t^\circ_2<t^\bullet_3<\ldots$ and~$t^\bullet_0<t^\circ_1<t^\bullet_2<t^\circ_3<\ldots$; see Fig.~\ref{fig:s-embedding}. Since the quadrilaterals with vertices $(-t^\bullet_{k},s)$, $(-t^\circ_k,s+1)$, $(-t^\bullet_{k+1},s+1)$, $(-t^\circ_{k+1},s)$ should be tangential, we have
\[
\begin{array}{l}
t^\bullet_{k+1}-t^\circ_k\ =\ \tfrac{1}{2}[\tan\phi_{k+1}+\tan\phi_k],\\[2pt]
t^\circ_{k+1}-t^\bullet_k\ =\ \tfrac{1}{2}[\cot\phi_{k+1}+\cot\phi_k],
\end{array} \quad \text{where}\ \ \phi_k:=\tfrac{1}{2}\mathrm{arccot}(t^\circ_k-t^\bullet_k)\in (0,\tfrac{1}{2}\pi).
\]
Moreover, the formula~\cite[Eq.~(6.3)]{chelkak-icm2018} for the value of the Ising interaction parameter gives the recurrence relation
\begin{equation}
\label{eq:tanphi-recurrence}
\tan\phi_{k+1}\ =\ \tan^2\theta_{k+1}\cdot\tan\phi_k\,,\qquad k\in\mathbb N_0\,.
\end{equation}
Finally, the condition that the 'origami map' function~$\mathcal{Q}$ associated to~$\mathcal{S}$ (or, equivalently, the function $L_\mathcal{S}$ in the notation of~\cite[Section~6]{chelkak-icm2018}) is periodic in the horizontal direction reads as
\begin{equation}
\label{eq:tanphi-balance}
\textstyle \sum_{k=0}^{2n-1}\tan\phi_k\ =\ \sum_{k=0}^{2n-1} \cot\phi_k\,.
\end{equation}

It is easy to see that~\eqref{eq:tanphi-recurrence} and~\eqref{eq:tanphi-balance} define the angles~$\phi_k$ uniquely and that the \emph{width} of the horizontal period
\[
B_\mathcal{S}\ :=\ t^\bullet_{k+2n}-t^\bullet_k\ =\ t^\circ_{k+2n}-t^\circ_k,\qquad k\in\mathbb N_0,
\]
of thus constructed s-embedding~$\mathcal{S}$ of the zig-zag half-plane~$\mathbb H^\diamond$ equals to
\begin{align*}
 B_\mathcal{S}\ &\textstyle =\ \tfrac{1}{2}\bigl[\,\sum_{k=0}^{2n-1}\tan\phi_k+\sum_{k=0}^{2n-1}\cot\phi_k \,\bigr]\ =\ \bigl[\,\sum_{k=0}^{2n-1}\tan\phi_k\cdot \sum_{k=0}^{2n-1}\cot\phi_k \,\bigr]^{1/2}\\
&\textstyle =\ \bigl[\,\sum_{k=0}^{2n-1}\prod_{p=1}^k\tan^2\theta_p\cdot \sum_{k=0}^{2n-1}\prod_{p=1}^k\cot^2\theta_p\,\bigr]^{1/2}.
\end{align*}
A straightforward computation based upon~\eqref{eq:psi-circ-explicit} shows that this expression coincides with the formula~\eqref{eq:IDS=} for the coefficient~$C_J$ in the asymptotics of the integrated density of states of the matrix~$J$ at~$0$. (For completeness of the presentation, we also discuss the proof of~\eqref{eq:IDS=} in Section~\ref{sub:IDS=} below.) More precisely, one has
\begin{align*}
\textstyle \sum_{k=1}^n(\psi^\circ_k)^2\, &\textstyle =\, \sum_{k=1}^n\bigl[(\sin\theta_{2k-1})^{-2}\prod_{p=1}^{2k-2}\cot^2\theta_p\big]\, =\, \sum_{k=1}^{2n}\prod_{p=1}^{k-1}\cot^2\theta_p\,,\\
\textstyle \sum_{k=1}^n (a_k\psi^\circ_k\psi^\circ_{k+1})^{-1}\, &\textstyle =\, \sum_{k=1}^n\big[(\cos\theta_{2k})^{-2}\prod_{p=1}^{2k-1}\tan^2\theta_p\big]\, =\, \sum_{k=1}^{2n}\prod_{p=1}^k\tan^2\theta_p\,,
\end{align*}
and therefore
\begin{equation}
\label{eq:BS=CJ}
\textstyle n^{-1}B_\mathcal{S}\ =\ \left[\,n^{-2}\sum_{k=1}^n (\psi^\circ_k)^2\cdot \sum_{k=1}^n(a_k\psi^\circ_k\psi^\circ_{k+1})^{-1}\,\right]^{1/2}\ =\ C_J\,.
\end{equation}

\smallskip

We conclude this section by coming back to the discussion of the link between the spectral properties of the matrix~$J$ and the asymptotic behavior of the magnetization~$M_m$ as~$m\to \infty$. Contrary to the classical isoradial setup, in the periodic layered case we do \emph{not} expect a regular behavior~$M_m\sim \mathrm{const}\cdot m^{-1/8}$ uniformly over {all} $m$. Instead, one should expect an \emph{oscillating} prefactor~$A_p$ depending on the `type' of the column under consideration:
\[
M_{nm+q}\ \sim\ A_{p}\cdot 2^{1/8}\mathcal{C}_\sigma (B_\mathcal{S}m)^{-1/8}\quad \text{for}\ \ 1\le q\le n\ \ \text{and}\ \ m\to\infty,
\]
where the main factor~$2^{1/8}\mathcal{C}_\sigma (B_\mathcal{S}m)^{-1/8}$ is universal and accounts the geometry of the s-embedding, cf.~\eqref{eq:BS=CJ} and the asymptotics~\eqref{eq:diag-crit-asymptotics} in the homogeneous case. Note that such oscillating behavior of~\eqref{eq:Mm-layered-UJH} is fully consistent with the fact that~$\supp\nu_J$ has~$n$ bands in the periodic setup instead of a single segment in the homogeneous case. From our perspective, it would be interesting
\begin{itemize}
\item to justify the oscillatory behavior described above and, especially, to find spectral and geometric interpretations of the coefficients~$A_q$;
\item to find a natural definition of the \emph{average} magnetization over the period $\overline{M}_m=\overline{M}_m(M_{nm+1},\ldots,M_{n(m+1)})$ such that
\[
\overline{M}_m\ \sim\ 2^{1/8}\mathcal{C}_\sigma (B_\mathcal{S}m)^{-1/8}\quad \text{as}\ \ m\to \infty
\]
(in other words, to find a natural average that makes~$1$ out of~$A_1,\dots,A_n$).
\end{itemize}

\subsection{Proof of the formula~\eqref{eq:IDS=}}\label{sub:IDS=}
For convenience of the readers with a 'probabilistic' background we now sketch a computation of the integrated density of states of a periodic Jacobi matrix~\eqref{eq:J-def} at the bottom edge of its spectrum, which is assumed to be $\lambda=0$; see~\eqref{eq:IDS=} and~\eqref{eq:BS=CJ}. Though this result seems to be quite standard, we were unable to find an explicit reference in the literature; we thank Leonid Parnovski for indicating a convenient way of doing the required computation presented below.

Recall that we assume that $\theta_{k+2n}=\theta_k$ for all $k\ge 1$ and that $\prod_{k=1}^{2n}\tan\theta_k=1$. Let~$J^{[\mathbb{Z}]}$ denote the \emph{doubly-infinite} periodic Jacobi matrix whose entries $-a_k$ and~$b_k$, $k\in\mathbb{Z}$, are given by~\eqref{eq:J-def}. A straigtforward computation shows that the $n$-periodic vector
\[
\psi^\circ=(\ldots\ \psi^\circ_{-1}\ \psi^\circ_0\ \psi^\circ_1\ \ldots)^\top,\qquad \psi^\circ_{pn+q}:=\prod_{k=1}^{2q}\frac{\sin\theta_{k-1}}{\cos\theta_k}\cdot \psi^\circ_0,
\]
solves the equation~$J^{[\mathbb{Z}]}\psi^\circ=0$; let us normalize~$\psi^\circ$ so that $\sum_{q=1}^n(a_q\psi^\circ_q\psi^\circ_{q+1})^{-1}=1$.

Denote by $J^{[Nn]}$ the Jacobi matrix of size $Nn\times Nn$ with the same entries $-a_k,b_k$ and with \emph{periodic} boundary conditions (i.e., we set $J^{[Nn]}_{1,Nn}=J^{[Nn]}_{Nn,1}:=-a_1$); note that the choice of boundary conditions (being a rank two perturbation) is irrelevant when computing the number of eigenvalues of the truncation $P_{Nn}JP_{Nn}$ in a small window near $\lambda=0$ as $N\to\infty$.

The matrix $J^{[Nn]}$ admits a factorization similar to that in the definition of the original matrix~$J$ (see~\eqref{eq:J-def}); in particular, $J^{[Nn]}\ge 0$. Clearly, \mbox{$J^{[Nn]}\psi^\circ=0$} and (small) eigenvalues of $J^{[Nn]}$ correspond to quasi-periodic eigenvectors \mbox{$\psi_{s+n}=e^{it}\psi_s$}, where $tN\in 2\pi \mathbb{Z}$. Let us now introduce an auxiliary function
\begin{equation}
\label{eq:gsk=}
\textstyle g_s(t):=\exp\bigl[\,it\sum_{q=1}^{s-1}(a_q\psi^\circ_q\psi^\circ_{q+1})^{-1}\bigr]
\end{equation}
and denote
\begin{equation}
\label{eq:Jnk-def}
\begin{array}{l} J^{[Nn]}(t)\ :=\ (G^{[Nn]}(t))^{-1}J^{[Nn]}G^{[Nn]}(t)\,,\\[2pt]
\text{where}\ G^{[Nn]}(t)\ :=\ \operatorname{diag}\{g_s(t)\}_{s=1,\ldots,Nn}.
\end{array}
\end{equation}
Due to the choice of the multiplicative normalization of the vector~$\psi^\circ$ made above, we have~$g_{s+n}=e^{it}g_s$. Therefore, studying eigenvalues of the matrix $J^{[Nn]}$ corresponding to quasi-periodic (i.e., $\psi_{s+n}=e^{it}\psi_s$) eigenvectors is equivalent to studying eigenvalues of the \emph{family} of self-adjoint matrices $J^{[Nn]}(t)$ corresponding to \emph{periodic} (i.e., $\psi_{s+n}=\psi_s$) eigenvectors, which are nothing but the eigenvalues of the $n\times n$ matrices $J^{[n]}(t)$ with $tN\in 2\pi \mathbb{Z}$.

The question is now reduced to the standard setup of the perturbation theory of (simple) lowest eigenvalues of matrices $J^{[n]}(t)$ as $t\to 0$. It is well known that both these eigenvalues $\lambda(t)\to 0$ and the corresponding, properly normalized, eigenvectors $\psi(t)\to\psi^\circ$ admit asymptotic expansions
\[
\begin{array}{lll}
\lambda(t) &=& \lambda^{(1)}t+\lambda^{(2)}t^2+\ldots,\\[2pt]
\psi(t) &=& \psi^\circ+\psi^{(1)}t+\psi^{(2)}t^2+\ldots
\end{array}\ \ \text{as}\ \ t\to 0,
\]
where $\langle\psi^{(1)};\psi^\circ\rangle =0$ and $\langle\psi^{(2)};\psi^\circ\rangle=-\frac{1}{2}\langle\psi^{(1)};\psi^{(1)}\rangle$. Moreover, it is easy to see from~\eqref{eq:Jnk-def} that~$J^{[n]}(t)=J^{[n]}+J^{[n],(1)}t+J^{[n],(2)}t^2+\ldots$ as $t\to 0$, where $J^{[n],(1)}$ and $J^{[n],(2)}$ are two-diagonal matrices with entries
\[
J^{[n],(1)}_{q+1,q}=-J^{[n],(1)}_{q,q+1}=i\cdot (\psi_q^\circ\psi_{q+1}^\circ)^{-1};\qquad
J^{[n],(2)}_{q+1,q}=J^{[n],(2)}_{q,q+1}=\tfrac{1}{2}a_q^{-1}\cdot (\psi_q^\circ\psi_{q+1}^\circ)^{-2}.
\]
In particular, we have $J^{[n],(1)}\psi^\circ=0$; note that this is exactly where a special choice of the function~\eqref{eq:gsk=} plays a very important role by simplifying the computations.

\smallskip

Considering the linear (in $t$) terms in the identity $J^{[n]}(t)\psi(t)=\lambda(t)\psi(t)$, we see that $J^{[n]}\psi^{(1)}=\lambda^{(1)}\psi^\circ$, which yields $\lambda^{(1)}=0$ and $\psi^{(1)}=0$ since $\langle \psi^{(1)};\psi^\circ\rangle=0$ and $\psi^\circ$ is an eigenvector of $J^{[n]}$ corresponding to the simple eigenvalue $\lambda^\circ=0$. Expanding the same identity up to the second order in $t$ we obtain the equation
\[
J^{[n]}\psi^{(2)}+J^{[n],(2)}\psi^\circ=\lambda^{(2)}\psi^\circ.
\]
Since $\langle J^{[n]}\psi^{(2)},\psi^\circ\rangle=\langle \psi^{(2)},J^{[n]}\psi^\circ\rangle=0$, this allows us to compute
%the required coefficient
\begin{align*}
\lambda^{(2)}\ =\ \frac{\langle J^{[n],(2)}\psi^\circ,\psi^\circ\rangle}{\langle \psi^\circ,\psi^\circ\rangle}\
&=\ \frac{\sum_{q=1}^{n}(a_q\psi^\circ_q\psi^\circ_{q+1})^{-1}}{\sum_{q=1}^{n}(\psi^\circ_q)^2}\\
&=\ \biggl[\,\sum_{q=1}^{n}(\psi^\circ_q)^2\cdot\sum_{q=1}^n(a_q\psi^\circ_q\psi^\circ_{q+1})^{-1}\biggr]^{-1}=\ (nC_J)^{-2},
\end{align*}
where in the third equality we used the prescribed normalization of the vector $\psi^\circ$; note that the formula~\eqref{eq:IDS=} does not depend on the choice of this normalization.

Therefore, we have $\lambda(t)=(nC_J)^{-2}t^2+\ldots$ as $t\to 0$, which means that, for small enough $\lambda_0$ and $N\to\infty$, the $Nn\times Nn$ periodic Jacobi matrix $J^{[Nn]}$ has approximately $NnC_J\cdot \pi^{-1}\sqrt{\lambda_0}$ eigenvalues $\lambda(t)\le\lambda_0$ with $tN\in 2\pi \mathbb{Z}$. In other words, the integrated density of states of $J$ behaves like $C_J\cdot \pi^{-1}\sqrt{\lambda}$ as $\lambda\to 0$.
%, where the constant $C_J$ is given by~\eqref{eq:IDS=}.

\renewcommand{\thesection}{A}

\setcounter{equation}{0}
\section{Appendix. Critical Ising model~$\theta^\mathrm{h}=\theta^\mathrm{v}=\frac{\pi}{4}$: diagonal correlations\\ and the half-plane magnetization via Legendre polynomials} %\label{sc:appendix}

In this appendix we work with the fully homogeneous critical (i.e., $\theta^\mathrm{h}\!=\!\theta^\mathrm{v}\!=\!\frac{\pi}{4}$) Ising model on the $\frac{\pi}{4}$-\emph{rotated} square grid of mesh size~$\sqrt{2}$. (Note that this setup is actually more similar to Section~\ref{sc:layered} rather than to~Section~\ref{sc:homogeneous}.) We begin with a discussion of the famous result of Wu (see Theorem~\ref{thm:wu} below) that provides an explicit expression of the diagonal spin-spin correlations in terms of factorials. Using the same approach as in the core part of our paper, we give a short proof of this theorem by reducing the computation to the norms of the classical \emph{Legendre polynomials}. {This derivation was first published in~\cite[Section~3]{chelkak-ecm2016} based upon an early version of this paper.} (As communicated to the authors by J.H.H.~Perk, a similar link with Legendre functions and Wronskian identities was the starting point of their joint with H.~Au-Yang treatment~\cite{perk-au-yang-Dubna84} of the two-point correlations at criticality via quadratic identities from~\cite{Perk-81dubna}; see also Remark~\ref{rem:quadratic}.) {We reproduce this short proof of Theorem~\ref{thm:wu} below instead of quoting~\cite{chelkak-ecm2016} for two reasons: to keep the presentation self-contained and, more importantly, to emphasize its link with similar explicit formulas for the magnetization $M_m$ in the $(2m)$-th column of {the zig-zag half-plane}~$\mathbb H^\diamond$. To obtain the latter, we use a simple \emph{Schwarz reflection} argument instead of applying our main result, Theorem~\ref{thm:layered}, in the spirit of Section~\ref{subsect:wetting}. This gives a set of exact identities (see Theorem~\ref{thm:Mm-homo-even} and Remark~\ref{rem:Mm-odd}) between $M_m$ and diagonal correlations in the full plane which appear to be new.}

\begin{rem} The interested reader is also referred to~\cite[Section~3]{chelkak-ecm2016} where the non-critical case~$\theta=\theta^\mathrm{h}=\theta^\mathrm{v}<\frac{1}{4}\pi$ is handled in the same way, via the OPUC polynomials corresponding to the weight~$w_q(t)=|1-q^2e^{it}|$ with~$q:=\tan\theta<1$. It would be interesting to understand the precise link between asymptotics of these orthogonal polynomials obtained by Basor, Chen and Haq in~\cite{basor-et-al-2015} and asymptotics of the diagonal Ising correlations obtained by Perk and Au-Yang in~\cite{perk-au-yang-09}.
\end{rem}

Let~$n\in\mathbb N_0$ and assume that the~$\frac{\pi}{4}$-rotated square grid is shifted so that its vertices (resp., centers of faces) form the lattice~$(-n-\frac{1}{2}+k,s)$ (resp., $(n+\frac{1}{2}+k,s)$) with~$k,s\in\mathbb Z$ and~$k+s\in 2\mathbb Z$. Let
\[
D_n:=\mathbb E[\sigma_{(-n+\frac{1}{2},0)}\sigma_{(n+\frac{1}{2},0)}]
\]
be the (infinite-volume limit of the) diagonal spin-spin correlation at distance of~$n$ diagonal steps. Denote~$\mathbf v:=(-n-\frac{1}{2},0)$, $\mathbf u:=(n+\frac{1}{2},0)$ and let
\begin{align*}
V(k,s)\ :=\ %\Psi_{[\mathbf v,\mathbf u]}((k,s))=
X_{[\mathbf v,\mathbf u]}((k,s)),\qquad & k,s\in\mathbb Z,\ k\!+\!s\!+\!n\in 2\mathbb Z,
%V^\circ(k,s):=\Psi_{[\mathbf v,\mathbf u]}(k,s)=iX_{[\mathbf v,\mathbf u]}(k,s),\quad & k,s\in\mathbb Z,\ k\!+\!s\!+\!n\not\in 2\mathbb Z.
\end{align*}
recall that~$V$ is a spinor on the double covers branching over~$\mathbf v$ and~$\mathbf u$. It follows from Proposition~\ref{prop-massive-harmonicity} (or, equivalently, Proposition~\ref{prop-harmonicity-layered}) that $V$ satisfies the standard discrete harmonicity condition~$[\Delta V](k,s)=0$ for all~$k,s$ except at the points~$(\pm n,0)$ near the branchings, where
\begin{align*}
[\Delta V](k,s)\;:=&\;-V(k,s)\\&+\ \tfrac{1}{4}[V(k\!-\!1,s\!-\!1)+V(k\!+\!1,s\!-\!1)+V(k\!-\!1,s\!+\!1)+V(k\!+\!1,s\!+\!1)].
\end{align*}
It directly follows from the definition of the observable~$X_{[\mathbf v,\mathbf u]}$ and the self-duality of the critical model that
\begin{equation}
\label{eq:Vcrit=Dn}
V(-n,0)\ =\ V(n,0)\ =\ D_n\,.
\end{equation}
Moreover, a straightforward computation similar to the proof of Proposition~\ref{prop-massive-harmonicity} implies that
\begin{equation}
\label{eq:Lcrit=Dn+1}
\begin{array}{ll}
[\Delta V](\pm n,0)\  =\ -\tfrac{1}{2}D_{n+1}\quad & \text{if~$n\ge 1$}\,,\\{}
[\Delta V](0,0)\ =\ -D_1 & \text{if~$n=0$}\,.
\end{array}
\end{equation}
Applying the optional stopping theorem as in the proof of Lemma~\ref{lem:uniqueness}, it is easy to see that the \emph{uniformly bounded} discrete harmonic spinor~$V$ is uniquely defined by its values~\eqref{eq:Vcrit=Dn} near the branchings. Following exactly the same route as in Section~\ref{subsect:construction-homogen} we now construct~$V$ explicitly; {a similar idea was used in~\cite[Appendix~A]{GHP-19} to construct the harmonic measure of the tip in the slit plane, which can be viewed as an analogue of the function~$V(k-n,s)$ for~$n=\infty$.}
\begin{lem}
\label{lem:diag-explicit}
Let~$P_n(x):=(2^nn!)^{-1}\frac{d}{dx}[(x^2\!-\!1)^n]$ be the $n$-th Legendre polynomial. Then, for all $k\in\mathbb Z$ and $s\in\mathbb N_0$ such that~$n\!+\!k\!+\!s\in 2\mathbb Z$, one has
\begin{equation}
\label{eq:Vcrit-explicit}
V(k,\pm s)\ =\ \frac{C_n}{2\pi}\int_{-\pi}^{\pi} e^{-ikt}(y(t))^s P_n(\cos t) dt,
\end{equation}
where~$y(t)=(1-|\sin t|)/\cos t$
and~$C_n$ is chosen so that~$V(\pm n,0)=D_n$.
\end{lem}
\begin{proof} It is easy to see that
\begin{itemize}
\item the values~$V(k,s)$ defined by~\eqref{eq:Vcrit-explicit} are uniformly bounded since~$|y(t)|\le 1$;

\item $[\Delta V](k,s)=0$ if~$s\ne 0$ since~$y(t)=\frac{1}{2}\cos t\cdot (1+(y(t))^2)$;

\item $V(k,0)=0$ if~$|k|>n$, thus one can view~\eqref{eq:Vcrit-explicit} as a function (spinor) defined on the \emph{double cover} branching over~$\mathbf v$ and~$\mathbf u$ and vanishing over the real line outside the segment~$[\mathbf v,\mathbf u]$, this spinor satisfies the discrete harmonicity property at~$(k,0)$ with~$|k|>n$ due to symmetry reasons.
\end{itemize}
Moreover, the orthogonality in $L^2([-1,1])$ of~$P_n(x)$ to all monomials~$1,x,\dots,x^{n-1}$ gives
\begin{align*}
-[\Delta V](k,0)\ &=\ V(k,0)-\tfrac{1}{2}[V(k\!-\!1,1)+V(k\!+\!1,1)]\\
&=\ \frac{C_n}{2\pi}\int_{-\pi}^\pi e^{-ikt}(1-y(t)\cos t)P_n(\cos t)dt\\
& =\ \frac{C_n}{2\pi}\int_{-\pi}^\pi\cos(kt)|\sin t|P_n(\cos t)dt\ =\ \frac{C_n}{\pi}\int_{-1}^1 T_{|k|}(x)P_n(x)dx\ =\ 0
\end{align*}
for all~$|k|<n$, where~$T_k(x):=\cos(k\arccos x)$ are the Chebyshev polynomials. Therefore, the Kadanoff--Ceva fermion~$X_{[\mathbf v,\mathbf u]}((k,s))$ must coincide with the right-hand side of~\eqref{eq:Vcrit-explicit} up to a multiplicative constant.
\end{proof}

The following theorem can be obtained as a simple corollary of Lemma~\ref{lem:diag-explicit}.
\begin{theo}[{\bf Wu}]
\label{thm:wu} The following explicit formula is fulfilled:
\begin{equation}
\label{eq:Wu-factorials}
D_{n}\ =\ \biggl(\frac{2}{\pi}\biggr)^{\!\!n}\cdot\,\prod_{k=1}^{n-1}\biggl(1-\frac{1}{4k^2}\biggr)^{\!\!k-n},\quad n\ge 0.
\end{equation}
\end{theo}
\begin{proof} Denote by~$p_n:=(2^nn!)^{-1}\cdot (2n)!/n!$ the leading coefficient of the Legendre polynomial~$P_n$ and let~$t_n:=2^{n-1}$, $n\ge 1$ be the leading coefficient of the Chebyshev polynomial~$T_n$, note that the value~$t_0=1$ does not match the general case. It follows from~\eqref{eq:Vcrit-explicit} that~$D_n=C_n\cdot 2^{-n}p_n$. On the other hand,
\begin{align*}
-[\Delta V](\pm n,0)\; &=\; \frac{C_n}{\pi}\int_{-1}^1 T_n(x)P_n(x)dx\; =\; \frac{C_nt_n}{\pi p_n}\cdot \|P_n\|_{L^2([-1,1])}^2\; =\; \frac{2C_nt_n}{\pi (2n\!+\!1)p_n}\,.
\end{align*}
Due to~\eqref{eq:Lcrit=Dn+1}, we conclude that for \emph{all}~$n\ge 0$ the following recurrence relation holds:
\[
\frac{D_{n+1}}{D_n}\ =\ \frac{2^{n+1}C_n}{\pi(2n\!+\!1)p_n}\ =\ \frac{2^{2n+1}}{\pi(2n\!+\!1)p_n^2}\ =\ \frac{2}{\pi}\cdot \frac{((2n)!!)^2}{(2n\!-\!1)!!(2n\!+\!1)!!}\,.
\]
This easily gives~\eqref{eq:Wu-factorials} by induction.
\end{proof}

We now move on to an explicit expression for the magnetization in the $(2m)$-th column of the zig-zag half-plane~$\mathbb H^\diamond$ with `$+$' boundary conditions:
\[
M_m\ :=\ \mathbb E^{+}_{\mathbb H^\diamond}[\sigma_{(-2m-\frac{1}{2},0)}]\,.
\]

\begin{theo}
\label{thm:Mm-homo-even}
The following identities are fulfilled for all~$m\in\mathbb N_0$:
\begin{equation}
\label{eq:M/M=D/D-even}
\frac{M_{m+1}}{M_m}\ =\ \frac{D_{2m+2}}{D_{2m+1}}\,,\qquad
M_{m}\ =\ \biggl(\frac{2}{\pi}\biggr)^{\!\!m}\cdot\,\prod_{k=1}^{2m-1}\biggl(1-\frac{1}{4k^2}\biggr)^{\!\!\lfloor\frac{k}{2}\rfloor-m}.
\end{equation}
\end{theo}
\begin{proof} Similarly to Section~\ref{subsect:observable-layered}, let~$\mathbf v=(-2m\!-\!\tfrac{3}{2},0)$ and
\[
H(-k,s)\ :=\ X_{[\mathbf v]}((-k,s)),\qquad k\in\mathbb N_0,\ s\in\mathbb Z,\ k\!+\!s\not\in 2\mathbb Z
\]
be the half-plane fermionic observable. This is a bounded discrete harmonic (except at~$(-2m\!-\!1,0)$) spinor on the double cover of~$\mathbb H^\diamond$ branching over~$\mathbf v$ which satisfy the boundary conditions
\[%\begin{equation} \label{eq:CR-boudary-crit}
H(0,s)\ =\ 2^{-1/2}\cdot [H(-1,s\!-\!1)+H(-1,s\!+\!1)],\quad s\not\in 2\mathbb Z,
\]%\end{equation}
on the imaginary line (see~\eqref{eq:CR-boundary} and~\eqref{eq:H-Crelations}). Denote
\begin{equation}
\label{eq:V-as-H-crit}
\begin{array}{ll}
V(\pm k,s)\ :=\ CH(k,s) & \text{if~$k\in \mathbb N$},\\
V(0,s):=2^{-1/2}\cdot CH(0,s)\quad &\text{if~$k=0$},
\end{array}
\quad s\in\mathbb Z,\ \ k\!+\!s\not\in 2\mathbb Z,
\end{equation}
where~$C:=D_{2m+1}/M_m$; {up to a change of the multiplicative normalization, this is nothing but the extension of~$H$ from the left half-plane to the full plane via the discrete Schwartz reflection.} By construction,~$V$ is a spinor on the double cover of the full-plane branching over~$\mathbf v$ and~$\mathbf u:=(2m\!+\!\tfrac{3}{2},0)$ which is discrete harmonic everywhere (including points on the imaginary line) except at points~$(\pm(2m\!+\!1),0)$ near the branchings, where one has~$V(\pm(2m\!+\!1),0)=D_{2m+1}$. Therefore, it coincides with the full-plane observable~$X_{[\mathbf v,\mathbf u]}((k,s))$ discussed above. In particular, \eqref{eq:V-as-H-crit} implies the identity
\[
\tfrac{1}{2}D_{2m+2}\ =\ -[\Delta V](-2m\!-\!1,0)\ =\ -C\cdot [\Delta H](-2m\!-\!1,0)\ =\ C\cdot \tfrac{1}{2}M_{m+1}
\]
which is equivalent to the first identity in~\eqref{eq:M/M=D/D-even}. The explicit formula for~$M_m$ easily follows from the explicit formula~\eqref{eq:Wu-factorials} by induction.
\end{proof}

\begin{rem} \label{rem:Mm-odd}
Similarly, let~$M_{m-\frac{1}{2}}$ %:=\mathbb E^+_{\mathbb H^\diamond}[\sigma_{(-2m+\frac{1}{2},\pm 1)}]$
denote the magnetization in the \mbox{$(2m\!-\!1)$-th} column of the critical homogeneous Ising model in the zig-zag plane. It is not hard to repeat the proof of Theorem~\ref{thm:Mm-homo-even} in this situation and to obtain the identity
\[
M_{m+\frac{1}{2}}\,/\,M_{m-\frac{1}{2}}\ =\ D_{2m+1}\,/\,D_{2m},\quad m\in\mathbb N_0,
\]
where we formally set~$M_{-\frac{1}{2}}:=\sqrt{2}$, this convention is the result of the additional factor relating the values of the half-plane and the full-plane observables on the imaginary line via~\eqref{eq:V-as-H-crit}. By induction, one easily gets the identity
\begin{equation}
\label{eq:MM=D-crit}
M_{m+\frac{1}{2}}M_{m\vphantom{\frac{1}{2}}}\ =\ \sqrt{2}\cdot D_{2m+1},\qquad m\in \tfrac{1}{2}\mathbb N_0,
\end{equation}
and an explicit formula for~$M_{m+\frac{1}{2}}$, which is similar to~\eqref{eq:M/M=D/D-even}. Finally, a straightforward analysis gives the asymptotics
\begin{equation}
\label{eq:diag-crit-asymptotics}
D_n\ \sim\ \mathcal{C}_\sigma^2\cdot (2n)^{-1/4},\quad M_m\ \sim\ 2^{1/8}\mathcal{C}_\sigma\cdot (2m)^{-1/8},\quad n,m\to\infty,
\end{equation}
where~$\mathcal{C}_\sigma=2^{\frac{1}{6}}e^{\frac{3}{2}\zeta'(-1)}$ is the same universal constant as in Theorem~\ref{thm:crit-homogen}. Note that we prefer to encapsulate the factors~$2n$ and~$2m$ (rather than simply~$n$ and~$m$), respectively, as they are equal to the geometric distance between the two spins under consideration and the distance from the spin~$\sigma_{(-2m-\frac{1}{2},0)}$ to the boundary of the half-plane~$\mathbb H^\diamond$, respectively.
\end{rem}

%%%%%%%%%%%%%%%%%%%%%%%%%%
%\bibliographystyle{plain}
%\bibliography{LayeredIsing}
%\end{document}
%%%%%%%%%%%%%%%%%%%%%%%%%%

\end{document}